\documentclass[a4paper,onecolumn,11pt,accepted=2026-08-17]{quantumarticle}
\pdfoutput=1
\usepackage[utf8]{inputenc}
\usepackage{lipsum}

\usepackage{amsmath, amssymb,amsfonts,amsthm,nicefrac}
\usepackage{mathtools}
\mathtoolsset{
  mathic=true,
  centercolon=true,  
}
\usepackage{thmtools,thm-restate}
\usepackage{booktabs}
\usepackage{xspace,graphicx,relsize,bm,xcolor}
\usepackage{soul} 

\usepackage[labelfont=bf]{caption}
\usepackage{parskip}  

\usepackage[LGR,T1]{fontenc} 
\usepackage{mlmodern}
\usepackage{dsfont}  
\usepackage[cal=cm,scr=rsfs,bb=ams]{mathalpha} 
\usepackage[italic,selfgreek]{mathastext}
\MTboldvariant{sb}
\Mathastext

\usepackage{xurl}
\usepackage[
    backend=biber,
    style=alphabetic,
    sorting=ynt,
    backref=true,
    maxbibnames=99
    ]{biblatex}
\DefineBibliographyStrings{english}{%
  backrefpage = {page},
  backrefpages = {pages},
  mathesis = {MSc thesis},
}

\usepackage{hyperref}  
\usepackage[margin=1in]{geometry}
\usepackage[singlespacing]{setspace}
\definecolor{linkcol}{rgb}{0.0,0.5,0.65}  
\definecolor{citecol}{rgb}{0.0, 0.6, 0.45}
\definecolor{urlcol}{rgb}{0.7, 0.0, 0.55}
\hypersetup{
	colorlinks,
	linkcolor={linkcol},
	citecolor={citecol},
	urlcolor={urlcol},
  breaklinks=true
}

\usepackage{subcaption}
\usepackage{mleftright}  
\mleftright  
\usepackage{tipa}
\usepackage{multirow}
\usepackage{authblk}  

\usepackage{subfiles} 

\usepackage{algorithm}
\usepackage{algpseudocode}

\usepackage{multirow}

\usepackage{pifont}

\usepackage{float}

\usepackage{chngpage}

\usepackage[capitalize,noabbrev]{cleveref}  

\usepackage{enumerate}
\usepackage[shortlabels]{enumitem}

\usepackage{pdflscape}

\def\01{\{0,1\}}

\DeclareMathOperator*{\Ex}{\mathbf{E}}
\let\Pr\relax
\DeclareMathOperator*{\Pr}{\mathbf{Pr}}

\renewcommand{\epsilon}{\varepsilon}

\DeclarePairedDelimiterX\ceil[1]{\lceil}{\rceil}{
    \IfBlankTF{#1}{\thinspace\cdot\thinspace}{#1}
}
\DeclarePairedDelimiterX\floor[1]{\lfloor}{\rfloor}{
    \IfBlankTF{#1}{\thinspace\cdot\thinspace}{#1}
}
\DeclarePairedDelimiterX\norm[1]{\lVert}{\rVert}{
    \IfBlankTF{#1}{\thinspace\cdot\thinspace}{#1}
}

\DeclarePairedDelimiterX\abs[1]{\lvert}{\rvert}{
    \IfBlankTF{#1}{\thinspace\cdot\thinspace}{#1}
}
\DeclarePairedDelimiterX\ket[1]{\lvert}{\rangle}{
    \IfBlankTF{#1}{\psi}{#1}
}

\DeclarePairedDelimiterX\bra[1]{\langle}{\rvert}{
    \IfBlankTF{#1}{\psi}{#1}
}

\DeclarePairedDelimiterX\expval[1]{\langle}{\rangle}{
    \IfBlankTF{#1}{\thinspace\cdot\thinspace}{#1}
}
\DeclarePairedDelimiterX\braket[2]{\langle}{\rangle}{
    #1\,\delimsize\vert\,\mathopen{}#2
}

\DeclarePairedDelimiterX\mel[3]{\langle}{\rangle}{
    #1\delimsize\rvert\,\mathopen{}#2\,\delimsize\lvert\mathopen{}#3
}

\DeclarePairedDelimiterX\ketbra[2]{\lvert}{\rvert}{
    #1\delimsize\rangle\negthinspace\delimsize\langle\mathopen{}#2
}
\DeclarePairedDelimiterX\proj[1]{\lvert}{\rvert}{
    \IfBlankTF{#1}{
        \psi\delimsize\rangle\negthinspace\delimsize\langle\mathopen{}\psi
    }{
        #1\delimsize\rangle\negthinspace\delimsize\langle\mathopen{}#1
    }        
}

\DeclarePairedDelimiterX\Set[1]\{\}{%

#1
}

\newcommand{\tr}[1]{\operatorname{tr}\left[{#1}\right]}

\NewDocumentCommand{\RefRestate}{m}{%
    \nameCref{#1} \hyperlink{#1-restate}{\ref*{#1}}%
}

\NewDocumentCommand{\labelRestate}{m}{
    \hypertarget{#1-restate}{}
}
\newtheoremstyle{mydefinitionsty}
{10pt}
{10pt}
{}
{}
{}
{}
{.5em}
{\textbf{\thmname{#1}~\thmnumber{#2}:  }\thmnote{(#3)}}
\theoremstyle{mydefinitionsty}
\declaretheorem[name=Definition,style=mydefinitionsty]{definition}
\declaretheorem[name=Remark,style=mydefinitionsty]{remark}

\declaretheorem[name=Observation,style=mydefinitionsty]{observation}

\declaretheorem[name=Example,style=mydefinitionsty]{example}

\crefname{fact}{fact}{facts}
\Crefname{fact}{Fact}{Facts}
\crefname{observation}{observation}{observations}
\Crefname{observation}{Observation}{Observations}

\newtheoremstyle{mythmsty}
{10pt}
{10pt}
{\itshape}
{}
{}
{}
{.5em}
{\textbf{\thmname{#1}~\thmnumber{#2}:  }\thmnote{(#3)}}
\theoremstyle{mythmsty}

\declaretheorem[name=Theorem,style=mythmsty]{theorem}
\declaretheorem[name=Lemma,style=mythmsty]{lemma}
\declaretheorem[name=Problem,style=mythmsty]{problem}
\declaretheorem[name=Corollary,style=mythmsty]{corollary}

\RenewDocumentEnvironment{proof}{s o}{%
    \textbf{Proof\IfBooleanT{#1}{ Sketch}\IfValueT{#2}{ of #2}:}%
}{%
    \hfill\(\square\)%
}

\NewDocumentEnvironment{proofsketch}{o}{%
    \textbf{Proof Sketch\IfBlankF{#1}{ (of #1):}}%
}{%
    \hfill\(\square\)%
}

\definecolor{marcelorange}{rgb}{0.8, 0.43, 0}

\definecolor{alexcolor}{rgb}{0.0, 0.47, 0.75}

\newcommand{\stkout}[1]{\ifmmode\text{\sout{\ensuremath{#1}}}\else\sout{#1}\fi}
\newif\ifverbose
\newcommand{\ins}[1]{\ifverbose\textcolor{blue}{#1}\else#1\fi}

\title{Interactive proofs for verifying (quantum) learning and testing}

\author[1,2]{Matthias C.~Caro}
\orcid{0000-0001-9009-2372}
\email{matthias.caro@warwick.ac.uk}
\author[2,3]{Jens Eisert}
\email{jense@zedat.fu-berlin.de}
\author[2]{Marcel Hinsche}
\email{m.hinsche@fu-berlin.de}
\author[2]{Marios Ioannou}
\email{marios.ioannou@fu-berlin.de}
\author[2]{Alexander Nietner}
\email{a.nietner@fu-berlin.de}
\author[4,5,6,7]{Ryan Sweke}
\email{rsweke@aims.ac.za}
\affil[1]{\small Department of Computer Science, University of Warwick, UK}
\affil[2]{\small Dahlem Center for Complex Quantum Systems, Freie Universit\"at Berlin, Berlin, Germany}
\affil[3]{\small Helmholtz-Zentrum Berlin für Materialien und Energie, Berlin, Germany}
\affil[4]{\small IBM Quantum, Almaden Research Center, San Jose, CA, USA}
\affil[5]{African Institute for Mathematical Sciences (AIMS), South Africa}
\affil[6]{Department of Mathematical Sciences, Stellenbosch University, Stellenbosch 7600, South Africa}
\affil[7]{National Institute for Theoretical and Computational Sciences (NITheCS), South Africa}

\date{}

\begin{document}
\maketitle
\thispagestyle{empty}

\begin{abstract}
\noindent We consider the problem of testing and learning from data in the presence of resource constraints, such as limited memory or weak data access, which place limitations on the efficiency and feasibility of testing or learning. In particular, we ask the following question: Could a resource-constrained learner/tester use interaction with a resource-unconstrained but untrusted party to solve a learning or testing problem more efficiently than they could without such an interaction?  In this work, we answer this question both abstractly and for concrete problems, in two complementary ways:
\begin{enumerate}
    \item For a wide variety of scenarios, we prove that a resource-constrained learner \textit{cannot} gain any advantage through classical interaction with an untrusted prover. As a special case, we show that for the vast majority of testing and learning problems in which \textit{quantum memory} is a meaningful resource, a memory-constrained quantum algorithm cannot overcome its limitations via classical communication with a memory-unconstrained quantum prover.
    \item In contrast, when quantum communication is allowed, we construct a variety of interactive proof protocols, for specific learning and testing problems, which allow memory-constrained quantum verifiers to gain significant advantages through delegation to untrusted provers.
\end{enumerate} 
These results highlight both the limitations and potential of delegating learning and testing problems to resource-rich but untrusted third parties.
\end{abstract}

\newpage
\thispagestyle{empty}
\tableofcontents
\newpage

\clearpage
\pagenumbering{arabic}
\section{Introduction}\label{s:introduction}

The success of modern \emph{machine learning} (ML) models so far often relies on extraordinary amounts of high-quality data and computation time.
Indeed, current \emph{large language models} (LLMs) have hundreds of billions of parameters and are trained on tens or even hundreds of terabytes of data~\cite{hoffmann2022training}.
The training process itself is executed on thousands of GPUs over several weeks or even months, with a total cost of many millions of dollars.
Additionally, \emph{reinforcement learning from human feedback} (RLHF) \cite{christiano2017rlhf}, employed to further improve the quality of answers provided by an LLM as seen by users, requires collecting a specific kind of data that comes at the cost of many human working hours.
Finally, developing competitive ML models also requires advanced programming expertise.
Thus, we see that one of the arguably most notable current-day ML applications requires both amounts and kinds of resources that are unavailable to all but a select few entities.

Beyond this practical relevance, the impact of resource mismatches -- either qualitative or quantitative -- between what is available to a learner and what is required to solve a learning problem of interest has long been appreciated in learning theory. 
For instance, in classical learning theory, the seminal work \cite{kearns1998efficient} demonstrated that learning $n$-bit parity functions requires exponentially-in-$n$ many \emph{statistical queries} (SQs) but only linearly-in-$n$ many random examples. 
More recently, Ref.\  \cite{arunachalam2023role} showed an analogous 
separation between quantum SQs (QSQs) and quantum examples for learning quadratic functions, 
and Refs.\ \cite{arunachalam2023role,nietner2023unifying} proved a separation between learning stabilizer states from QSQs and (incoherent) copies of the state.
In a different vein, juxtaposing classical and quantum learning theory, we know that Fourier-based learning of, e.g., disjunctive normal forms (DNFs) from quantum examples, or even QSQs, improves upon learning from classical examples in terms of sample complexity \cite{bshouty1998learning, arunachalam2020qsq} and likely also in terms of computation time \cite{grilo2019learning}.
Finally, recent developments in quantum learning and testing have highlighted and experimentally demonstrated exponential sample and query complexity advantages of coherent multi-copy learners over incoherent single-copy procedures for a variety of tasks \cite{chen2022exponential, huang2022quantum-advantage}, showing that \textit{quantum memory} is a valuable 
resource for many quantum testing and learning problems.

Thus, both from a practical and from a theoretical perspective, successful machine learning can become an unattainable goal for resource-constrained learners that lack access to, e.g., the required type of data or the required processing capabilities. 
This motivates an approach in which the resource-constrained party delegates part of the learning to a resource-unconstrained (or at least less resource-constrained) but potentially untrusted service provider.
Modeling this approach with an interactive proof framework, the recent line of work \cite{goldwasser2021interactive, oconnor2021delegating, mutreja2022pac-verification, caro2023verification, gur2024power} has explored the potential and limitations of interactively verifying delegated learning for different scenarios with specific resource mismatches. In this work, we take a broader perspective and ask:
\begin{center}
    \textit{When can a resource-constrained verifier profit from interacting with a resource-unconstrained but untrusted prover in solving problems of learning from data?}
\end{center}
Here, the verifier profiting from the interaction is understood as achieving improvements in sample, query, and/or time complexity over what the verifier would require on their own. 
To illustrate concrete instantiations of this abstract question, let us return to the resource mismatches highlighted above. 
In the spirit of Ref.\ \cite{kearns1998efficient}, we may ask: Can a verifier with only SQ-access interact with a prover, who has random example access, to solve problems related to parity learning with, say, polynomially-in-$n$ many SQs?
Similarly, in the context of Refs.\ \cite{bshouty1998learning, arunachalam2020qsq}: Can a classical verifier with random example access interact with a quantum prover, who has quantum example access, and thereby efficiently learn DNFs?
Finally, motivated by Refs.\ \cite{chen2022exponential, huang2022quantum-advantage}, an example question becomes: Can a memory-constrained quantum verifier with access to copies of an unknown state interact with a memory-unconstrained prover, who has access to copies of the same state, and can make coherent measurements on multiple copies, to efficiently estimate the purity of the state?

We answer these questions both in the abstract and for concrete examples. 
On the one hand, we show that, for a wide variety of scenarios, resource-constrained learning or testing algorithms \textit{cannot} gain any query complexity advantages through interactions with a resource-unconstrained but untrusted prover over what they could achieve 
on their own. This no-go result in particular implies that for any of the myriad tasks in which quantum memory is known to be a meaningful resource~\cite{chen2022exponential, huang2022quantum-advantage}, limitations arising from a constrained quantum memory cannot be overcome via delegation \textit{over classical communication channels} to unconstrained but untrusted service providers.
On the other hand, we study memory-constrained quantum verifiers interacting with memory-unconstrained quantum provers via a \textit{quantum communication} channel.
In this setting, we overcome the above no-go result and construct a variety of concrete interactive proof protocols for learning and testing problems that do indeed allow the constrained parties to gain advantages via interaction with unconstrained service providers. 
Our results are presented in more detail in \Cref{ss:results_overview}.

\subsection{Framework}\label{ss:framework}
Throughout this work, we consider a wide variety of learning and testing problems. To solve such problems, an algorithm should output a valid solution when given oracle access to some unknown instance of the problem. We are particularly interested in problems for which the query complexity of solving the problem depends on the details of either the oracle, or the computational resources (such as memory) that the algorithm has access to. When it comes to testing, we mainly focus on many-vs-one distinguishing problems. \ins{This is well aligned with several recent works, notably Refs.~\cite{chen2022exponential, huang2022quantum-advantage, arunachalam2023role}, that exhibited exponential quantum advantages for such tasks.}

\begin{definition}[Many-vs-one distinguishing problem]\label{definition:many-vs-one-distinguishing-informal} 
    A \emph{many-vs-one distinguishing problem} is defined by a triple $(x_A,\mathcal{X}_R,\mathcal{X})$, where $\mathcal{X}$ is some instance set, $\mathcal{X}_R\subset \mathcal{X}$ is a reject set, and $x_A\in\mathcal{X}\setminus\mathcal{X}_R$ is an accept instance. Valid solutions for the triple are defined as follows:
    \begin{itemize}
        \item[(i)] When given oracle access to $x_A$, the valid solution is ``accept".
        \item[(ii)] When given oracles access to $x\in\mathcal{X}_R$, the valid solution is ``reject".
        \item[(iii)] When given oracle access to $x\in\mathcal{X}\setminus (\mathcal{X}_R\cup\{x_A\})$, then both ``accept" and ``reject" are valid solutions.
    \end{itemize}
\end{definition}

Note that we consider a many-vs-one distinguishing problem to be a \textit{testing problem}, as one can interpret the problem as requiring one to test whether an unknown instance $x\in\mathcal{X}$ is the accept instance $x_A$ or a reject instance $x\in\mathcal{X}_R$, promised that either one of these is the case.
When it comes to learning we will consider a variety of different types of learning problems.
However, as explained in Section~\ref{ss:learning_frameworks}, all of the problems we consider can be understood as special instances of an agnostic learning problem.

\vspace*{-0.25cm}
\begin{restatable}[Agnostic learning problem]{definition}{defagnostic}
    \label{def:agnostic-learning}%
    An \textit{agnostic learning problem} is defined by a quintuple ${(\mathcal{X},\mathcal{M},\ell, \epsilon,\alpha)}$. 
    Here, $\mathcal{X}$ and \(\mathcal{M}\) are sets; $\mathcal{X}$ is the instance set and \(\mathcal{M}\) is referred to as the model class. Additionally, $\ell:\mathcal{M}\times \mathcal{X}\to \mathbb{R}_{\geq 0}$ is called the loss function,  and $\epsilon\in (0,1]$ and $\alpha\geq 0$ are both accuracy parameters. When given 
    oracle access to some instance $x\in\mathcal{X}$, any model $m\in\mathcal{M}$ satisfying
    \begin{equation}\label{eq:agnostic_criteria}
            \ell (m,x) 
            \leq 
            \alpha 
            \min_{m'\in\mathcal{M}}\left[\ell(m',x)\right]
            + \epsilon
    \end{equation}
    is a valid solution.
\end{restatable}

At a high level, for any $x\in\mathcal{X}$, a model $m\in\mathcal{M}$ is a valid solution if it is ``close" to the best model in $\mathcal{M}$ for $x$, with respect to the loss function $\ell$ and accuracy parameters $\alpha$ and $\epsilon$. We note that one most often considers agnostic learning problems with $\alpha=1$. However, constructing efficient algorithms that solve such agnostic learning problems is often formidably difficult, and so in this work we also consider agnostic learning problems with $\alpha > 1$.

With these definitions in hand, we can now define what it means for an algorithm to \textit{solve} either of the above problems with respect to a specific oracle model. Intuitively, the oracle \ins{-- called $\mathsf{O}$ in the following definition -- }plays the role of abstracting a certain form of data access of the algorithm.

\begin{definition}[Solving a learning/testing problem]\label{def:solve_a_problem}
    We say that an algorithm $\mathcal{A}$ solves a learning/testing task from $\mathsf{O}$-access, with probability $\geq 1-\delta$, if for all admissible instances $x\in\mathcal{X}$, when $\mathcal{A}$ is given access to $\mathsf{O}(x)$, it outputs a valid solution with probability $\geq 1-\delta$. Here, the probability is with respect to potential randomness in both the oracle and the algorithm. 
    (If \(\mathcal{A}\) and \(\mathsf{O}\) are deterministic, then \(\delta=0\).)
\end{definition}

In addition, we often in this work consider the complexity of algorithms which solve a \textit{randomized} many-vs-one distinguishing problem, whose definition we postpone to Section~\ref{ss:learning_frameworks}.

With these notions of testing and learning established, we now consider the following situation: We imagine a verifier $V$ who wants to solve a testing or learning problem, given access to the unknown object via some oracle $\mathsf{O}_V$. Crucially, we imagine that this verifier is \textit{resource-constrained} in some way, which leads to a lower bound on the number of queries or computational time required by $V$ to solve the problem. This resource constraint could come in different forms. We could, for example, consider a verifier that has limited memory, access to a ``weak oracle'' (e.g., random examples as opposed to membership queries), or, in the quantum case, the inability to make coherent measurements across multiple copies of an unknown state. We then imagine that $V$ has the ability to interact with a resource-unconstrained but untrusted prover $P$, who has access to some oracle $\mathsf{O}_P$. In particular, we consider scenarios in which, given that $P$ is resource-unconstrained, it could solve the problem with fewer queries or less time than required by $V$ (given the resource constraint). The question we ask in this work is whether there exists a well-defined procedure via which $V$ can meaningfully solve the problem through interaction with $P$, while simultaneously surpassing the lower bound on complexity imposed by the resource constraint. Following Refs.~\cite{goldwasser2021interactive, oconnor2021delegating, mutreja2022pac-verification, caro2023verification, gur2024power} we refer to such a procedure as an \textit{interactive proof (IP)} 
for testing or learning.

\begin{figure}
    \centering
     \includegraphics[width=0.65\textwidth]{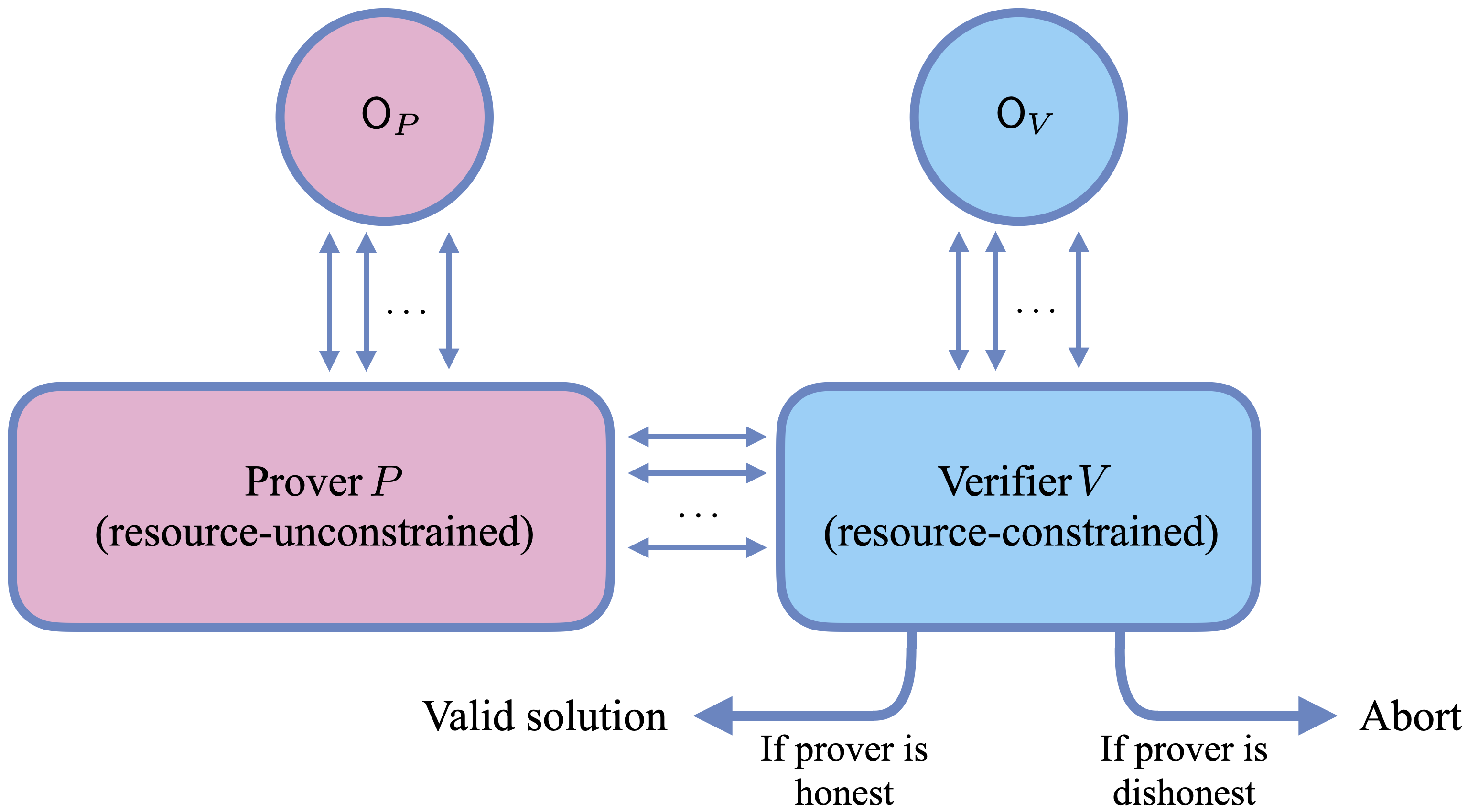}
    \caption[Interactive Learning Protocol]{\textbf{Illustration of interactive testing/learning.} The verifier $V$ interacts with the prover $P$ to solve a testing or learning task. The verifier has access to the oracle \(\mathsf{O}_V(x)\) and the prover has access to the oracle \(\mathsf{O}_P(x)\). The verifier is resource-constrained, while the prover is resource-unconstrained.} 
    \label{fig:interactive_learning_protocol} 
\end{figure}

\begin{definition}[Interactive proofs for testing and learning]\label{definition:interactive-proofs-for-learning-informal}
    A pair $(V,P)$ of -- possibly randomized or quantum -- algorithms constitutes an \emph{interactive proof system for a testing/learning task}, with success probability $1-\delta$, from $(\mathsf{O}_V, \mathsf{O}_P)$-access if the following are satisfied:
    \begin{itemize}
        \item \textbf{Completeness:} For any $x\in\mathcal{X}$, when $V$ is given access to the oracle $\mathsf{O}_V(x)$ and interacts with the honest prover $P$, who is given access to the oracle $\mathsf{O}_P(x)$, then: With probability $\geq 1-\delta$, $V$ does not abort the interaction and outputs a valid solution to the testing/learning task.
        \item \textbf{Soundness:} For any $x\in\mathcal{X}$ and for any (potentially adversarial) prover $P'$, when $V$ is given access to the oracle $\mathsf{O}_V(x)$ and interacts with $P'$, then: With probability $< \delta$, $V$ does not abort the interaction and outputs an invalid solution to the testing/learning task.
    \end{itemize}
    Additionally, if $\mathsf{A}$ is a class of resource-constrained algorithms, we say that a testing/learning task admits an interactive proof system from $(\mathsf{O}_V, \mathsf{O}_P)$-access with $\mathsf{A}$-constrained verifier if $V$ can be chosen to lie in $\mathsf{A}$.
\end{definition}

\Cref{definition:interactive-proofs-for-learning-informal} is illustrated in \Cref{fig:interactive_learning_protocol}. 
As discussed before \Cref{definition:interactive-proofs-for-learning-informal}, we are naturally also interested in questions of efficiency. In particular we are primarily interested in the query/time complexity of $V$, through the course of the interaction with $P$, and how this compares to the query/time complexity required by $V$ when solving the problem in isolation. \ins{More concretely, how do the number of oracle queries and the overall amount of (classical or quantum) computation that $V$ executes during the interaction with $P$ compare to the number of oracle queries and the compute that $V$ would use to solve the problem on their own?} Depending on the context, we may also care about the communication complexity of the protocol -- i.e.,  the number of (quantum) bits exchanged between $V$ and $P$ -- as well as the query/time complexity of the honest prover $P$. 

The notions of completeness and soundness in \Cref{definition:interactive-proofs-for-learning-informal} become two-sided for the case of a many-vs-one distinguishing task. For example, in our soundness requirement, no matter whether $x=x_A$ or $x\in\mathcal{X}_R$, no dishonest prover $P'$ should be able to convince the verifier $V$ to mistakenly reject or accept, respectively (except with small probability). 
This two-sidedness arises naturally when trying to unify interactive proofs for testing and learning in a single definition.
From a complexity-theoretic perspective, however, one-sided notions of completeness and soundness -- e.g., for soundness, we would only require that no dishonest prover $P'$ can convince the verifier $V$ to mistakenly accept an $x\in\mathcal{X}_R$ -- may seem more natural. 
We recall that the classical complexity class $\mathsf{IP}$ is commonly defined with one-sided completeness and soundness, but it in fact equals its  two-sided version because $\mathsf{IP}=\mathsf{PSPACE}$ \cite{lund1992algebraic, shamir1992ip-equals-pspace} and $\mathsf{PSPACE}=\mathsf{coPSPACE}$, so $\mathsf{IP}=\mathsf{coIP}=\mathsf{IP}\cap \mathsf{coIP}$.
However, this equality in general does not hold relative to an oracle \cite{fortnow1988there}.
Hence, when defining interactive verification of testing, it is not immediately clear whether one-sided completeness and soundness should be favored over their two-sided versions.
For our results, however, the distinction does not matter: The proof of our no-go result (see \Cref{theorem:no-go-interactive-many-vs-one-distinguishing-intro-version}) works even if only one-sided completeness and soundness are assumed. And the interactive protocols for our positive results achieve the stronger two-sided guarantees.  

\subsection{Motivating examples}\label{ss:motivating_examples}

Having established the framework, we provide here a variety of concrete motivating examples to help illustrate both the scope of problems we consider, as well as the different types of resource constraints one could consider, and their effects on query complexity. For a more comprehensive list of problems that we consider, we refer to Appendix.~\ref{app:problems}, and for a longer list of examples of resource constraints and their effects on complexity, we refer to Appendix~\ref{app:motivating_examples}.

\begin{example}[Uniformity testing]\label{example:uniformity_testing} 
    As defined in Problem~\ref{prob:uniformity}, uniformity testing is the many-vs-one distinguishing task which asks one to decide whether an unknown distribution is the uniform distribution or some distribution $\epsilon$-far from the uniform distribution with respect to total variation distance. For this problem it is known that \textit{memory} is a relevant resource. More specifically, for distributions with support size $k$, without any memory constraint $\Theta(\nicefrac{\sqrt{k}}{\epsilon^2})$ samples are necessary and sufficient~\cite{paninski2008coincidence}. However, when only $m$ bits of memory are available, $\Omega(\nicefrac{k}{(m\epsilon^2)})$ samples are necessary~\cite{diakonikolas2019communication,berg2022memory,canonne2023simpler}. As such, for this problem, we can ask whether there exists an interactive proof protocol which allows a memory-constrained verifier, through interaction with an unconstrained prover, to outperform the memory-imposed lower bound.
\end{example}

\begin{example}[Purity testing]\label{example:purity-testing-informal} 
    As defined in Problem~\ref{prob:purity_testing}, purity testing is the problem which asks one to decide whether an unknown state is the maximally mixed state or a pure state.  For the randomized version of purity testing (see Definition~\ref{def:solve_random_problem}), with $\mu$ the Haar measure over pure states, it is known that both \textit{quantum memory}, and \textit{copy access} to the unknown state (i.e. access to copies of the unknown quantum state, as opposed to quantum statistical queries for example), are relevant resources. With respect to memory, if one has enough quantum memory to perform coherent measurements on two copies of the unknown state, then $O(1)$ copies are sufficient overall. However, with single-copy incoherent measurements, $\Omega(\sqrt{d} )$ copies of the unknown state are required~\cite{chen2022exponential, chen2024optimalpauli}. We can therefore ask whether a memory-constrained single-copy verifier can outperform the lower bound imposed by single-copy incoherent measurements via an interactive proof with an unconstrained prover, who has the ability and sufficient quantum memory to make multi-copy coherent measurements. With respect to copy access, when given only QSQ access to the unknown state, then $2^{\Omega(\tau^2 d)}$ QSQs of tolerance $\tau$ are necessary~\cite{arunachalam2023role}. As such, one could also ask whether a verifier with QSQ access could meaningfully delegate purity testing to an untrusted verifier with access to copies of the unknown quantum state.
\end{example}

\begin{example}[Quantum state tomography]\label{example:quantum-state-tomography} 
    As defined in Problem~\ref{prob:qst}, quantum state tomography is the problem in which, when given access to an  unknown quantum state, one should produce a classically described density matrix which is $\epsilon$-close to the unknown state with respect to the trace distance. 
    Optimal sample complexities for this problem are by now well understood. When restricted to single-copy measurements (but still making use of arbitrary and not only Pauli measurements), $\Theta(d^3/\varepsilon^2)$ copies are necessary and sufficient \cite{kueng2017low, chen2023whendoesadaptivity}. For general multi-copy measurements, the optimal sample complexity is $\Theta(d^2/\varepsilon^2)$ \cite{haah2016sample, odonnell2016efficient}, an improvement over the dimension-dependence of the single-copy case. (We note that this separation vanishes when the unknown state is promised to be pure, compare, e.g., Ref.~\cite{zhao2023learning}.)
    More recently, 
    Ref.\ \cite{chen2024optimalstate} demonstrated a smooth dependence of the sample complexity on the number of copies that can be measured simultaneously. 
    We see that \textit{quantum memory} is again a relevant resource, and we can ask about the existence of interactive proofs that allow single-copy verifiers to delegate quantum state tomography to untrusted multi-copy provers.
\end{example}

\begin{example}[Pauli shadow tomography]\label{example:Pauli-shadow-tomography}
    As defined in Problem~\ref{problem:pauli_shadow_tomography}, Pauli shadow tomography is the problem in which, when given access to an unknown quantum state, one should simultaneously predict the expectation values of all $4^n$ Pauli observables.
    For this problem, \emph{quantum memory}, \emph{copy access}, and \emph{noiseless quantum processing} are all known to be relevant resources.
    \begin{itemize}
        \item With respect to memory, if one has enough quantum memory to perform coherent measurements on two copies of the unknown state, then $\mathcal{O}(n)$ copies suffice \cite{huang2021information, king2024triply, chen2024optimalpauli}.
        However, with single-copy incoherent measurements, $\Omega(2^n)$ copies are required \cite{chen2022exponential, chen2024optimalpauli}.
        Thus, as in \Cref{example:purity-testing-informal}, we may ask whether a memory-constrained quantum verifier can overcome the exponential single-copy lower bound by interacting with an untrusted but memory-unconstrained quantum prover.
        \item With respect to copy access, when given only QSQ access to the unknown state, then $\Omega(\tau^2 2^{2n})$ QSQs of tolerance $\tau$ are necessary \cite{arunachalam2023role}.
        We may thus again ask whether a QSQ verifier can improve upon this exponential query complexity lower bound when delegating Pauli shadow tomography to a prover with quantum copy access.
        \item Finally, with respect to noiseless quantum processing, the authors of Ref.\ \cite{chen2023complexity} have shown that \emph{noisy intermediate-scale quantum} (NISQ) algorithms, where the noise is given by single-qubit depolarizing noise of noise strength $\lambda$, require $\Omega((1-\lambda)^{-n})$ copies of an unknown state to predict even the absolute values of all Pauli expectation values. As this simplified Pauli shadow tomography task can be solved from two-copy measurements on $\mathcal{O}(n)$ copies and with efficient noiseless quantum processing, this raises the question of whether a NISQ verifier can delegate Pauli shadow tomography in an interactive proof system to achieve an advantage in the number of copies.
    \end{itemize}
\end{example}

As already mentioned, for a more comprehensive list of concrete problems and motivating examples, including quantum state certification \cite{buadescu2019quantum}, stabilizer testing \cite{gross2021schur}, Pauli spike detection \cite{chen2023efficient}, agnostic parity learning and agnostic quantum state tomography \cite{grewal2024agnostic, chen2024stabilizerbootstrappingrecipeefficient}, we refer the reader to Appendices~\ref{app:problems} and \ref{app:motivating_examples}. 

\subsection{Overview of main results}\label{ss:results_overview}

\subsubsection{Limitations of interactive proofs for verifying learning and testing}\label{sss:limitations_overview}

Our first result is a broadly applicable ``no-go" theorem that puts limits on when and how much a resource-constrained verifier can profit from interacting with an unconstrained but untrusted prover. In particular, this theorem can be used to show that for a wide variety of scenarios, resource-constrained learning or testing algorithms \textit{cannot} gain any query complexity advantages through interactions with a resource-unconstrained but untrusted prover over what they could achieve when trying to solve the problem on their own.

\begin{restatable}[Limitations of interactive proofs for many-vs-one distinguishing]{theorem}{thmLimitations}
    \label{theorem:no-go-interactive-many-vs-one-distinguishing-intro-version}%
    Let $\mathsf{A}$ be a class of resource-constrained algorithms 
    such that any $A\in\mathsf{A}$ requires $m$ oracle queries to solve the many-vs-one distinguishing task $(x_A,\mathcal{X}_R,\mathcal{X})$ from $\mathsf{O}_V$-access, with probability $\geq 1-\delta$.
    Let $(V,P)$, with $V\in\mathsf{A}$, be an interactive proof system with success probability $1-\delta$ for $(x_A,\mathcal{X}_R,\mathcal{X})$ from $(\mathsf{O}_V, \mathsf{O}_P)$-access. If, for all $x\in\mathcal{X}$, the actions of the honest prover $P$ with access to $\mathsf{O}_P(x_A)$, when interacting with $V$ with access to $\mathsf{O}_V(x)$, can be simulated by an algorithm in $\mathsf{A}$, then $V$ has to make at least $m$ oracle queries.
\end{restatable}

We emphasize that here we \emph{do not} assume that the honest prover $P$ can fully be simulated in $\mathsf{A}$, which would trivially imply that an interactive proof with such a prover yields no advantage. Rather, we assume only that the actions of $P$ \emph{when it has oracle access specifically to the accept instance} of the many-vs-one distinguishing task can be simulated in $\mathsf{A}$.

In line with the earlier examples, we will consider different classes of resource-constrained algorithms. For example, we can take $\mathsf{A}$ to be the class of algorithms that are memory-constrained or the class of NISQ quantum algorithms.

As discussed in detail in \Cref{s:limitations}, this no-go result can be applied to any scenario in which the only constraint on the verifier is a weaker form of oracle access than what is available to the prover, and to any quantum scenario in which a memory-constrained quantum verifier is interacting with a memory-unconstrained quantum prover via a \textit{classical} communication channel. As a consequence, our results imply that for any of the myriad tasks in which quantum memory is known to be a meaningful resource~\cite{chen2022exponential, huang2022quantum-advantage}, any party with a memory-constrained quantum device cannot overcome any of the known limitations implied by the constrained quantum memory via delegation, \textit{over classical communication channels}, to unconstrained but untrusted service providers. That is, we obtain from Theorem~\ref{theorem:no-go-interactive-many-vs-one-distinguishing-intro-version} the following corollaries:

\begin{restatable}[No advantage delegating to memory-unconstrained server via classical communication]{corollary}{corLimitationsmemory}
    \label{cor:no-go-delegating-memory}%
Any memory-constrained quantum verifier, interacting with an untrusted but memory-unconstrained quantum prover via a \textit{classical} communication channel, cannot gain any advantage in query complexity by delegating a many-vs-one distinguishing task to the prover via an interactive proof. Specifically, this implies that quantum memory-constrained verifiers cannot gain any query complexity advantage via delegation over what they could achieve in isolation for any of the following many-vs-one distinguishing tasks: 
\begin{enumerate}
    \item Purity testing~\cite{chen2022exponential}.
    \item State Pauli spike detection \cite{chen2022exponential, chen2024optimalpauli}, state preparation channel Pauli spike detection \cite{huang2022quantum-advantage}, and Choi state Pauli spike detection \cite{caro2023learning}.
    \item State certification and mixedness testing \cite{chen2022tight}.
    \item Unitarity testing \cite{chen2022exponential, chen2023unitarity}.
\end{enumerate}
\end{restatable}

\begin{restatable}[No query complexity advantage delegating to server with stronger oracle]{corollary}{corLimitationsoracles}
    \label{cor:no-advantage-delegating-stronger-oracle}
    Any verifier, interacting with an untrusted prover who has a stronger form of oracle access, cannot gain any advantage in query complexity by delegating a many-vs-one distinguishing task to the prover via an interactive proof. Specifically, this implies that for any of the following many-vs-one distinguishing tasks, verifiers with the indicated weaker oracle cannot gain any query complexity advantage via delegation to provers with the indicated stronger oracle:
    \begin{enumerate}
        \item Parity testing when $V$ has statistical query access to the unknown distribution and $P$ has sample access~\cite{kearns1998efficient}. 
        \item Purity testing when $V$ has quantum statistical query access to the unknown state and $P$ has access to copies of the unknown state~\cite{arunachalam2023role,nietner2023unifying}.
        \item Stabilizer testing when $V$ has quantum statistical query access to the unknown state and $P$ has access to copies of the unknown state~\cite{arunachalam2023role,nietner2023unifying}.
        \item Quadratic function testing when $V$ has quantum statistical query access to the unknown state and $P$ has access to copies of the unknown state~\cite{arunachalam2023role}.
    \end{enumerate}
\end{restatable}

Note that \Cref{theorem:no-go-interactive-many-vs-one-distinguishing-intro-version} and the subsequent corollaries are stated only for many-vs-one distinguishing problems, and provide query complexity lower bounds for their verification. However, in \Cref{s:limitations} we provide extensions both to learning problems and to computational complexity lower bounds. Using these extensions, one can also prove that memory-constrained verifiers cannot gain any query complexity advantage through delegation to untrusted but memory-unconstrained provers for any of the following learning or estimation problems:
\begin{enumerate}
    \item Purity estimation \cite{chen2022exponential}.
    \item (Pauli) threshold decision/search and (Pauli) shadow tomography~\cite{chen2022exponential,chen2024optimalpauli}.
    \item Unitarity estimation \cite{chen2022exponential, chen2023unitarity}.
    \item Symmetry classification \cite{aharonov2022quantum, chen2022exponential}.
    \item Learning polynomial-time quantum processes ~\cite{huang2022quantum-advantage}.
    \item Pauli transfer matrix learning~\cite{caro2023learning}.
\end{enumerate}
Similarly, one can prove that memory-constrained verifiers cannot gain an advantage in the number of measurements needed for Pauli channel eigenvalue estimation and channel Pauli spike detection \cite{chen2022quantum, chen2024tight, chen2023efficient} through untrusted classical delegation.
Moreover, relying on similar arguments, verifiers with restricted oracles cannot gain any query complexity advantages via delegation to untrusted provers with stronger oracles for the following learning problems:
\begin{enumerate}
    \item Learning Clifford circuits when $V$ has quantum process statistical query access~\cite{WadhwaDoosti} to the unknown circuit and $P$ has query access. 
    \item Pauli shadow tomography when $V$ has quantum statistical query access to the unknown state and $P$ has access to copies~\cite{arunachalam2023role}. 
    \item Learning displacement amplitudes when $P$ has access to copies of the unknown state and its conjugate state, but $V$ only has access to copies of the state~\cite{king2024exponential}.
\end{enumerate}
Finally, \Cref{theorem:no-go-interactive-many-vs-one-distinguishing-intro-version} also implies that a $\mathsf{NISQ}$ verifier cannot gain any query complexity advantage for state Pauli spike detection by 
interacting with an untrusted $\mathsf{BQP}$ prover; for Pauli shadow tomography, the possible advantage to be gained from delegation is extremely limited, in particular the $\mathsf{NISQ}$ verifier still faces an exponential lower bound even as part of an IP with an untrusted but powerful quantum prover.

\ins{The observations above highlight the \textit{necessity} of quantum communication for gaining advantages via delegation for all of the learning and testing problems mentioned above. We highlight that it is interesting to contrast this with the setting of delegated quantum computation \textit{without data}, where it is known via Mahadev's breakthrough results \cite{mahadev2022classical-verification} that (with a computational assumption) purely classical communication \textit{is sufficient} for the delegation of quantum computation~\cite{gheorghiu2018verification}. As such, the contrast between these two settings highlights the important and distinctive role played by the presence of data.}

\subsubsection{Power of interactive proofs for verifying quantum learning and testing}\label{sss:power_overview}

In the previous section we have shown how Theorem~\ref{theorem:no-go-interactive-many-vs-one-distinguishing-intro-version} implies stringent limitations on the ability of quantum memory-constrained parties to meaningfully delegate learning and testing tasks to memory-unconstrained parties via \textit{classical} communication channels. However, as explained in detail in Section~\ref{s:go_results}, these limitations may no longer apply when the communication is via a \textit{quantum} channel. As such, we shift our focus to the development of concrete interactive proof protocols in this setting, which allow memory-constrained verifiers to reap 
some advantage via delegation.

A straightforward first method for developing 
interactive proofs for quantum testing and learning in the presence of a resource mismatch is to employ \textit{verified blind quantum computing} (VBQC) protocols~\cite{gheorghiu2018verification, fitzsimonsUnconditionallyVerifiableBlind2017}.
VBQC protocols allow a client 
to securely delegate quantum computations to a remote server, ensuring the server learns 
nothing about the computations while the client can verify their correctness.
Moreover, via quantum communication, the client can provide the server with any input state of their choice to apply the target computation on, including the unknown quantum data.
In the present  context, a verifier $V$ can use VBQC to delegate to a prover $P$ the execution of a multi-copy measurement on multiple copies of the unknown quantum state by sending them one at a time.
Thus, $V$ can ensure that the intended computation was performed on the provided quantum data using VBQC. The following observation, which we derive in more detail in \Cref{app:VBQC}, summarizes this idea:

\begin{restatable}[Generic interactive proofs for quantum 
testing/learning -- VBQC version]{observation}{obsgenericVBQC}\label{lem:interactive-quantum-proofs-vbqc}
    If a quantum learning/testing problem can be solved with success probability $\geq 1-\frac{\delta}{2}$ from $m$ copies by a coherent multi-copy learner $\mathcal{A}$, then there exists an interactive proof system $(V,P)$, which succeeds with probability $\geq 1-\delta$, with quantum communication between an incoherent single-copy $V$ and a coherent multi-copy $P$, such that $V$ uses at most $m$ copies. The computational and communication complexities of this protocol are efficient with respect to the gate complexity of $\mathcal{A}$.
\end{restatable}

It is interesting to note that, while the VBQC protocols of \cite{gheorghiu2018verification, fitzsimonsUnconditionallyVerifiableBlind2017} and, for the same reasons, of \cite{broadbent2018howtoverify}, give rise to generic interactive proofs for quantum testing and learning, it is not true in general that verified quantum computing implies verified quantum testing and learning. Consider for instance verified quantum computing protocols with a classical verifier and multiple entangled quantum provers \cite{reichardt2013classical, coladangelo2024verifier-on-a-leash} or the breakthrough result of \cite{mahadev2022classical-verification} that allows a classical verifier to interactively verify the polynomial-time quantum computations of a single prover under a cryptographic hardness assumption. In these scenarios, the verifier is entirely classical, hence the communication between verifier and prover(s) is necessarily classical as well. Therefore, \Cref{theorem:no-go-interactive-many-vs-one-distinguishing-intro-version} implies that neither of the interactive verification protocols from \cite{reichardt2013classical, coladangelo2024verifier-on-a-leash,mahadev2022classical-verification} enable the verifier to achieve query complexity advantages \ins{for any of the learning and testing problems discussed above}. This demonstrates that, although certain forms of verified quantum computing also facilitate quantum testing and learning, the presence of quantum data -- particularly untrusted data on the prover side -- introduces a layer of complexity to verified quantum testing and learning, extending beyond the established scope of verifying quantum computations: The key issue becomes to verify that $P$ performs the correct processing \emph{on the correct data}.
\ins{Additionally, we note that VBQC may be too powerful a tool in the sense that for our purposes, we only require verifiability but not blindness. Whether lighter verification protocols than VBQC can be used to get an analogue of \Cref{lem:interactive-quantum-proofs-vbqc} is an interesting open question.}

The VBQC approach to interactive proofs for quantum testing and learning has the advantage of being generic and working almost out-of-the box. However, it also comes with two drawbacks.
\begin{enumerate}
    \item First, to make use of VBQC protocols in our context, $V$ needs to know the circuit (implementing the measurement) that they would run to solve the problem at hand if they were capable of coherent multi-copy quantum processing.   
    \item Second, all the quantum data comes from the verifier's oracle $\mathsf{O}_V$; the oracle $\mathsf{O}_P$ is not used at all. So, this way $V$ can overcome single-copy lower bounds through interaction with $P$, but it can never surpass multi-copy lower bounds.
\end{enumerate}

In the following we will give two different approaches that partially alleviate these drawbacks. Moreover, we will apply these approaches to derive concrete interactive proof protocols for specific problems, which allow resource-constrained verifiers to gain query complexity advantages over what they could achieve on their own, or over what even the unconstrained prover could achieve on their own. 
We summarize the query complexities obtained by our protocols in \Cref{table:advantages-via-IPs}, along with other relevant lower bounds which allow one to extract the advantages which can be achieved via the interactive proofs we provide.

\begin{table}
\begin{center}
\makebox[\linewidth][c]{
\begin{tabular}{||c | c c c c||} 
 \hline
  & $\begin{matrix}
  \text{Interactive proof}\\
  \text{(quantum comm.)}
  \end{matrix}$ & $\begin{matrix}
  \text{Interactive proof}\\
  \text{(classical comm.)}
  \end{matrix}$ & $\begin{matrix}
  \text{Coherent}\\
  \text{measurements}
  \end{matrix}$ & $\begin{matrix}
  \text{Incoherent}\\
  \text{measurements}
  \end{matrix}$ \\[5pt] 
 \hline\hline{} & {}& {}&{} &{}\\[-1em]
  Purity testing &  $\begin{matrix}
  \mathcal{O}(1) \\
  \text{(Theorem~\ref{theorem:interactive-purity-testing-black-box-intro-version})}
  \end{matrix}$ &  $\begin{matrix}
  \Omega(\sqrt{d}) \\
  \text{(Corollary~\ref{cor:no-go-delegating-memory})}
  \end{matrix}$ & $\begin{matrix}
  \mathcal{O}(1) \\
  \text{\cite{chen2022exponential}}
  \end{matrix}$ &  $\begin{matrix}
  \Omega(\sqrt{d}) \\
  \text{\cite{chen2022exponential}}
  \end{matrix}$\\ 
  \hline{} & {}& {}&{} &{}\\[-1em]
 $\begin{matrix}
  \text{Quantum}\\
  \text{state tomography}
  \end{matrix}$ &  $\begin{matrix}
  \mathcal{O}(d\epsilon^{-2}) \\
  \text{(Theorem~\ref{theorem:interactive-state-tomography-intro-version})}
  \end{matrix}$ &  $\begin{matrix}
  \Omega(d^{3/2} \epsilon^{-2}) \\
  \text{(Corollary~\ref{cor:no-advantage-qst})}
  \end{matrix}$ & $\begin{matrix}
  \Omega(d^2 \epsilon^{-2}) \\
  \text{\cite{odonnellEfficientQuantumTomography2017}}
  \end{matrix}$ &  $\begin{matrix}
  \Omega(d^3 \epsilon^{-2}) \\
  \text{\cite{haah2016sample}}
  \end{matrix}$\\
 \hline {} & {}& {}&{} &{}\\[-1em]
 $\begin{matrix}
  \text{Agnostic}\\
  \text{rank-$k$ 
  tomography}
  \end{matrix}$ & $\begin{matrix}
  \mathcal{O}(k^2\epsilon^{-4}) \\
  \text{(Theorem~\ref{theorem:agnostic-rank-k-tomography})}
  \end{matrix}$ & ? & $\begin{matrix}
  \widetilde{\Omega}(dk) \\
  \text{\cite{odonnell2016efficient}}
  \end{matrix}$ &? \\
 \hline{} & {}& {}&{} &{}\\[-1em]
 $\begin{matrix}
  \text{Agnostic}\\
  \text{stabilizer learning}
  \end{matrix}$ & $\begin{matrix}
  \mathcal{O}(\epsilon^{-2})\\
  \text{(Theorem~\ref{theorem:agnostic-stabilizer-state-learning-intro-version})}
  \end{matrix}$ & $\begin{matrix}
      \Omega(\sqrt{n})\\
      \text{(\Cref{corollary:no-go-stabilizer-testing})}
  \end{matrix}$ & $\begin{matrix}
  \Theta(n)\\
  \text{(via~\cite{arunachalam2023optimal})}
  \end{matrix}$ & $\begin{matrix}
  \Theta(n^2)\\
  \text{(via~\cite{arunachalam2026optimalstabilizertestinglearning})}
  \end{matrix}$\\[0.9em]
 \hline
\end{tabular}}\caption{
\textbf{Highlights of quantum results for memory-constrained verifiers.} The last two columns contain known upper and lower bounds for solving some of the problems we consider in this work \ins{in the non-delegated setting.}
We contrast these with lower bounds on the query complexity of verification with classical communication (implied by Theorem~\ref{theorem:no-go-interactive-many-vs-one-distinguishing-intro-version}), 
as well as upper bounds derived from the concrete interactive-proof protocols we provide when quantum communication is allowed.
}\label{table:advantages-via-IPs}
\end{center}
\end{table}

In removing the first drawback, we aim to construct protocols in which the only information that $V$ needs is: 
\begin{enumerate}
    \item The description of the problem -- e.g., the tuple $(x_A,\mathcal{X}_R,\mathcal{X})$ for a many-vs-one distinguishing problem.
    \item A promise that the honest prover can solve the problem given $m$ copies of the unknown quantum state.
\end{enumerate}
Said another way, we wish to construct interactive protocols in which the verifier $V$ can treat the prover $P$ as a ``problem solving black-box". This means we aim for protocols in which $V$ directly verifies that $P$ correctly solved the task of interest, rather than verifying that $P$ performed a certain desired quantum computation known to $V$. As such, we call an interactive protocol a \emph{black-box} protocol whenever $V$ is only required to know the two pieces of information listed above, and not a description of the coherent multi-copy quantum algorithm for solving the problem.  
Intuitively, we can think of a black-box protocol as one in which the verifier also outsources the programming expertise required to solve a learning/testing task of interest.
We give a concrete example of such a protocol below, for the problem of purity testing, defined in Problem~\ref{prob:purity_testing} and discussed in Example~\ref{example:purity-testing-informal}.

\begin{restatable}[Interactive proof for purity testing -- black-box version]{theorem}{thmInteractivePurityTesting}
    \label{theorem:interactive-purity-testing-black-box-intro-version}
    There exists a black-box interactive proof system $(V,P)$ for purity testing, between an incoherent single-copy $V$ and a coherent multi-copy $P$, efficiently communicating via a quantum channel, such that $V$ and $P$ both are computationally efficient and such that $V$ uses at most $\mathcal{O}(1)$ copies.
\end{restatable}

To remove the second drawback of VBQC-based interactive proof protocols, we could instead imagine a proof system where $V$ delegates the entire learning algorithm to $P$, including the quantum data requirements, and simply asks $P$ to send over the solution. For such a proof system to be sound, $V$ then only needs to use their data oracle to verify the quality of the alleged solution sent by the prover -- i.e., to decide whether it is a valid solution. 
For our first example in which this is achieved, we consider quantum state tomography as defined in Problem~\ref{prob:qst} and discussed in Example~\ref{example:quantum-state-tomography}.
More specifically, we construct an interactive proof protocol in which the verifier first asks the prover to solve the quantum state tomography problem using their own copies of the unknown quantum state, before then interacting with the prover to verify that this is indeed a valid solution. This way, as shown in Table~\ref{table:advantages-via-IPs}, the verifier is able to use fewer quantum state copies than even the coherent multi-copy prover would require to solve the problem in isolation.

\begin{restatable}[Interactive quantum state tomography -- informal]{theorem}{thmInteractiveTomography}
    \label{theorem:interactive-state-tomography-intro-version}
    There exists an interactive proof system $(V,P)$ for state tomography in trace distance with quantum communication between an incoherent single-copy $V$ and a coherent multi-copy $P$ such that $V$ uses at most $\mathcal{O}\left(\frac{d \log \delta^{-1}}{\varepsilon^2}\right)$ copies of the unknown \(d\)-dimensional quantum state.
\end{restatable}

This result extends to interactive verification of tomography of rank-$k$ quantum states -- i.e. the setting in which the unknown state is promised to be of rank-$k$. Here, $V$ uses at most $\mathcal{O}\left(\frac{k \log \delta^{-1}}{\varepsilon^2}\right)$ copies, irrespective of the dimension $d$.

Our second example of how to overcome the second drawback of VBQC-based  protocols revolves around the task of \textit{agnostic} rank-$k$ quantum state tomography. As defined in Problem~\ref{prob:agnostic_rank_k_quantum_state}, here a learner is given copy access to an \textit{arbitrary} unknown quantum state, and should output a description of a rank-$k$ quantum state which is a close-to-optimal rank-$k$ approximation to the unknown state. As for quantum state tomography, we show that here the verifier can also achieve significant sample complexity improvements upon even the coherent multi-copy lower bounds through an IP.

\begin{restatable}[Agnostic rank-$k$ state tomography -- informal]{theorem}{thmAgnostikLowRankTomography}
    \label{theorem:agnostic-rank-k-tomography}
    There exists an interactive proof system $(V,P)$ for interactive agnostic rank-$k$ state tomography in trace distance, with quantum communication between an incoherent single-copy $V$ and a coherent multi-copy $P$, such that $V$ uses at most \(\mathcal{O}\left(\frac{k^2 \log\delta^{-1}}{\varepsilon^4}\right)\) copies of the unknown state.
\end{restatable}

Our final result is concerned with the task of agnostic stabilizer state learning, defined in Problem~\ref{prob:gnostic_stabilizer_learning}. 
Again, we give an IP in which the verifier's sample complexity outperforms what even a coherent multi-copy quantum algorithm could achieve in isolation.

\begin{theorem}[Agnostic stabilizer learning -- informal]\label{theorem:agnostic-stabilizer-state-learning-intro-version}     
    There exists an interactive proof system $(V,P)$ for $8$-agnostic stabilizer state learning, between a single-copy verifier $V$ and a coherent multi-copy prover $P$, efficiently communicating via a quantum channel, such that $V$ is computationally efficient and such that $V$ uses at most $\mathcal{O}\left(\frac{\log\delta^{-1}}{\varepsilon^2}\right)$ copies of the unknown pure $n$-qubit state. 
\end{theorem}

\subsection{Related work}\label{ss:related_work}

As already mentioned, the framework we use here, described in Section~\ref{ss:framework}, extends the framework of interactive proofs for PAC verification recently introduced in  Ref.~\cite{goldwasser2021interactive}. This framework, motivated by the increasing relevance of delegating the solution of machine learning problems, allows one to reason about such delegation to untrusted servers with more or better quality data, but who are still constrained to be efficient. Since then, this framework has been extended in a variety of ways: To statistical query learning in Ref.~\cite{mutreja2022pac-verification}, to limited communication complexity in Ref.~\cite{oconnor2021delegating}, to classical verification of quantum learning in~\cite{caro2023learning}, and to a broader class of concrete learning problems in~\cite{gur2024power}. In this work, we contribute further to the understanding of this framework by providing powerful and broad reaching no-go results, as well as a variety of interactive proof protocols for the concrete and relevant setting of quantum memory-constrained verifiers interacting, via quantum communication channels, with memory-unconstrained quantum provers, \ins{with the goal of solving a wide range of quantum \textit{learning and testing} problems.} We note that in this sense, our work also extends to the quantum (learning and testing) setting, the notion of a \textit{streaming interactive proof}, in which a space-bounded verifier, that can only view a large input object as a stream, interacts with an unconstrained prover to verify a computation that requires large space~\cite{chakrabarti2014annotations,cormode2011verifyingcomputationsstreaminginteractive}.

Another body of related work is that of \textit{verified blind quantum computing} (VBQC) protocols \cite{gheorghiu2018verification, fitzsimonsUnconditionallyVerifiableBlind2017}, which aims to develop protocols via which quantum clients can securely and verifiably delegate quantum computations to remote servers, in a way that ensures that the server learns nothing about the computations. The primary differences between the standard VBQC setting and the setting we consider here are that (a) we consider problems for which the server is also in possession of a (potentially more powerful) data oracle \ins{(i.e. the parties are trying to solve a \textit{learning or testing} problem)} and (b) our verifier is limited by virtue of some resource constraint. Despite this, VBQC techniques do indeed provide an interesting toolbox for the construction of interactive proofs in our setting, which we discuss in Section~\ref{ss:proofs_from_VBQC}.

\ins{With the above in mind, we stress that our work \textit{initiates} the study of quantum \textit{learning and testing} problems in the delegated setting, and as such there is no prior work on either concrete protocols, or inherent limitations with which to directly compare the upper and lower bounds we obtain. Indeed, as discussed at length, given the inevitability of delegation for the vast majority of future quantum computations, together with the increasing importance of quantum algorithms for learning and testing, we strongly believe that the setting we propose and study here will be of practical relevance in the future, and we hope that our work stimulates and provides a foundation for further work on concrete interactive proof protocols for the delegation of quantum learning and testing.}

\subsection{Directions for future work}\label{ss:future_overview}

Our work raises several further questions regarding interactive proofs for (classical and quantum) testing and learning.
Below, we highlight some directions of interest.

\textbf{Further concrete limitations of interactive proofs for learning/testing:} One may aim to exhibit further scenarios, and testing or learning problems, in which \Cref{theorem:no-go-interactive-many-vs-one-distinguishing-intro-version} limits the advantage that a resource-constrained algorithm can gain by delegating to a resource-unconstrained prover. 
For instance, in terms of oracle constraints, one could consider the following scenarios:
\begin{enumerate}
    \item  A verifier with budget-constrained oracle access (as in 
    Ref.\ \cite{cesa-bianchi2011efficient}) and a prover with unconstrained oracle access.
    \item A verifier with noisy example access (as in Ref.\ \cite{angluin1988learning}) and a prover with noiseless example access.
    \item A quantum verifier for PAC learning which only has access to quantum example states, interacting with a prover with access to both quantum example states and the circuit that generates these states (as in Ref.~\cite{salmon2023provable}).
\end{enumerate}
All of these scenarios fit the setting of \Cref{cor:no-advantage-delegating-stronger-oracle}.
However, it remains to identify suitable many-vs-one distinguishing tasks for which these oracle constraints imply lower bounds. Additionally on the quantum side, we conjecture that \Cref{theorem:no-go-interactive-many-vs-one-distinguishing-intro-version} can be used to show that an incoherent quantum verifier cannot achieve Heisenberg-limited scaling in Hamiltonian learning \cite{huang2023heisenberg, li2023heisenberglimitedhamiltonianlearninginteracting, dutkiewicz2024advantagequantumcontrolmanybody} when interacting with an untrusted coherent prover able to interleave the Hamiltonian time evolution with control operations.

\textbf{Generic black-box interactive proofs:} In Theorem~\ref{theorem:interactive-purity-testing-black-box-intro-version} we gave a \textit{black-box} interactive proof protocol for delegating purity testing to memory-unconstrained provers. In Section~\ref{ss:black_box} we show how similar protocols can be constructed for unitarity testing, stabilizer testing, and Clifford testing. It is interesting to ask whether one can construct a \textit{generic} black-box interactive-proof protocol that can be used for delegating arbitrary many-vs-one distinguishing tasks while achieving a verifier query complexity that matches known coherent multi-copy complexities (even if computationally inefficiently).

\textbf{On the necessity of quantum communication for VBQC:} 
Juxtaposing our results from \Cref{sss:limitations_overview,sss:power_overview} highlights that quantum communication is a crucial ingredient to the power of interactive proofs for delegating quantum learning and testing\ins{ -- i.e. for delegating quantum algorithms for learning and testing \textit{from data}. It is interesting, and perhaps suprising, to contrast this with what is known about the power of classical communication for delegation of quantum \textit{computations} (i.e. without data). Specifically, in the setting without data, it is known from Mahadev's breakthrough result, that classical communication \textit{is sufficient} for delegating quantum computations, provided one makes a computational assumption~\cite{gheorghiu2018verification}. Without any computational assumptions, it is currently a major open question whether or not classical communication is sufficient for the delegation of quantum computation, or whether ``a little bit of quantumness'' is needed~\cite{aaronson2007prize, aharonov2013isquantummechanicsfalsifiable, gheorghiu2018verification}. In this sense, our impossibility results with purely classical communication for the seemingly closely related problem of delegation of learning and testing algorithms on data, highlights the potentially surprising fact that data indeed plays a very important distinctive role. We hope that these insights may inspire new approaches or techniques for the long-standing open problem of understanding whether or not quantum communication is truly necessary for the delegation of quantum computation. }

\textbf{From sketching to interactive proofs for testing and learning:} Sketching algorithms are a well studied approach to computation with memory constraints, with a wide variety of applications in streaming settings~\cite{nelson2020sketching}. As such, sketching algorithms are promising candidates for calculating ``certificates of validity" (see Section~\ref{ss:trivial}) that are needed for interactive proof protocols in which a memory-constrained verifier can outperform even an unconstrained isolated prover. Can concrete instances of such a connection be found, in either the classical case, where sketching algorithms are well developed, or in the quantum case, where sketching is less developed?

\subsection{Structure of this work}

This work is structured as follows: We begin in Section~\ref{s:preliminaries} by providing an overview of the required preliminary material. In particular, we give an overview of, and motivation for, the different testing and learning frameworks considered here. Additionally, we also give some insights into straightforward interactive proof protocols for a variety of settings, which motivate the particular focus of our work, and serve as a starting point for some of the interactive proofs that we provide.

With these preliminaries established, we then proceed in Section~\ref{s:limitations} to establish a variety of results on the \textit{limitations} of interactive proofs for verifying learning and testing. In particular, after proving a widely applicable no-go theorem, we discuss the implications for memory-constrained quantum verifiers with classical communication in Section~\ref{sec:nogo-implications-quantum}, the implications for verifiers with restricted oracles in Section~\ref{sec:nogo-implications-testing}, and the implications for NISQ verifiers in Section~\ref{sec:no-go-nisq}.

Finally, in Section~\ref{s:go_results} we construct a variety of concrete interactive proof protocols for the setting in which a quantum memory-constrained verifier interacts with a memory-unconstrained prover via a \textit{quantum} communication channel -- a setting for which our no-go result is not applicable. Specifically, we start by observing a generic protocol via universal verified blind quantum computing (VBQC) in Section~\ref{ss:proofs_from_VBQC}. As mentioned previously, this generic protocol has a variety of shortcomings. In particular it requires the verifier to know the algorithm it would like to delegate, and it does not allow the verifier the possibility of using less queries than that required by an isolated prover. With this in mind, we proceed in Section~\ref{ss:black_box} to construct \textit{black-box} interactive proof protocols, which address the first shortcoming of VBQC-based protocols in that they do not require the verifier to know the algorithm which should be run by the prover. We then address the second shortcoming of VBQC-based protocols in Section~\ref{ss:outperforming_provers}, by providing interactive proof protocols for quantum state tomography in Section~\ref{sss:tomography}, agnostic rank-$k$ quantum state tomography in Section~\ref{sss:rank_k}, and agnostic stabilizer learning in Section~\ref{sss:agnostic-stabilizer}, all of which allow the verifier to outperform lower bounds that apply to even a memory-unconstrained quantum learner in isolation.

\section{Setting and preliminaries}\label{s:preliminaries}

\subsection{Learning and testing frameworks}\label{ss:learning_frameworks}

In Definition~\ref{def:solve_a_problem} we have given a definition for what it means for an algorithm, with a given type of oracle access, to solve a learning or testing problem. In addition to this notion of solving a learning or testing problem we will also sometimes consider the query complexity of algorithms which solve a \textit{randomized} many-vs-one distinguishing problem (as mentioned in Example~\ref{example:purity-testing-informal}), which we define as follows:

\begin{definition}[Solving a randomized many-vs-one distinguishing problem]\label{def:solve_random_problem} 
    We say that an algorithm $\mathcal{A}$ solves a randomized many-vs-one distinguishing problem $(x_A,\mathcal{X}_R,\mathcal{X})$, with respect to some distribution $\mu$ over $\mathcal{X}_R$, from $\mathsf{O}$-access if, when $\mathcal{A}$ is given access to $\mathsf{O}(x)$, where
    \begin{enumerate}
    \item $x=x_A$ with probability $1/2$,
    \item $x\sim\mu$ with probability $1/2$,
    \end{enumerate}
    then, with probability $\geq 1-\delta$, $\mathcal{A}$ outputs a valid solution. Here, the probability of success is with respect to the randomness in the choice of $x$ and any potential randomness in the algorithm $\mathcal{A}$ and the oracle $\mathsf{O}$.
\end{definition}

Additionally, as hinted at in Section~\ref{ss:framework}, agnostic learning provides a unifying framework which encompasses a variety of different learning problems. We begin here by making the connection to \textit{realizable learning} more precise. For convenience, we restate the definition of an agnostic learning problem:

\labelRestate{def:agnostic-learning}
\defagnostic*

With this in hand we note that the special case in which $\mathcal{X}=\mathcal{M}$ is often referred to as a \textit{(proper) realizable} learning problem. In particular, note that for this case one always has, for any $x\in\mathcal{X}$, that
\begin{equation}
    \min_{m'\in\mathcal{M}}\ell(m',x) = 0.
\end{equation}
As a result, for any realizable learning problem, the parameter $\alpha$ plays no role,  and for any $x\in\mathcal{X}$, any model $m\in\mathcal{M}$ satisfying $\ell(m,x)\leq \epsilon$ is a valid solution. Given that $\alpha$ plays no role and that $\mathcal{X}=\mathcal{M}$, we can fully describe a realizable learning problem by the tuple $(\mathcal{M},\ell,\epsilon)$. With this connection clear, we summarize with the following definition:

\begin{definition}[Realizable learning problem]\label{def:realizable_problem} 
    A \textit{realizable learning problem} is defined by a tuple $(\mathcal{M},\ell,\epsilon)$, where $\mathcal{M}$ is a set, $\ell:\mathcal{M}\times\mathcal{M}\rightarrow \mathbb{R}_{\geq 0}$ is a loss function, and $\epsilon\in(0,1]$ is an accuracy parameter. When given oracle access to any instance $x\in\mathcal{M}$, any instance $h\in\mathcal{M}$ satisfying $\ell(h,x)\leq \epsilon$ is a valid solution.
\end{definition}
We call an algorithm that solves a realizable learning problem $(\mathcal{M},\ell,\epsilon)$, for any $\epsilon,\delta\in (0,1]$, a \textit{proper probably approximately correct} (PAC) learner for $\mathcal{M}$, which we often just abbreviate to ``a PAC learner". Additionally, we also note that any algorithm which solves the agnostic learning problem $(\mathcal{X},\mathcal{M},\ell,\alpha,\epsilon)$ for some set  $\mathcal{X}\supseteq \mathcal{M}$, immediately solves the realizable learning problem $(\mathcal{M},\ell,\epsilon)$. However, the converse statement is in general not true. As such, the realizable learning problem $(\mathcal{M},\ell,\epsilon)$ is (usually) strictly easier than the agnostic learning problem $(\mathcal{X},\mathcal{M},\ell,\alpha,\epsilon)$, and lower bounds for realizable learning imply lower bounds for agnostic learning.

\subsection{Interactive proofs via certificates of solution validity}\label{ss:trivial}

In this section we illustrate conditions under which learning problems admit straightforward interactive proofs, and give a variety of examples. These insights help motivate the selection of learning problems we consider in this work, as well as the strategies we use for developing concrete interactive proofs in Section~\ref{s:go_results}.

We begin by giving an intuitive definition of what it means for an algorithm to be able to \textit{decide valid solutions} to an agnostic learning problem (which, as per the previous section, includes the special case of realizable learning). 

\begin{definition}[Deciding valid solutions]\label{def:decide-valid} 
    We say that an algorithm $\mathcal{A}$ can decide valid solutions for an agnostic learning problem $(\mathcal{X},\mathcal{M},\ell,\alpha,\epsilon)$, with success probability $1-\delta$, from $\textsf{O}$-access if, for all $x\in\mathcal{X}$, when $V$ is given access to $\textsf{O}(x)$ as well as some candidate solution $h\in\mathcal{M}$, then $V$ can, with probability $\geq 1-\delta$, decide whether $h$ is a valid solution for the instance $x$.
\end{definition}
In words, an algorithm can decide valid solutions from $\textsf{O}$-access if, when given oracle access to an unknown instance $x$, and a candidate solution for that instance, it can decide whether or not this candidate solution is a valid solution. We will sometimes refer to a quantity that an algorithm needs to compute to decide the validity of a solution as a \emph{certificate} of solution validity. We define the query complexity of an algorithm for deciding valid solutions as the worst case number of queries it must make with respect to all possible unknown instances $x\in\mathcal{X}$ and candidate solutions $h\in\mathcal{M}$. 
With these notions in hand, we now make the following observation: Whenever a verifier $V$ can decide valid solutions for an agnostic learning problem from $\mathsf{O}_V$ access, then there exists a trivial interactive proof system (as in \Cref{definition:interactive-proofs-for-learning-informal}) for this agnostic learning problem in which the query complexity of the verifier is the same as the query complexity required for deciding valid solutions.

\begin{observation}[Trival IP when $V$ can decide valid solutions]\label{observation:trivial-ip-from-valid-solution-decision}  
    Assume $V$ can decide valid solutions for an agnostic learning problem $(\mathcal{X},\mathcal{M},\ell,\alpha,\epsilon)$ from $\mathsf{O}_V$-access, and that there exists an honest $P$ that can solve the agnostic learning problem from $\mathsf{O}_P$ access.
    Then, the following provides a valid interactive proof for solving $(\mathcal{X},\mathcal{M},\ell,\alpha,\epsilon)$ from $(\textsf{O}_P,\textsf{O}_V)$-access:
    \begin{enumerate}
        \item $V$ asks $P$ to solve the problem 
        on their own and to send a candidate hypothesis $h$.
        \item $V$ then uses $\textsf{O}_V(x)$ access to the unknown instance $x$ to decide whether $h$ is a valid solution, and accepts or rejects the interaction as appropriate.
    \end{enumerate}
Clearly, the query complexity of $V$ is the query complexity required to decide whether or not the candidate solutions are valid.
\end{observation}

In light of the above observation, a natural first strategy for obtaining interactive proofs for learning and testing would be to develop algorithms for deciding valid solutions. We note that one is particularly interested in the case when there exists an algorithm for deciding valid solutions from $\mathsf{O}$-access which has a smaller query complexity than that required to solve the problem from $\mathsf{O}$-access. Specifically, when this is the case, the trivial interactive proof from Observation~\ref{observation:trivial-ip-from-valid-solution-decision} allows the verifier to immediately gain a query complexity advantage via the interactive proof. We show in the examples below that there exist classes of interesting learning problems which admit straightforward algorithms for deciding valid solutions, and therefore straightforward interactive proofs.

\begin{example}[Triviality of interactive proofs for realizable learning of Boolean functions]\label{example:trivial-boolean} As discussed in Ref.~\cite{goldwasser2021interactive}, consider the realizable learning problem $(\mathcal{M},\ell,\epsilon)$ where $\mathcal{M}$ is any set of Boolean functions on $n$ bits, and 
\begin{equation}
        \ell(h,c) := \mathrm{Pr}_{x\sim D}[h(x)\neq c(x)]
\end{equation}
is the misclassification probability with respect to some distribution $D$ over $\{0,1\}^n$. Additionally, let $\textsf{O}_V(c)$ be the random example oracle -- i.e.,  when queried, oracle $\textsf{O}_V(c)$ responds with $(x,c(x))$ where $x$ is drawn from $D$. In this case, when given access to $\textsf{O}_V(c)$ for some unknown instance $c$, as well as some candidate $h$, the verifier $V$ can easily use $O(1/\epsilon^2)$ queries to $\textsf{O}_V(c)$ to estimate $\ell (h,c)$ accurately enough to decide whether $\ell (h,c)\leq \epsilon$, and therefore whether $h$ is a valid solution. As such, for any realizable learning problem, whenever the verifier is constrained through a random example oracle (as opposed to membership queries for example), there exists an efficient interactive proof with $n$-independent query complexity  $O(1/\epsilon^2)$.
\end{example}

\begin{example}[Triviality of interactive proofs for realizable quantum state learning w.r.t. fidelity]\label{example:trivial-verification-quantum} 
    Consider the realizable learning problem $(\mathcal{M},\ell,\epsilon)$ where $\mathcal{M}$ is any set of pure $d$-dimensional quantum states and $\ell$ is the fidelity. Additionally, let $\mathsf{O}_V$ be the standard quantum copy oracle which when queried provides a copy of the unknown state. For a concrete example, consider $d=2^n$ and $\mathcal{M}=\mathrm{Stab}_n$ the set of pure $n$-qubit stabilizer states. Here, a single-copy memory-constrained verifier $V$ can estimate the fidelity between a given candidate stabilizer state and any unknown pure state from $\mathcal{O}(1)$ copies (e.g., using the technique from \cite{flammia2011direct}), and thus decide validity of solutions in the sense of Definition~\ref{def:decide-valid}. This reasoning extends to $\mathcal{M}$ being the class of states prepared by Clifford gates and $\mathcal{O}(\log n)$ many $T$ gates.
    As another example, the even computationally efficient learning algorithms of~\cite{zhao2023learning} for states prepared by $\mathcal{O}(\log n)$-size circuits and of~\cite{huangLearningShallowQuantum2024a} for states prepared by  geometrically local shallow 2D circuits output descriptions of the hypothesis state in terms of a circuit for preparing that state. The fidelity of a proposed hypothesis can then be estimated by a single-copy $V$ simply by applying the inverse of the hypothesis circuit on the unknown state and estimating the overlap with the all-zero state.
    For our final example, \cite{huang2024certifying} can be viewed as giving rise to a single-copy procedure for deciding validity of proposed pure hypothesis states, assuming that those states have good relaxation times and that the prover provides the hypothesis in a form that allows to evaluate computational basis amplitudes (possibly up to a normalization)\footnote{\cite{gupta2025singlequbitmeasurementssufficecertify} recently gave an algorithm that achieves essentially the same certification guarantees as in \cite{huang2024certifying} without assumptions on mixing times, certifying arbitrary single-qubit states with single-copy measurements.}. 
\end{example}

The above observation and examples motivate the restriction of our focus to learning problems in which $V$ is \textit{not} able to trivially decide valid solutions. This is certainly the case for some realizable PAC learning problems, such as realizable distribution learning with respect to the total variation distance, or full quantum state tomography (realizable learning the set of \textit{all} quantum states), which we study in Section~\ref{sss:tomography}. However it is more typical for agnostic learning problems, as defined in \Cref{def:agnostic-learning}. More specifically, for an agnostic learning problem, even if $V$ can evaluate $\ell (m,x)$ when given access to $\textsf{O}(x)$, deciding whether or not $m$ is a valid solution requires knowing $\min_{m'\in\mathcal{M}}\left[\ell (m',x)\right]$, which is often a significant obstacle.

In addition to informing our choice of problems to study, the observations and examples above also motivate some of the strategies we use to develop meaningful interactive proof systems. In particular, the interactive proof protocols developed in Section~\ref{ss:outperforming_provers} are based around the construction of query efficient algorithms for deciding valid solutions, with which we can then instantiate the interactive proof protocol from Observation~\ref{observation:trivial-ip-from-valid-solution-decision}.

\subsection{Quantum memory as a resource}\label{ss:memory}

Throughout this work we study interactive proofs between quantum learning and testing algorithms with and without the resource of quantum memory. Here, in line with Refs.~\cite{chen2022exponential, huang2022quantum-advantage,chen2024optimalstate}, we clarify precisely what is meant by ``the resource of quantum memory". To start, as we are concerned with learning and testing algorithms, we will always assume that the algorithm has access to some oracle which, when queried, provides a copy $n$-qubit quantum state. This oracle could simply always provide a single copy of some fixed quantum state $\rho$ when queried, as is the case for quantum state tomography. However, the oracle could also be an oracle for an unknown quantum channel $T$, which when queried with some state $\rho$, provides as output the state $T(\rho)$. 

With this in hand, we say that an algorithm has a $k$-copy quantum memory if it can first perform $k$ sequential oracle queries, followed by computations and measurements on the $(nk)$-qubit quantum state $\bigotimes^k_{i = 1}\rho_i$, where $\rho_i$ represents the output of the $i$'th oracle query. As algorithms with a $k$-copy memory can first entangle the output states from multiple oracle queries before measurement, we also sometimes say that such algorithms have the resource of \textit{entangled/coherent multi-copy measurements}, or refer to such algorithms as \textit{coherent multi-copy algorithms}. In contrast, we call any algorithm that only has access to a $1$-copy quantum memory a single-copy quantum algorithm. As algorithms with only a single-copy quantum memory cannot perform entangling measurements on multiple oracle outputs, we sometimes refer to these algorithms as \textit{incoherent single-copy algorithms.}

Given this, when we refer to a memory-constrained quantum algorithm, we always refer to an algorithm which only has access to a single-copy quantum memory. And whenever we refer to an algorithm without a memory constraint, we refer to an algorithm with a $k$-copy memory for some $k\geq 2$.

\section{Limitations of interactive proofs for verifying learning and testing}\label{s:limitations}

Our first main contribution is a broadly applicable no-go result that puts limits on when and how much a resource-constrained verifier can profit from interacting with an unconstrained but untrusted prover. We begin by abstractly stating our result for testing problems. We then comment on its extensions to more general learning problems, before demonstrating its implications for concrete many-vs-one distinguishing problems. 

\labelRestate{theorem:no-go-interactive-many-vs-one-distinguishing-intro-version}
\thmLimitations*

Before giving the proof, we stress that in \RefRestate{theorem:no-go-interactive-many-vs-one-distinguishing-intro-version}
we do not require the class $\mathsf{A}$ to be able to simulate the actions of $P$ 
with access to any $\mathsf{O}_P(x)$.
Rather, $\mathsf{A}$ only needs the ability to simulate the behavior of $P$ in an interaction where $V$ has oracle access to an arbitrary $x\in\mathcal{X}$, but $P$ specifically has access to the oracle $\mathsf{O}_P(x_A)$ -- i.e.,  the oracle for the accept instance. With this in mind, the proof of~\Cref{theorem:no-go-interactive-many-vs-one-distinguishing-intro-version}, which generalizes ideas from Refs.\ \cite{mutreja2022pac-verification, caro2023verification}, and whose reasoning is illustrated in \Cref{fig:no-go-template}, is given below.

\begin{proof}[\RefRestate{theorem:no-go-interactive-many-vs-one-distinguishing-intro-version}] Assuming a pair $(V,P)$ as in the statement of the theorem, we construct a distinguishing algorithm $D\in\mathsf{A}$ that solves the many-vs-one distinguishing task as follows: Using their query access to $\mathsf{O}_V(x)$, $D$ simulates the interaction between $V$, given access to $\mathsf{O}_V(x)$, and $P$, given access to $\mathsf{O}_P(x_A)$.
If the simulated $V$ either aborts the interaction or outputs ``reject'', then $D$ outputs ``reject''. Otherwise, $D$ outputs ``accept''.
To see that $D$ indeed successfully solves the task, consider the two possible cases. On the one hand, if $x=x_A$, then by the completeness guarantee of $(V,P)$, the simulated verifier $V$ does not abort and correctly outputs ``accept'' with probability $\geq 1-\delta$. On the other hand, if $x\in\mathcal{X}_R$, then by the soundness guarantee of $(V,P)$, the simulated verifier $V$ does not abort and incorrectly outputs ``accept'' with probability $<\delta$. So, in both cases, $D$ outputs the correct answer to the many-vs-one distinguishing task with probability $\geq 1-\delta$.
\end{proof}

\begin{figure}
    \centering 
  \includegraphics[width=0.88\textwidth]{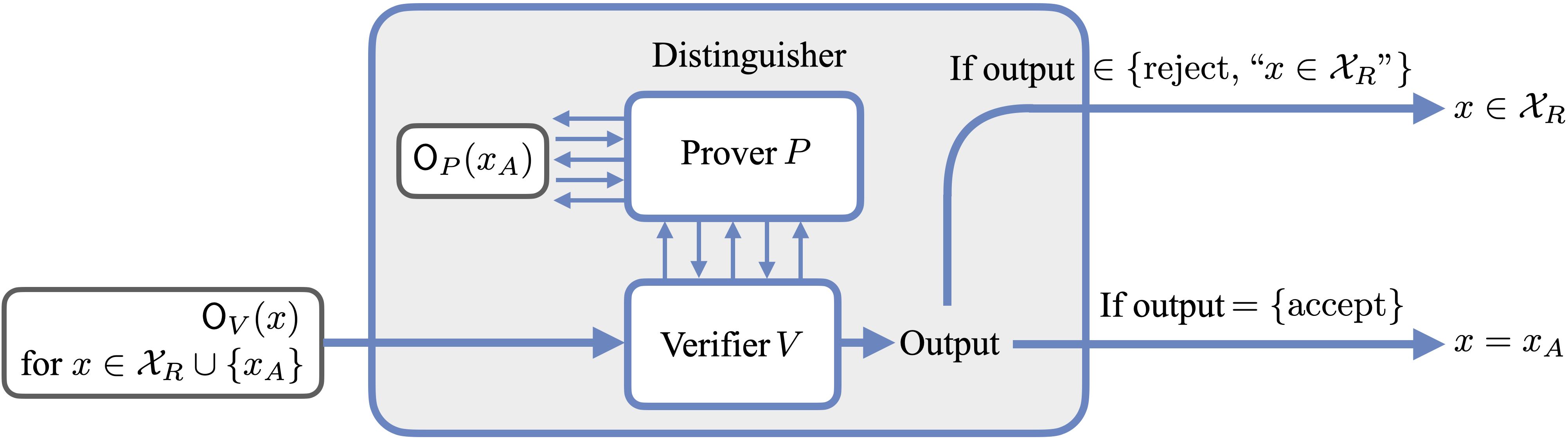} 
    
    \caption[Illustration of proof of \Cref{theorem:no-go-interactive-many-vs-one-distinguishing-intro-version}]{\textbf{Illustration of the proof of \Cref{theorem:no-go-interactive-many-vs-one-distinguishing-intro-version}}: By simulating the interaction between the verifier $V$ and the prover $P$, where $V$ is given oracle access to $x$ via $\mathsf{O}_v(x)$ and $P$ is given oracle access to $x_A$ via $\mathsf{O}_P(x_A)$, we construct a distinguisher for the original many-vs-one distinguishing task. 
    }\label{fig:no-go-template} 
\end{figure}

Having established the above result on limitations of interactive proofs for many-vs-one distinguishing tasks, we comment on two natural extensions of~\RefRestate{theorem:no-go-interactive-many-vs-one-distinguishing-intro-version}.

\textbf{Extension to computational complexity:} Firstly, whereas the focus in \Cref{theorem:no-go-interactive-many-vs-one-distinguishing-intro-version} is on query complexity, we straightforwardly obtain a version for computational complexities when taking the computational cost of simulating $P$ into account. Informally, let us assume that there exists an interactive proof protocol for solving a specific many-vs-one distinguishing task in which the verifier $V\in\mathsf{A}$ in isolation needs $C_V$ computational time, and for which there exists some algorithm in $\mathsf{A}$ that can simulate all actions of the honest prover, with access to the accept instance oracle $\textsf{O}_P(x_A)$, using $C_P$ computational time. Analogously to the proof of \RefRestate{theorem:no-go-interactive-many-vs-one-distinguishing-intro-version}, this implies a distinguisher $D\in\mathsf{A}$ which can solve the many-vs-one distinguishing task with $C_V + C_P$ computational time. As such, if a lower bound of $C_{\mathsf{A}}$ computational time exists for any algorithm in $\mathsf{A}$, then one must have $C_V + C_P \geq C_{\mathsf{A}}$ so $C_V\geq C_{\mathsf{A}} - C_P$.

\textbf{Extension to learning problems:} While \RefRestate{theorem:no-go-interactive-many-vs-one-distinguishing-intro-version} is phrased for many-vs-one distinguishing problems, it often also implies similar limitations on interactive proofs for learning. Informally, this is the case if: 
\begin{enumerate}
    \item The learning problem admits an underlying ``easier-than-learning'' many-vs-one distinguishing task -- i.e.,  a many-vs-one distinguishing task such that any valid solver for the learning problem gives rise to a solver for the many-vs-one distinguishing task. Such underlying distinguishing problems often exist, and are used to prove lower bounds for the corresponding learning problem.
    \item \RefRestate{theorem:no-go-interactive-many-vs-one-distinguishing-intro-version} applies to the underlying many-vs-one distinguishing task.
\end{enumerate}
Intuitively, if an interactive proof system existed for the learning problem, then by point 1 above there would exist an interactive proof system for the distinguishing task, whose query complexity is lower bounded via point 2. To be more quantitative, assume that an interactive proof for solving the learning problem requires $m_l$ queries by the verifier, solving the distinguishing problem when given a valid solution to the learning problem requires $m_{l\rightarrow d}$ queries by the verifier, and that solving the distinguishing problem in isolation requires $m_d$ verifier queries. 
\Cref{theorem:no-go-interactive-many-vs-one-distinguishing-intro-version} then implies that
\begin{equation}\label{eq:learning_to_testing}
    m_l + m_{l\rightarrow d}\geq m_d\, ,
\end{equation}
and therefore that $m_l \geq m_d - m_{l\rightarrow d}$, which can be non-trivial whenever $m_{l\rightarrow d}$ is sufficiently small. We give concrete examples of such reductions in Corollaries~\ref{cor:no-advantage-qst} and~\ref{cor:agnostic_parity_learning}.

With this in mind, in order to illustrate the scope of \RefRestate{theorem:no-go-interactive-many-vs-one-distinguishing-intro-version} and its extensions, we start by noting that, as \RefRestate{theorem:no-go-interactive-many-vs-one-distinguishing-intro-version} constitutes a generalization of ideas from Refs.\ \cite{mutreja2022pac-verification, caro2023verification} \ins{to verifiably solving many-vs-one distinguishing problems with general data access oracles}, it can in particular be used to recover \cite[Theorem 8]{mutreja2022pac-verification} \ins{on sample complexity lower bounds for classical PAC verification} and \cite[Theorem 15]{caro2023verification} \ins{on sample complexity lower bounds for classical verification of quantum PAC learning}.
To demonstrate the further reach of \Cref{theorem:no-go-interactive-many-vs-one-distinguishing-intro-version} and its extensions, we now highlight its implications for a wide variety of concrete problems and settings.

\subsection{Implications for quantum memory-constrained verifiers}\label{sec:nogo-implications-quantum}

\subsubsection{Many-vs-one distinguishing tasks}

We first demonstrate that \RefRestate{theorem:no-go-interactive-many-vs-one-distinguishing-intro-version} can be applied to any setting in which a quantum memory-constrained verifier wants to delegate a many-vs-one distinguishing task to a memory-unconstrained but untrusted quantum prover via a \textit{classical} communication channel. We also give a variety of concrete examples of problems for which a memory-constrained verifier faces query complexity limitations due to the memory constraint, and might have hoped to gain some advantage via delegation to a memory-unconstrained server.

\labelRestate{cor:no-go-delegating-memory}
\corLimitationsmemory*

\begin{proof} 
    To prove the above corollary via \RefRestate{theorem:no-go-interactive-many-vs-one-distinguishing-intro-version} we have to show that, for any many-vs-one distinguishing task $(x_A,\mathcal{X}_R,\mathcal{X})$, a quantum memory-constrained algorithm $\mathcal{A}$ that can only make incoherent single-copy measurements (but otherwise has no other computational constraints) can always simulate the actions of an honest prover $P$, with access to $\mathsf{O}_P(x_A)$ and the ability to make coherent multi-copy measurements, when interacting with $V$. We do this by considering a variety of cases:
    \begin{enumerate}
        \item A first trivial case is when the honest prover $P$ never queries its oracle, and never makes any coherent multi-copy measurements. In this case, the honest prover $P$'s actions can clearly be simulated by even a randomized classical algorithm $\mathcal{A}$ (which has no computational or classical memory constraints). 
        \item A second slightly less trivial case is when the honest prover $P$ never queries its oracle, but does at some points during the interaction with $V$ make a coherent multi-copy measurement. Given that $V$ and $P$ communicate via classical communication, the only way this can occur is if the honest prover $P$ at some point prepares some multi-register quantum state conditioned on the classical messages received from $V$, and then later performs some coherent multi-copy measurement also conditioned on classical messages received from $V$. In this case, again even a randomized classical algorithm could simulate (albeit computationally inefficiently) the actions of the known honest prover $P$, through a classical simulation of the state preparation and measurement process.
        \item Next we consider the case in which the honest prover $P$ does query its oracle $\mathsf{O}_P(x_A)$. First, we note that the most general possibility is that $\mathsf{O}_P(x_A)$ is a quantum channel which can be queried with a quantum state, and returns a quantum state. Next, we note that as the communication from $V$ to $P$ is purely classical, any time $P$ queries its oracle with a quantum state, it must have prepared this state itself conditioned on classical messages from $V$. Similarly, because $P$ can only send classical information to $V$, the most general case is that it sends back classical results of coherent-multi copy measurements on registers that contain states it either prepared, or received from oracle queries to $\mathsf{O}_P(x_A)$. As in the previous case, $\mathcal{A}$ can classically simulate any state preparations done by $P$ (and obtain classical descriptions of the states). Additionally, as $x_A$ is a known fixed instance, $\mathcal{A}$ can also classically simulate the action of $x_A$ on the states prepared by $P$ (and obtain classical descriptions of the output states). Finally, $\mathcal{A}$ can also classically simulate the (coherent multi-copy) measurements performed on the states that $P$ stores in its registers (as $\mathcal{A}$ has classical descriptions of all of these states). Hence, $\mathcal{A}$ can always simulate -- albeit inefficiently -- the actions of the honest prover $P$, with access to $\mathsf{O}_P(x_A)$ and the ability to make coherent multi-copy measurements, when interacting with $V$. As such, \RefRestate{theorem:no-go-interactive-many-vs-one-distinguishing-intro-version} applies.
    \end{enumerate} 
\end{proof}

The crucial aspect in the proof of \Cref{cor:no-go-delegating-memory} is that the communication from $V$ to $P$ is classical. In fact, even if we allow $P$ to send quantum states to $V$ (which have to fit into $V$'s memory), a variant of the above reasoning still applies. Namely, a memory-constrained quantum $\mathcal{A}$ can first classically (inefficiently) simulate $P$'s (possibly multi-copy) quantum processing \ins{(via straighforward state-vector simulation with exponentially large vectors and matrices)} to obtain a classical description of the state that $P$ would send, and then (inefficiently) prepare a copy of that state to be sent to $V$. Hence, \Cref{theorem:no-go-interactive-many-vs-one-distinguishing-intro-version} applies and we have justified the following observation.

\begin{observation}[Quantum communication from prover to verifier does not help]\label{observation:quantum-prover-to-verifier-doesnt-help}
    The conclusion of \Cref{cor:no-go-delegating-memory} remains valid even if the communication from prover to verifier can be quantum, as long as the communication from verifier to prover remains classical.
\end{observation}

We briefly note here that the conclusions of Corollary~\ref{cor:no-go-delegating-memory} \textit{do not} hold in the setting when the communication from $V$ to $P$ is also allowed to be via a quantum channel. In this case, one could consider a situation in which $V$ repeatedly queries its oracle $\mathsf{O}_V$, each time obtaining an unknown state and sending it to the prover. If the honest prover $P$ then has to collect all these states and perform a coherent multi-copy measurement, then the honest $P$ \textit{could not} be simulated by a single-copy algorithm $\mathcal{A}$, even without any computational constraints. In particular, because the states it should act on are \textit{unknown} (they come from $V$'s oracle queries), a single-copy algorithm cannot prepare and compute on classical representations of these states, and it cannot store them quantumly with its limited memory. We discuss this setting in much greater depth in Section~\ref{s:go_results}.

\ins{Before proceeding however, as mentioned earlier, we highlight again the fact that the above results establish the \textit{necessity} of quantum communication for the delegation of quantum algorithms for learning and testing from data. This is in stark contrast to the standard setting of delegated quantum computing, where Mahadev's breakthrough results show that (under a computational assumption) classical communication \textit{is sufficient} for the delegation of quantum computations~\cite{gheorghiu2018verification}. As such, the above results highlight the important distinctive role played by the presence of \textit{data}.}

Finally, using the single-copy lower bound for stabilizer testing recently proved in \cite{hinsche2024singlecopystabilizertesting}, we can establish a limitation on the potential advantage in delegating stabilizer testing: 
\begin{corollary}[Limited advantage in delegating stabilizer testing]\label{corollary:no-go-stabilizer-testing}
Any memory-constrained quantum verifier that interacts with an untrusted but memory-unconstrained quantum prover via a classical communication channel to solve the stabilizer testing problem with success probability $\geq 2/3$ has to use at least $\Omega(\sqrt{n})$ copies of the unknown state.
\end{corollary}
In contrast to \Cref{cor:no-go-delegating-memory}, we do not entirely rule out a potential advantage from delegation in this task, but instead establish a limitation. The best known single-copy stabilizer testing algorithm requires $O(n)$ copies \cite{hinsche2024singlecopystabilizertesting}, leaving room for potential advantage when compared to the $\Omega(\sqrt{n})$ lower bound.

The proof of \Cref{corollary:no-go-stabilizer-testing} is entirely analogous to that of \Cref{cor:no-go-delegating-memory} and is hence omitted.

\subsubsection{Learning problems}\label{sss:limitations_learning_memory}

\Cref{cor:no-go-delegating-memory} has focused on the limitations of interactive proofs for quantum many-vs-one distinguishing problems in which quantum memory is a relevant resource.
Next, to illustrate the extension of \Cref{theorem:no-go-interactive-many-vs-one-distinguishing-intro-version} to learning problems as discussed above, we turn our attention to quantum learning problems. 
Here, we give concrete examples of limitations of interactive proofs for quantum learning problems in which quantum memory is a relevant resource.
While we formulate these results for verifiers and provers that communicate classically, they remain valid even if the prover can send quantum messages to the verifier (compare \Cref{observation:quantum-prover-to-verifier-doesnt-help}).

We begin with the problem of full quantum state tomography, as defined in Problem~\ref{prob:qst} and discussed in Example~\ref{example:quantum-state-tomography}. In contrast to its many-vs-one distinguishing counterpart, quantum state certification, here our results do not completely eliminate the possibility of an advantage from delegation. However, we still establish strong limitations on how much of an advantage is possible. For comparison see Table~\ref{table:advantages-via-IPs}.

\begin{corollary}[Limited advantage delegating quantum state tomography]\label{cor:no-advantage-qst} 
    Any memory-constrained quantum verifier that interacts with an untrusted but memory-unconstrained quantum prover via a \textit{classical} communication channel to solve the qudit state tomography problem to accuracy $\varepsilon$ (in trace distance) with success probability $\geq 2/3$ has to use at least $\Omega(d^{3/2}/\varepsilon^2)$ many copies of the unknown qudit state.
\end{corollary}

As the copy complexity of single-copy qudit state tomography is given by $\Theta(d^3/\varepsilon^2)$ \cite{kueng2017low, chen2023whendoesadaptivity}, \Cref{cor:no-advantage-qst} demonstrates that a memory-constrained verifier can gain at most a quadratic improvement in the $d$-dependence from classically interacting with a memory-unconstrained prover.
In particular, for $n$-qubit states, we have $d=2^n$ and thus the verifier requires an exponential number of copies even in the interactive proof setting with classical communication.

\begin{proof}
    Recall that solving the many-vs-one distinguishing task of qudit mixedness testing with accuracy $\varepsilon$ and success probability $\geq 2/3$ from incoherent single-copy access requires $\Omega(d^{3/2}/\varepsilon^2)$ copies \cite{chen2022tight}.
    Thus, given \Cref{cor:no-go-delegating-memory} and using the language introduced around \Cref{eq:learning_to_testing}, to prove \Cref{cor:no-advantage-qst} we want to argue that any algorithm solving qudit state tomography gives rise to an algorithm solving qudit mixedness testing, and to characterize the query overhead $m_{l\to d}$.
    This, however, is obvious even with $m_{l\to d}=0$. Namely, to obtain the distinguishing algorithm, first run the state tomography algorithm with accuracy $\varepsilon/3$; if the output state is $(\varepsilon/3)$-close (in trace distance) to the maximally mixed state, output ``maximally mixed''; if the output state is $(2\varepsilon/3)$-far (in trace distance) from the maximally mixed state, output ``far from maximally mixed''.
    This finishes the proof.
\end{proof}

Analogous arguments and reductions can be applied to show that, as a consequence of \RefRestate{theorem:no-go-interactive-many-vs-one-distinguishing-intro-version} an incoherent single-copy quantum verifier $V$ classically communicating with a coherent multi-copy quantum prover cannot gain any quantum query/sample complexity advantage over the best known lower bounds (see cited references) for what $V$ could achieve on their own for any of the following learning tasks:
\begin{itemize}
    \item Purity estimation \cite{chen2022exponential},
    \item (Pauli) threshold decision/search and (Pauli) shadow tomography \cite{chen2022exponential, chen2024optimalpauli},
    \item Unitarity estimation \cite{chen2022exponential, chen2023unitarity},
    \item Symmetry classification \cite{aharonov2022quantum, chen2022exponential},
    \item Learning polynomial-time quantum processes~\cite{huang2022quantum-advantage},
    \item Pauli transfer matrix learning~\cite{caro2023learning}.
\end{itemize}

\subsection{Implications for verifiers with restricted oracles}\label{sec:nogo-implications-testing}

\subsubsection{Many-vs-one distinguishing tasks}

\Cref{sec:nogo-implications-quantum} has demonstrated the implications of \RefRestate{theorem:no-go-interactive-many-vs-one-distinguishing-intro-version} for limitations of interactive proofs in which verifier and prover have the same kind of data access but the verifier has less quantum memory than the prover. That is, the difference between verifier and prover was one of data processing capabilities. 
Here, we now focus on settings in which the verifier and the prover differ only in that the latter has access to a stronger data oracle than the former.
Our focus will again be on query complexity. As before, we first consider interactive proofs for many-vs-one distinguishing tasks.

\labelRestate{cor:no-advantage-delegating-stronger-oracle}
\corLimitationsoracles*

\begin{proof} Lets first begin by assuming that $V$ is computationally unbounded. Specifically, we assume that $V$ has no memory or computational constraints, and that the only resource constraint on $V$ is a weaker form of oracle access than what is available to the prover~$P$ (which results in $V$ requiring more oracle queries than $P$ to solve the many-vs-one distinguishing task at hand). As $x_A$ is a known fixed instance, any computationally unbounded algorithm can simulate any query to $\mathsf{O}_P(x_A)$, and therefore also the actions of a known honest prover with access to $\mathsf{O}_P(x_A)$. As such, \RefRestate{theorem:no-go-interactive-many-vs-one-distinguishing-intro-version} applies to such a computationally unbounded verifier. Now we note that any other verifier with some computational constraints could not achieve a better query complexity than the computationally unbounded verifier we just assumed. As such, the statement of \RefRestate{theorem:no-go-interactive-many-vs-one-distinguishing-intro-version} holds.
\end{proof}

\subsubsection{Learning problems}\label{sec:nogo-implications-learning}

Having established these implications on interactive proofs for testing given weaker oracle access, we now give examples of how 
\RefRestate{theorem:no-go-interactive-many-vs-one-distinguishing-intro-version} 
implies limitations on interactive proofs for a learning problem, via the reduction to many-vs-one distinguishing tasks described earlier. We start with agnostic parity learning, as defined in Problem~\ref{prob:agnostic_parity_learning}.

\begin{corollary}[Limitations of interactive proofs for agnostic parity learning]\label{cor:agnostic_parity_learning} 
    Consider the task of agnostic parity learning, described in Problem~\ref{prob:agnostic_parity_learning}. Consider a verifier $V$ that only has access to an SQ-oracle $\mathsf{O}_{V}^{\mathrm{SQ}}$ with inverse polynomial accuracy, and that interacts with a prover $P$ who has access to a random example oracle $\mathsf{O}_P^{\mathrm{EX}}$. If $(V,P)$ is an interactive proof system for agnostic parity learning, then $V$ requires at least $2^{\Omega(n)}$ statistical queries.
\end{corollary}

\begin{proof}
    We start by showing that if $V$ can solve the agnostic parity learning problem, then $V$ can solve the parity testing problem with a single extra statistical query -- i.e.,  parity testing provides a suitable ``easier than learning'' many-vs-one distinguishing task for agnostic parity learning with $m_{l\to d}=1$. To do this, we start by noting that if $V$ is given a bit string $s\in\{0,1\}^n$, as well as access to the SQ oracle $\textsf{O}^{SQ}_V(D)$ for some unknown distribution $D$, then $V$ can estimate (up to inverse polynomial precision $\tau$) the misclassification probability $\ell (D_s,D)$ defined in Eq.~\eqref{eq:misclassification_prob} via a single query to $\textsf{O}^{SQ}_V(D)$. More specifically, $V$ simply queries $\textsf{O}^{SQ}_V(D)$ with the indicator function $\phi(x,y) = \mathds{1}_{\mathrm{par}_s\neq y}$ and with the desired inverse polynomial tolerance $\tau$. 
    Now, assume $V$ can solve the agnostic parity problem, it then does the following to solve the parity testing task when given access to $\textsf{O}_V^{SQ}(D)$ for some $D$:
    \begin{enumerate}
        \item Solve the agnostic parity learning problem defined by $\textsf{O}_V^{SQ}(D)$, with accuracy parameter $\epsilon=\nicefrac{1}{4}$, and obtain a valid solution $s$.
        \item Evaluate $\ell (D_s,D)$, up to accuracy $\tau = \frac{1}{8}$, via a single query to $\textsf{O}_V^{SQ}(D)$. Call the obtained estimate $\hat{\ell}(D_s,D)$.
        \item If $\hat{\ell}(D_s,D)\in [\nicefrac{3}{8},\nicefrac{5}{8}]$, then output ``accept'' (i.e.,  $D$ is the uniform distribution). Else, if $\hat{\ell}(D_s,D)\leq \nicefrac{1}{8}$, output ``reject'' (i.e.,  $D$ is a parity distribution).
    \end{enumerate}
    To see that this works, note that if the unknown distribution $D$ is a parity distribution -- i.e.,  $D = D_{s'}$ for some $s'\in\{0,1\}^n$ -- then by the agnostic learning condition Eq.~\eqref{eq:agnostic_criteria}, and the fact that $\epsilon=\nicefrac{1}{4}$, we know that $\ell (D_s,D_{s'}) \leq \nicefrac{1}{4}$. But, parity functions have the property that 
    \begin{equation}
        \ell (D_s,D_{s'}) = \begin{cases} \nicefrac{1}{2} \text{ if } s\neq s', \\ 0 \text{ if }s = s'.
        \end{cases} \, 
    \end{equation}
    As a result,  $\ell (D_s,D_{s'}) \leq \nicefrac{1}{4}$ implies that $s=s'$ (and therefore that $\ell (D_s,D) = 0$). This further implies that the estimation obtained in Step 2 above will satisfy $\hat{\ell}(D_s,D)\leq \nicefrac{1}{8}$, and therefore the output in Step 3 will be correct. On the other hand, if the unknown distribution is the uniform distribution -- i.e.,  $D = \mathcal{U}_{n+1}$ -- then for all $s\in\{0,1\}^n$ we have $\ell (D_s,\mathcal{U}) = \nicefrac{1}{2}$. Consequently, $\hat{\ell}(D_s,D)$ will satisfy $\hat{\ell}(D_s,D)\in [\nicefrac{3}{8},\nicefrac{5}{8}]$ and again the output in Step 3 will be correct.  
    
    As we already have from \Cref{cor:no-advantage-delegating-stronger-oracle} an SQ query complexity lower bound for the verifier, for any interactive proof for parity testing, the stated result then follows from the fact that a learner that can solve agnostic parity learning can solve parity testing with a single extra statistical query.
\end{proof}

lower bound the query complexity of the constrained QSQ verifier is by reducing the learning problem to a many-vs-one testing problem, and applying the machinery we developed for lower bounds in that setting. With this in mind, the applicability of our framework to placing limitations on what a QSQ verifier can achieve for any specific learning problem (such as QAC0 learning) hinges on whether or not QSQ lower bounds for this learning problem can be proven via a lower bound for a related many-vs-one testing problem

\ins{We note that the logic of the above proof consists of reducing the learning problem to a many-vs-one testing problem, and then applying the machinery we developed for lower bounding the query complexity of oracle-constrained verifiers in that setting. As such, we can use similar reasoning to place limitations on the query complexity of oracle-constrained verifiers in an interactive proof for any learning problem for which lower bounds can be proven via a reduction to a many-vs-one testing problem. As concrete examples, this logic can be used to} 
show that verifiers with restricted oracles cannot gain any advantages in query complexity via interactive proofs for delegation for any of the following problems: 

\begin{enumerate}
    \item Learning Clifford circuits when $V$ has quantum process statistical query access~\cite{WadhwaDoosti} to the unknown circuit and $P$ has query access.
    \item Pauli shadow tomography when $V$ has quantum statistical query access to the unknown state and $P$ has access to copies~\cite{arunachalam2023role}. 
    \item Learning displacement amplitudes when the $P$ has access to copies of the unknown state and its conjugate state, but $V$ only has access to copies of the state~\cite{king2024exponential}.
\end{enumerate}

\subsection{Implications for NISQ verifiers}\label{sec:no-go-nisq}

For our final example application of \Cref{theorem:no-go-interactive-many-vs-one-distinguishing-intro-version},  
we consider the problem of Pauli shadow tomography -- as defined in Problem~\ref{problem:pauli_shadow_tomography} and discussed in Example~\ref{example:Pauli-shadow-tomography} -- for the case in which a $\mathsf{NISQ}$ verifier interacts with a $\mathsf{BQP}$ prover~\cite{chen2023complexity}. 
Here, we follow the definitions of \cite{chen2023complexity} and consider $\mathsf{NISQ}_\lambda$ algorithms, where $\lambda$ denotes the noise strength of single-qubit depolarizing noise.
Our next result puts strong limitations on how much a $\mathsf{NISQ}$ verifier can gain from interacting with a $\mathsf{BQP}$ prover. 

\begin{corollary}[Limited advantage to delegating Pauli shadow tomography from $\mathsf{NISQ}$ to $\mathsf{BQP}$]\label{cor:no_advantage_PST_nisq}
    Any $\mathsf{NISQ}_\lambda$ verifier that interacts with an untrusted $\mathsf{BQP}$ prover via a classical or quantum communication channel to solve the Pauli state tomography problem to accuracy $1/3$ with success probability $\geq 2/3$ has to use at least $\Omega ((1-\lambda)^{-n})$ many copies of the unknown state.
\end{corollary}

We note that Ref.~\cite{huang2022foundationslearningnoisyquantum} exhibited a $\mathsf{NISQ}_\lambda$ algorithm that solves the absolute value version of Pauli shadow tomography from $(1-\lambda)^{-\Theta(n)}$ copies, which was proved to be qualitatively optimal in \cite{chen2023complexity}. Thus, \Cref{cor:no_advantage_PST_nisq} shows that the interaction with an untrusted $\mathsf{BQP}$ prover cannot help a $\mathsf{NISQ}$ verifier significantly improve 
upon the query complexity that they could have achieved in isolation, and they in particular still face an exponential lower bound.

\begin{proof} 
    We start by noting that in Ref.~\cite{chen2023complexity}, it is assumed that in the NISQ model, when a $\mathsf{NISQ}_\lambda$ algorithm queries a copy oracle $\mathsf{O}(\rho)$ for an unknown state, it \textit{does not} receive a perfect copy of the unknown state $\rho$, but rather a copy with a local depolarizing channel of strength $\lambda$ \textit{already applied}. Alternatively, one can imagine this depolarizing channel occurs immediately as any $\mathsf{NISQ}_\lambda$ algorithm stores the unknown perfect state $\rho$ in its quantum register, and that a natural way to formalize this is by replacing the $\mathsf{NISQ}_\lambda$ algorithm's access to the perfect copy oracle $\mathsf{O}(\rho)$ with a noisy copy oracle~$\mathsf{O}_\lambda(\rho)$, which when queried provides noisy copies of $\rho$ (corrupted by a single application of local depolarizing channel of strength $\lambda$).
    
    Given this definitional aspect of the $\mathsf{NISQ}_\lambda$ model, in order to prove the statement of the Corollary it is sufficient to prove that any $\mathsf{NISQ}_\lambda$ verifier with access to a \textit{noisy} copy oracle $\mathsf{O}_V = \mathsf{O}_\lambda$, which interacts with an untrusted $\mathsf{BQP}$ (or even arbitrarily powerful quantum) prover and solves the Pauli shadow tomography problem to accuracy $1/3$ with success probability $\geq 2/3$, has to use at least $\Omega ((1-\lambda)^{-n})$ queries. We will do this by proving the \textit{stronger} statement for a $\mathsf{BQP}$  (or even arbitrarily powerful quantum) verifier with access to a noisy copy oracle -- i.e., we will prove the statement for the more challenging setting in which the only constraint on the verifier is the noisy copy oracle. To this end, we start by noting that any algorithm that solves the Pauli shadow tomography problem to accuracy $1/3$ can immediately distinguish between the maximally mixed state and states of the form $\frac{1}{2^n}(\mathds{1}_2^{\otimes n} + P)$ for a non-identity $n$-qubit Pauli $P$, a simplified version of state Pauli spike detection. Thus, in the language introduced around \Cref{eq:learning_to_testing}, we 
    have an underlying ``easier-than-learning'' many-vs-one distinguishing task with $m_{l\to d}=0$, and to prove our desired statement it is sufficient to prove that that any $\mathsf{BQP}$ verifier, with access to a \textit{noisy} copy oracle $\mathsf{O}_V = \mathsf{O}_\lambda$, which interacts with an untrusted $\mathsf{BQP}$ prover, who has access to a noiseless copy oracle, and solves the many-vs-one distinguishing problem above has to use at least $\Omega ((1-\lambda)^{-n})$ queries. To do this, we can use Corollary~\ref{cor:no-advantage-delegating-stronger-oracle}, which applies for many-vs-one distinguishing problems when the only difference between prover and verifier is the strength of the oracle access. All that is needed is an appropriate lower bound on the many-vs-one distinguishing task that holds for any $\mathsf{BQP}$ algorithm with access to the noisy copy oracle $\mathsf{O}_\lambda$. Precisely such a lower bound is proven in \cite[Lemma F.3]{chen2023complexity}, where it is shown that any $\mathsf{BQP}$ algorithm requires $\Omega ((1-\lambda)^{-n})$ many queries to a noisy copy oracle $\mathsf{O}_\lambda$ to solve this many-vs-one distinguishing task.
\end{proof}

As seen in the proof, because quantum state oracle access in the $\mathsf{NISQ}$ definition of \cite{chen2023complexity} is phrased in terms of a noisy oracle access, \Cref{cor:no_advantage_PST_nisq} can be viewed as a consequence of \Cref{cor:no-advantage-delegating-stronger-oracle} about delegation of quantum testing and learning under an oracle mismatch.
In particular, the proof exemplifies the following observation: 
Whenever one can prove a query complexity lower bound for an interactive quantum testing or learning proof in which both the verifier and prover are $\mathsf{BQP}$ (or even arbitrarily powerful quantum) machines, but the verifier has only a noisy copy oracle $\mathsf{O}_\lambda$, then this lower bound directly carries over to IPs with a $\mathsf{NISQ}_\lambda$ verifier and a $\mathsf{BQP}$ prover, even with quantum communication.

\section{Power of interactive proofs for verifying quantum learning and testing}\label{s:go_results}

To apply \Cref{theorem:no-go-interactive-many-vs-one-distinguishing-intro-version}, the actions of the honest prover $P$ with access to $\textsf{O}_P(x_A)$, interacting with $V$ with access to $O_V(x)$, must be simulatable by an algorithm with the same resource constraints as $V$. While the many examples in the previous section meet this criterion, this section explores cases where meaningful interactive protocols under resource constraints can still be constructed.
To gain intuition for the sort of settings in which the assumption of \Cref{theorem:no-go-interactive-many-vs-one-distinguishing-intro-version} is violated, we consider the following two examples:

\begin{example}[Interactive uniformity testing with memory constraints] 
    Consider the uniformity testing task from \Cref{example:uniformity_testing}, for distributions with support size $k$. 
    Assume that the verifier only has $\log(k)$ bits of memory. 
    Hence, the verifier can only store a single sample from the unknown distribution. 
    Next, imagine an interactive proof protocol in which $V$ repeatedly obtains a single sample from its oracle \(\mathsf{O}_V\), sends it to the memory-unconstrained prover $P$, and clears its memory. 
    Assume, that the honest prover stores all the samples sent by $V$, for example to calculate some statistic (such as the number of collisions) from which one can decide uniformity. 
    Then, an algorithm with $V$'s memory constraint \textit{cannot} simulate the actions of \(P\) interacting with \(V\), even though it could simulate the fixed accept oracle $\textsf{O}_P(x_A = \mathcal{U})$. 
    As such, the no-go result from Theorem~\ref{theorem:no-go-interactive-many-vs-one-distinguishing-intro-version} does not immediately apply. Indeed in Appendix~\ref{app:uniformity_ip} we show that one \textit{can} in fact construct an interactive proof protocol for memory-constrained uniformity testing, which allows a memory-constrained verifier to surpass memory imposed query complexity lower bound.
\end{example}

\begin{example}[Interactive purity testing with memory constraints and quantum communication]\label{example:purity_testing_quantum_comm} 
    Consider the purity testing task from \Cref{example:purity-testing-informal}, but with a memory-constrained quantum verifier that can only store a single copy of the unknown quantum state, and that interacts via quantum communication with a memory-unconstrained prover. 
    Similar to the previous example, 
    one can set up 
    an interactive proof protocol in which $V$ repeatedly obtains a single quantum state from its oracle, sends it to the memory-unconstrained prover $P$, and clears its memory. 
    The honest prover could store all the samples sent by $V$, 
    to perform a multi-copy coherent measurement on copies sent by $V$ and copies obtained from its own oracle. 
    Then, an algorithm with a memory constraint 
    like that of $V$
    \textit{cannot} simulate the actions of the honest prover when interacting, even though it could simulate the fixed accept oracle $\textsf{O}_P(x_A = \nicefrac{\mathds{1}_2^{\otimes n}}{2^n})$. 
\end{example}

The above examples can be generalized to a wide variety of testing/learning problems in which a memory-constrained verifier interacts with a memory-unconstrained prover. 
In \Cref{ss:proofs_from_VBQC} we show how  \Cref{example:purity_testing_quantum_comm} can be generalized to other interactive proofs for a quantum testing/learning problem in which a memory-constrained verifier interacts with a memory-unconstrained prover via \textit{quantum communication}. 
For such scenarios, we will show that the verifier can gain from the interaction with the prover.
This is in stark contrast with the results of the previous section:
\Cref{theorem:no-go-interactive-many-vs-one-distinguishing-intro-version} 
prohibits 
interactive proofs for many quantum testing problems with quantum-memory-constrained verifiers, when 
interacting via \textit{classical communication} (at least from $V$ to $P$).

\subsection{Interactive proofs for quantum learning/testing via verified blind quantum computing}\label{ss:proofs_from_VBQC}

A natural method for developing interactive proofs for quantum testing and learning, in the presence of a resource mismatch, is to employ \textit{verified blind quantum computing} (VBQC) protocols~\cite{gheorghiu2018verification, fitzsimonsUnconditionallyVerifiableBlind2017}.
VBQC protocols allow a client to securely delegate quantum computations to a remote server, ensuring the server learns nothing about the computations while the client can verify their correctness. Of key importance for this work are the existence of universal VBQC schemes (such as the one from~\cite{fitzsimonsUnconditionallyVerifiableBlind2017}) in which:
\begin{enumerate}
\item The client does not need to know the state of the input qubits to the computation. In particular, they could come from a third party, or from oracle queries.
\item The client does not need to be able to store the entire input state to the computation. In particular, the client can receive the input state one qubit at a time from a third party or oracle, each time operating on the qubit and sending it to the prover, before receiving the next qubit.
\end{enumerate}
At a high level, for any VBQC protocol which satisfies the above conditions, we can obtain an interactive proof for a learning or testing problem by letting the memory-constrained verifier delegate the desired coherent multi-copy learning/testing algorithm to the prover. The verifier can provide the (unknown) multi-copy input state to the prover by successively querying its oracle, sending the state to the prover, and clearing its memory. Observation~\ref{lem:interactive-quantum-proofs-vbqc} below summarizes the interactive protocol one obtains via this approach, by using the universal VBQC protocol from Ref.~\cite{fitzsimonsUnconditionallyVerifiableBlind2017}. We prove this observation in Appendix~\ref{app:VBQC}.

\labelRestate{lem:interactive-quantum-proofs-vbqc}
\obsgenericVBQC*

\begin{remark} We stress that one could obtain generic interactive proof protocols for learning and testing with memory constraints, from \textit{any} universal VBQC protocol which allows the client to delegate computations on unknown states, and whose memory requirements are within the constraints of the client (See Ref.~\cite{gheorghiu2018verification} for a comprehensive review of VBQC protocols). However, the complexity of the interactive proof -- i.e. the query complexity for the verifier, and computational and communication complexities of both client and verifier -- will depend on the details of the VBQC protocol, and may differ from the those given in Observation~\ref{lem:interactive-quantum-proofs-vbqc}, which are derived from using the specific MBQC-based universal VBQC scheme from Ref.~\cite{fitzsimonsUnconditionallyVerifiableBlind2017}.
\end{remark}

We also note that VBQC protocols are typically designed for delegating \textit{universal} quantum computations. However, many testing and learning tasks can be addressed with specialized algorithms. Adapting VBQC protocols to these specific algorithms can simplify interactive proofs and reduce resource requirements. To demonstrate this point, consider \textit{Bell sampling} \cite{montanaro2017learning,hangleiter2024bellsamplingquantumcircuits}, a popular 2-copy measurement subroutine used in quantum algorithms for tasks like purity testing, stabilizer state learning, and Pauli shadow tomography \cite{huang2021information, king2024triply, chen2024optimalpauli,montanaro2017learning}.
Consider a verifier $V$ who is restricted to incoherent single-copy measurements and who would like to solve one of these tasks, say purity testing, via Bell sampling. When interacting with $P$, they only need to delegate Bell sampling across 2 copies of their state, which, only requires a Clifford computation on these two copies followed by a computational basis measurement. The Clifford nature of the target computation allows for a simplified protocol \cite{gheorghiu2018verification}.

The VBQC approach to interactive proofs for quantum testing and learning has the advantage of working almost out-of-the box. However, it also comes with two drawbacks.
\begin{enumerate}
    \item First, to make use of VBQC protocols in our context, $V$ needs to know how to solve the problem at hand if they had coherent multi-copy quantum processing. This is because VBQC protocols allow the delegation of a specific desired quantum computation, say the execution of some quantum circuit, and the protocol depends on that circuit.  
    \item Second, in such VBQC-based IPs for testing/learning, all the quantum data comes from $\mathsf{O}_V$, the oracle $\mathsf{O}_P$ is not queried at all. So, while $V$ can beat the single-copy lower bound through interaction with $P$ this way, it can never surpass the multi-copy lower bound.
\end{enumerate}

We now describe ways of alleviating these drawbacks.

\subsection{Black-box interactive proofs for quantum learning/testing}\label{ss:black_box}

The first drawback becomes relevant when $V$ does not only lack the quantum memory resources that $P$ has access to, but also lacks the expertise required to use those resources for testing or learning. In other words, even if $V$ had the same resources as $P$, $V$ might not know how to successfully employ them on the available data. Thus, alleviating the first drawback will lead to interactive proof protocols that allow the verifier to delegate the resource-intensive quantum processing of data as well as the challenge of identifying the correct quantum processing to apply.
Therefore, in removing 
the first drawback, we aim to construct protocols in which the only information that $V$ needs is: 
\begin{enumerate}
    \item The description of the problem -- e.g., the tuple $(x_A,\mathcal{X}_R,\mathcal{X})$ for a many-vs-one distinguishing problem.
    \item A promise that the honest prover can solve the problem given $m$ copies of the unknown quantum state.
\end{enumerate}
In other words, we want $V$ to be able to treat the prover as a ``problem-solving black box''.
As such, we call an IP protocol a black-box IP protocol whenever $V$ is only required to know the information listed above, and not a description of the quantum algorithm for solving the problem. We give a concrete example of such a protocol below, for the problem of 
purity testing, as defined in Problem~\ref{prob:purity_testing} and discussed in Example~\ref{example:purity-testing-informal}. 

\begin{theorem}[Black-box Interactive proofs for purity testing (formal version of \Cref{theorem:interactive-purity-testing-black-box-intro-version})]\label{theorem:interactive-purity-testing-black-box-formal-version}
    There exists a black-box interactive proof system $(V,P)$ for solving purity testing with success probability $1-\delta$, between an incoherent single-copy $V$ and a coherent multi-copy $P$, efficiently communicating via a quantum channel, such that $V$ and $P$ both are computationally efficient and such that $V$ uses at most $\mathcal{O}(\log(1/\delta))$ copies.
    More precisely, the complexities are bounded as follows:
    \begin{itemize}
        \item $V$ uses $\mathcal{O}(\log(1/\delta))$ oracle queries and $\tilde{\mathcal{O}}\left(\log\log(d)\log^2(1/\delta)\right)$ classical and quantum computation time.
        \item $P$ uses no oracle queries and $\tilde{\mathcal{O}}(\log(1/\delta))$ classical and quantum computation time.
        \item The total communication consists of $\tilde{\mathcal{O}}(\log(1/\delta))$ qudits sent from $V$ to $P$ and $\mathcal{O}(\log(1/\delta))$ classical bits sent from $P$ to $V$.
    \end{itemize}
\end{theorem}

As shown in Table~\ref{table:advantages-via-IPs}, we highlight that the interactive proof of \Cref{theorem:interactive-purity-testing-black-box-formal-version} allows the single-copy incoherent verifier $V$ to significantly surpass the $\Omega(2^{\nicefrac{n}{2}})$ query complexity lower bound that it would be subject to when trying to solve the problem on its own. 
Next, we give the proof of \Cref{theorem:interactive-purity-testing-black-box-formal-version}. 
Our interactive proof system will make use of a common construction in interactive proofs, in which the 
verifier alternates randomly between different types of interaction rounds, some of which are ``test'' rounds designed to detect cheating provers, and some of which are ``computation'' rounds designed to delegate the desired computation to the prover.  The key property required in the design of these different rounds is that they should all be indistinguishable from the perspective of the prover. We refer to Ref.~\cite{gheorghiu2018verification} for more examples and detailed discussion of such strategies.

\begin{proof}[\Cref{theorem:interactive-purity-testing-black-box-formal-version}]
    We begin with a high-level sketch of our interactive protocol. This will be followed by a more detailed description in pseudo-code, before we then prove completeness and soundness, and do a complexity analysis.

    In our protocol, the verifier $V$ alternates uniformly at random between compute and test rounds. In each round, $V$ provides the prover with $m$ copies of a quantum state, where $V$ knows that $m$ copies are sufficient for the prover to solve the purity testing problem (with a suitable success probability), and asks the prover to solve the purity testing task, and to send over a single bit corresponding to the solution of the distinguishing task. 
    The rounds differ in which state is sent by the verifier as follows:
    \begin{itemize}
        \item Computation round: The verifier sends $m$ copies of $U\rho U^\dagger$.
        \item Pure test round: The verifier sends $U \ket{0^n}\bra{0^n} U^\dagger$.
        \item Mixed test round: The verifier sends  the maximally mixed state $\mathds{1}/2^n$.
    \end{itemize}
    Here, $U$ is a random unitary drawn from a unitary $m$-design and $\rho$ is the unknown state obtained from the verifier oracle $\mathsf{O}_V$. For the test rounds, the verifier knows whether the state they sent is pure or mixed, respectively, and hence they can check the correctness of the prover's answer and catch cheating attempts. However, from the perspective of the prover, the computation rounds are statistically indistinguishable from either the pure or mixed test round (depending on whether the unknown $\rho$ is pure or maximally mixed). Hence, if the number of rounds is sufficiently large, then by the indistinguishability of the rounds, any prover that attempts to cheat a relevant fraction of the time will get caught on a non-negligible fraction of test rounds.
    Additionally, $V$ will reject if the prover answers the computation rounds inconsistently, which means that provers that want to maliciously influence $V$'s outcome do indeed need to cheat often.
    
    We give pseudo-code for the interaction between the verifier $V$ and the honest prover $P$ in \Cref{alg:black-box-purity-testing}, written down for a general $\delta\in (0,1)$.
    The proof now consists in showing that this IP protocol is complete and sound, and that it has the claimed complexities.

    \begin{algorithm}[htp!]
        \caption{IP for purity testing with quantum communication}\label{alg:black-box-purity-testing}
            \begin{algorithmic}[1]
                \Require Verifier copy oracle $\mathsf{O}_V(\rho)$;  
                confidence parameter $\delta\in (0,1)$; function $m:(0,1)\to\mathbb{N}$ such that $P$ can solve the task with success probability $\geq 1-\eta$ from $m(\eta)$ copies.
                \Ensure ``maximally mixed'' or ``pure'' or ``abort''
                \State Set $N=\left\lceil\max\left\{72\ln(6/\delta), 4 \log_2\left(2/\delta\right)\right\}\right\rceil$, $\tilde{\delta} = \nicefrac{\delta}{2N}$, and $m=m(\tilde{\delta})$. 
                \For{$1\leq i\leq N$}
                    \State $V$ picks $s_i\in \{m,p,c\}$ uniformly at random.
                    \If{$s_i=m$}\Comment{Mixed test round}
                        \State One by one, $V$ prepares $m$ copies of $\frac{\mathds{1}}{d}$ and sends them to $P$.
                    \ElsIf{$s_i=p$} \Comment{Pure test round}
                        \State $V$ draws a unitary $U_i$ (privately) at random from a qudit unitary $m$-design.
                        \State One by one, $V$ prepares $m$ copies of $U_i\ket{0}\bra{0}U_i^\dagger$ and sends them to $P$.
                    \ElsIf{$s_i=c$} \Comment{Compute round}
                        \State $V$ draws a unitary $U_i$ (privately) at random from a qudit unitary $m$-design.
                        \State One by one, $V$ queries $\mathsf{O}_V(\rho)$ to obtain $m$ copies of $\rho$, then applies $U_i$, and sends $U_i\rho U_i^\dagger$ to $P$.
                    \EndIf
                    \State $V$ asks the prover to distinguish, with success probability $\geq 1-\Tilde{\delta}$, whether the $m$ copies that they received from $V$ came from a maximally mixed or a pure state, and to send single bit $b_i$ encoding the answer (``maximally mixed'' or ``pure'') back to $V$.
                    \If{$s_i=m$}
                        \State If $b_i=$``maximally mixed'', $V$ sets $p_i=1$. If $b_i=$``pure'', $V$ sets $p_i=0$.
                    \ElsIf{$s_i=p$}
                        \State If $b_i=$``pure'', $V$ sets $p_i=1$. If $b_i=$``maximally mixed'', $V$ sets $p_i=0$.
                    \ElsIf{$s_i=c$}
                        \State $V$ stores $b_i$ in classical memory.
                    \EndIf
                \EndFor
                \If{$\exists 1\leq i\leq N$ s.t.~$s_i\in\{m,p\}$ and $p_i=0$} \Comment{Check that the prover passed all tests}
                    \State $V$ outputs ``abort''.
                \ElsIf{$s_i=c=s_j \not\Rightarrow b_i=b_j$} \Comment{Check that the prover was consistent in all compute rounds}
                    \State $V$ outputs ``abort''.
                \Else
                    \State $V$ outputs $b_i$ for some $1\leq i\leq N$ with $s_i=c$.
                \EndIf
            \end{algorithmic}
    \end{algorithm}

    \noindent\textbf{Completeness:}
    In any single round, the honest prover $P$ by assumption succeeds with probability $\geq 1-\tilde{\delta}$ in Step~13. Thus, the probability that there is a round in which $P$ fails is at most $N\Tilde{\delta}= \delta/2 < \delta$.
    Conditioned on this event of $P$ succeeding in every round, by definition of how $V$ acts in Steps 22-28, $V$ does not abort the interaction with $P$, and $V$ outputs a valid solution to the purity testing problem.
    It remains to arguet about the quantum time and sample efficiency of the honest prover. To this end, observe that (e.g., via SWAP tests \cite{buhrman2001quantum}), $P$ can use multi-copy measurements to efficiently solve the purity testing problem with success probability $\geq 1-\eta$ from $\mathcal{O}(\log(1/\eta))$ copies, so $m(\eta)=C\log(1/\eta)$ for a sufficiently large constant $C>0$ suffices.
    This shows completeness.

    \noindent\textbf{Soundness:}
    First, we observe that, with high probability, every type of round appears at least $N/4$ many times, by a simple concentration argument. More precisely, we find
    \begin{align}
        \Pr [\exists s\in\{m,p,c\}: ~|\{1\leq i\leq N~|~s_i\neq s\}|< \nicefrac{N}{4}]
        &\leq \sum_{s=1}^3 \Pr[|\{1\leq i\leq N~|~s_i\neq s\}|< \nicefrac{N}{4}]\\
        \nonumber
        &= 3 \Pr[\mathrm{Binom}(N, 1/3)< \nicefrac{N}{4}]\\
        \nonumber
        &\leq 3 \Pr[\nicefrac{N}{3} - \mathrm{Binom}(N, 1/3) > \nicefrac{N}{12}]\\
        \nonumber
        &\leq 3\exp\left( - \frac{2\cdot (N/12)^2}{\sum_{k=1}^N (1-0)^2}\right)\\
        \nonumber
        &= 3\exp\left(-N/72\right)\\
        \nonumber
        &\leq \frac{\delta}{2}\, ,
    \end{align}
    where we have used a union bound followed by Hoeffding's inequality as well as our choice of $N$.
    For the remainder of the soundness proof, we condition on the high-probability event $E=\{\forall s\in\{m,p,c\}: ~|\{1\leq i\leq N~|~s_i\neq s\}|\geq  \nicefrac{N}{4}\}$.
    We want to bound the probability of $V$ incorrectly accepting the interaction with a dishonest prover $P'$ and outputting an invalid solution to the purity testing problem. That is, by the decision rule that $V$ uses, we want to bound
    \begin{align*}
            \Pr\Big[P'&\mathrm{~answers~all~rounds~with~}s_i=c\mathrm{~incorrectly}\\
             ~&\wedge~ P'\mathrm{~answers~all~rounds~with~}s_i\in\{m,p\}\mathrm{~correctly}  ~\Big|~ E\Big]\, .
    \end{align*}
    To do so, notice the following: If $\rho=\frac{\mathds{1}}{2^n}$, then the state preparation procedures for $s_i=m$ and for $s_i=c$ in fact prepare copies of the same quantum state.
    Therefore, in this case we may upper bound our probability of interest as 
    \small
    \begin{align*}
        &\Pr\Big[P'\mathrm{~answers~all~rounds~with~}s_i=c\mathrm{~incorrectly} \\
        &\hphantom{\Pr\Big[P'}~\wedge~ P'\mathrm{~answers~all~rounds~with~}s_i\in\{m,p\}\mathrm{~correctly}  ~\Big|~ E\Big]\\
        &\leq \Pr\Big[P'\mathrm{~answers~all~rounds~with~}s_i=c\mathrm{~incorrectly}\\
        &\hphantom{\Pr\Big[P'}~\wedge~ P'\mathrm{~answers~all~rounds~with~}s_i=m\mathrm{~correctly}  ~\Big|~ E\Big]\\
        &\leq \Pr\Big[\mathrm{in~every~round~with~}s_i\in\{m,c\},~P'\mathrm{correctly~distinguishes~whether~}s_i=m\mathrm{~or~}s_i=c~|~ E\Big]\, .
    \end{align*}
    \normalsize
    Now, because even an unbounded quantum prover $P'$ cannot correctly distinguish between $s_i=m$ and $s_i=c$ with probability strictly bigger than $1/2$, this final probability can be upper bounded by the probability of guessing a single correct assignment out of $\binom{N/2}{N/4}$ many possible assignments uniformly at random.
    Using $\binom{n}{k} \geq (\nicefrac{n}{k})^k$, we see that $\binom{N/2}{N/4}\geq 2^{N/4}$, thus the relevant probability is at most $2^{-N/4}\leq\delta/2$ by our choice of $N$.  
    
    Similarly, if $\rho=\ket{\psi}\bra{\psi}$ is a pure state, then, since each $U_i$ is drawn independently from a unitary $m$-design, the state preparation procedures for $s_i=p$ and for $s_i=c$ are indistinguishable to any (even unbounded) $P'$ from $m$ copies per round.
    So, we can argue as before and bound the relevant probability (conditioned on $E$) by $\delta/2$.
    As the bound is the same in both cases, we conclude that 
    \begin{align*}
        \Pr\Big[P'&\mathrm{~answers~all~rounds~with~}s_i=c\mathrm{~incorrectly} \\
        &~\wedge~ P'\mathrm{~answers~all~rounds~with~}s_i\in\{m,p\}\mathrm{~correctly}  ~\Big|~ E\Big]
        \leq \frac{\delta}{2}\, .
    \end{align*}
    Applying Bayes' rule and a union bound, we have proved the desired soundness.

    \noindent\textbf{Complexity analysis:} The verifier queries its oracle at most $N=\mathcal{O}(\log(1/\delta))$ many times.
    The total amount of quantum computation time on the verifier side is bounded by $Nm=\tilde{\mathcal{O}}(\log(1/\delta))$ times the worst-case quantum computation for the state preparations in Steps 5, 8, and 11. 
    Step 5 is clearly efficient.
    For Steps 8 and 11, we note that the above proof goes through when replacing the exact $m$-designs in Steps 8 and 11 by approximate $m$-designs with a small enough constant approximation parameter. (The probability of $V$ accepting and producing a false output will still decay exponentially in $N$, only the base will change from $1/2$ to $1/2 + \Delta$ if we consider a $\Delta$-approximate $m$-design.) 
    Assuming the unknown state to be an $n$-qubit state, i.e. $d=2^n$, such approximate $m$-designs can be implemented in depth $\tilde{\mathcal{O}}\left(\log (n) \log(1/\delta)\right)$ \cite{schuster2024randomunitariesextremelylow}. Thus, the overall quantum time complexity of $V$ (measured in terms of number of two-qubit gates) is bounded by $\tilde{\mathcal{O}}\left(\log (n) \log^2(1/\delta)\right)$.
    The classical time complexity of $V$ is dominated by checking the conditions in Steps 22 and 24, which can be done in time $\mathcal{O}(N)=\mathcal{O}(\log(1/\delta))$.
    
    The honest prover never queries its oracle.    
    The total amount of quantum computation time on the prover side is given by $N=\mathcal{O}(\log(1/\delta))$ times the quantum computation time of Step 13 in \Cref{alg:black-box-purity-testing}. As this step can, e.g., be realized by executing $m/2=\tilde{\mathcal{O}}(\log(1/\delta))$ many SWAP tests and checking whether one of them rejects, this is clearly quantumly 
    efficient. 
    The total amount of quantum communication is given by $Nm=\tilde{\mathcal{O}}(\log(1/\delta))$ qudits sent from $V$ to $P$.
    The total amount of classical communication is given by $N=\mathcal{O}(\log(1/\delta))$ classical bits sent from $P$ to $V$.
    Fixing $\delta=1/3$ then gives the complexity and communications bounds from \Cref{theorem:interactive-purity-testing-black-box-intro-version}.
\end{proof}

Let us make two comments about the protocol given above.
First, in \Cref{alg:black-box-purity-testing}, the number of rounds of each type is a random variable and, by a standard concentration argument, we argued that, with high probability, each type of round appears sufficiently often. Alternatively, one may consider a verifier that plays each round the same (deterministic) number of times, but randomizes the order in which rounds are played. The resulting protocol will lead to an interactive proof with the the same complexity scalings as in \Cref{theorem:interactive-purity-testing-black-box-formal-version}. The proof becomes simpler (since the aforementioned concentration argument is no longer required), but we chose the presentation above to avoid notational clutter.

Second, in \Cref{alg:black-box-purity-testing}, we have chosen a prover that uses $m$ registers of quantum memory simultaneously merely for ease of presentation.
However, as there is an honest prover whose processing in Step 13 consists in separate two-copy SWAP tests followed by classical post-processing, one can straightforwardly adapt the above procedure so that $P$ only ever uses two registers of quantum memory (plus non-memory auxiliary quantum registers for the SWAP tests). 
In fact, this way one can obtain a protocol in which, in the spirit of the quantum one-time pad \cite{mosca2000privatequantumchannelscost}, the verifier only uses random single-copy Paulis (which form a $1$-design) rather than an $m$-design.

From the above algorithm and proof, it is be clear that the ideas behind the interactive proof protocol, namely the alternation between indistinguishable test rounds for soundness and computational rounds for completeness, could in principle be applied to other problems. 
To apply such an approach, one needs to (a) design suitable tests for detecting dishonest provers, and (b) develop an efficient ``hiding strategy'' that $V$ can use to make the test and computation rounds indistinguishable to the prover. 
In particular, we can use these ideas to construct black-box IPs for unitarity testing (\Cref{prob:unitarity_testing}), stabilizer testing (\Cref{prob:stabilizer-testing}), and Clifford testing (\Cref{prob:clifford_testing}).
Instead of giving full proofs, let us sketch the tests and hiding strategies for these cases:
\begin{itemize}
    \item Unitarity testing can be solved with multi-copy measurements through purity testing on copies of the Choi state. To prepare one such copy in an IP protocol, $P$ sends half of a maximally entangled state to $V$; $V$ then either applies the maximally depolarizing channel (test round), or applies the identity channel pre- and post-processed by a unitary channel drawn from an (approximate in relative-error) design (test round), or queries the unknown channel pre- and post-processed by a unitary channel drawn from an (approximate in relative-error) design (computation round); afterwards, $V$ sends the quantum system back to $P$.
    \item Stabilizer testing can be solved with multi-copy measurements using Bell (difference) sampling \cite{gross2021schur, grewal2024improved}. Beyond the change from SWAP testing to Bell sampling, the main difference in an IP protocol for stabilizer testing compared to the purity testing protocol is that it suffices to use random Cliffords instead of unitaries drawn from a design.
    \item Clifford testing can be solved through stabilizer testing on copies of the Choi state. Thus, an IP protocol for Clifford testing proceeds similarly to that sketched above for unitarity testing, only that the verifier uses random Cliffords instead of unitaries drawn from a design.
\end{itemize}

\subsection{Interactive proofs for outperforming coherent multi-copy provers}\label{ss:outperforming_provers}

In the interactive proof systems discussed above, the incoherent single-copy verifier $V$ makes use of the coherent multi-copy prover $P$ to overcome single-copy quantum query complexity lower bounds. However, as $P$'s oracle is never used, $V$ has to provide all the quantum data for $P$, and thus $V$'s quantum query complexity in all these protocols cannot outperform the multi-copy lower bound for the relevant problem. To overcome this limitation, we could instead imagine proof systems like those discussed in Section~\ref{ss:trivial} -- i.e. a proof system where $V$ delegates the entire learning algorithm to $P$, including the quantum data requirements, and then simply asks $P$ to send over the solution. For such a proof system to be sound, $V$ then only needs to use their data oracle to decide the validity of the candidate solution sent by the prover. As we have seen in Examples~\ref{example:trivial-boolean} and~\ref{example:trivial-verification-quantum}, in certain settings such as realizable learning of Boolean functions and certain sets of quantum states, such an approach is straightforward. In this section we go beyond these straightforward settings, and provide a variety of interactive proof protocols for specific problems, which allow the verifier to surpass even the query complexity requirements of an isolated prover, via validation of a claimed solution. 

\subsubsection{Interactive proofs for state tomography}\label{sss:tomography}

For our first example, we consider quantum state tomography as defined in Problem~\ref{prob:qst} and discussed in \Cref{example:quantum-state-tomography}.  More specifically, we construct an interactive proof protocol in which the verifier first asks the prover to solve the quantum state tomography problem using their own copies of the unknown quantum state, before then interacting with the prover to verify that this is indeed a valid solution. In this way, the verifier is able to use less quantum state copies then even the coherent multi-copy prover would require to solve the problem on its own.

\begin{theorem}[Interactive quantum state tomography (formal version of \Cref{theorem:interactive-state-tomography-intro-version})]
    \label{theorem:interactive-state-tomography-formal-version}
    There exists an interactive proof system $(V,P)$ for state tomography in trace distance with success probability $\geq 1-\delta$, with quantum communication between an incoherent single-copy $V$ and a coherent multi-copy $P$, such that $V$ uses at most $\mathcal{O}\left(\frac{d \log \delta^{-1}}{\varepsilon^2}\right)$ copies of the unknown \(d\)-dimensional quantum state.
    More precisely, the complexities are as follows:
    \begin{itemize}
        \item $V$ uses $\mathcal{O}\left(\frac{d \log \delta^{-1}}{\varepsilon^2}\right)$ copies.
        \item $P$ uses $\Theta\left(\frac{d^2 \log \delta^{-1}}{\varepsilon^2}\right)$ copies.
    \end{itemize}
\end{theorem}

We highlight here that, as shown in Table~\ref{table:advantages-via-IPs}, the number of copies used by $V$ in \Cref{theorem:interactive-state-tomography-formal-version} improves upon both the optimal single-copy state tomography lower bound of $\Omega(\nicefrac{d^3}{\varepsilon^2})$ \cite{chen2023whendoesadaptivity} and the optimal multi-copy state tomography lower bound of $\Omega(\nicefrac{d^2}{\varepsilon^2})$ \cite{haah2016sample, odonnell2016efficient}. 

\begin{proof}
    We again begin our proof with a high-level description of the protocol. This is then followed by pseudo-code, proofs of comlpeteness and soundness, and a complexity analysis.

    Let \(\delta_V=\delta_P=\frac{\delta}{2}\).
	The interactive protocol follows the idea outlined before the theorem statement. 
	First, $V$ asks $P$ to perform full state tomography with accuracy $0.99 \varepsilon$ and success probability $\geq 1-\delta_P$, and to send the obtained hypothesis $\hat{\rho}$ satisfying $\norm{\rho-\hat{\rho}}_1\leq 0.99 \varepsilon$. Through multi-copy measurements, $P$ can achieve this from $\Theta\left(\frac{d^2 \log \delta_P^{-1}}{\varepsilon^2}\right)$ copies of $\rho$ \cite{haah2016sample, odonnell2016efficient}.
	Then, by \cite[Theorem 1.6]{buadescu2019quantum} $\mathcal{O}\left(\frac{d \log \delta_V^{-1}}{\varepsilon^2}\right)$ copies of $\rho$ and $\hat{\rho}$, respectively, suffice to, with success probability $\geq 1-\delta_V$, distinguish the case $\norm{\rho-\hat{\rho}}_1\leq 0.99\varepsilon$ from $\norm{\rho-\hat{\rho}}_1>\varepsilon$ using multi-copy measurements and promised that one of the two cases holds. 
    Thus, by \Cref{lem:interactive-quantum-proofs-vbqc}, $V$ interacting with $P$ can solve the same decision task using their own state oracle $\mathcal{O}\left(\frac{d \log \delta_V^{-1}}{\varepsilon^2}\right)$ times (and additionally preparing $\mathcal{O}\left(\frac{d \log \delta_V^{-1}}{\varepsilon^2}\right)$ copies of $\hat{\rho}$), again with success probability $\geq 1-\delta_V$. 
	Finally, \(V\) rejects if she obtained the outcome ``$\norm{\rho-\hat{\rho}}_1>\varepsilon$'' and aborts the interaction; otherwise, she accepts the interaction and outputs the hypothesis $\hat{\rho}$.

    We give pseudo-code for this procedure in \Cref{alg:state-tomography}.
    Let us analyze the completeness and soundness of this protocol. 
    
    \textbf{Completeness:} Completeness is easy to see: If $P$ is honest, then $\norm{\rho-\hat{\rho}}_1\leq 0.99 \varepsilon$ holds with probability $\geq 1-\delta_P \geq 1 - \delta$. In this case, by the completeness of the interactive proof version of Theorem 1.6 from Ref.\ \cite{buadescu2019quantum} obtained via \Cref{lem:interactive-quantum-proofs-vbqc}, $V$ obtains the ``$\norm{\rho-\hat{\rho}}_1\leq 0.99\varepsilon$'' outcome with probability $\geq 1-\delta_V -\delta_P=1-\delta$.
    Therefore $V$ accepts and outputs $\hat{\rho}$, which indeed satisfies $\norm{\rho-\hat{\rho}}_1\leq 0.99\varepsilon\leq \varepsilon$. 

    \begin{algorithm}[htp!]
        \caption{IP for state tomography with quantum communication}\label{alg:state-tomography}
            \begin{algorithmic}[1]
                \Require Verifier copy oracle $\mathsf{O}_V(\rho)$; 
                confidence parameter $\delta\in (0,1)$; accuracy parameter $\varepsilon\in (0,1)$; function $m:(0,1)\to\mathbb{N}$ such that $P$ can solve quantum state certification to accuracy $\varepsilon$ with success probability $\geq 1-\eta$ from $m(\eta)$ copies.
                \Ensure $\hat{\rho}$ s.t. $\norm{\rho - \hat{\rho}}_1\leq\varepsilon$ or ``abort''
                \State Set $\delta_V=\delta_P=\delta/2$ and $m=m(\delta_P)$.
                \State $V$ asks the prover to perform full state tomography with accuracy $0.99 \varepsilon$ and success probability $\geq 1-\delta_P$, and to send (a classical description of) the obtained density matrix $\hat{\rho}$ back to $V$.
                \State $V$ prepares $\mathcal{O}\left(m \log(1/\delta_V)\right)$ copies of $\hat{\rho}$.
                \State $V$ obtains $\mathcal{O}\left(m \log(1/\delta_V)\right)$ copies of $\rho$ by querying $\mathsf{O}_V(\rho)$.
                \State Using the copies of $\rho$ and $\hat{\rho}$, $V$ interacts with the prover to distinguish between $\norm{\rho-\hat{\rho}}_1\leq 0.99\varepsilon$ and $\norm{\rho-\hat{\rho}}_1>\varepsilon$ with success probability $\geq 1-\delta_V$, via VBQC-based interactive proof for learning (compare \Cref{lem:interactive-quantum-proofs-vbqc}).
                \If{Step 5 produces ``$\norm{\rho-\hat{\rho}}_1>\varepsilon$''}
                    \State $V$ outputs ``abort''
                \Else
                    \State $V$ outputs $\hat{\rho}$.
                \EndIf
            \end{algorithmic}
    \end{algorithm}

    \textbf{Soundness:}
    We now turn to the soundness analysis.
    For any $P$, the probability that $V$ accepts the interaction and outputs an \emph{invalid} hypothesis state $\hat{\rho}$ is upper bounded by the probability that, in the second step, $V$ obtains the ``$\norm{\rho-\hat{\rho}}_1\leq 0.99\varepsilon$'' outcome despite $\norm{\rho-\hat{\rho}}_1>\varepsilon$. By the soundness of the interactive proof version of Theorem 1.6 from \cite{buadescu2019quantum} obtained via \Cref{lem:interactive-quantum-proofs-vbqc}, this happens with probability at most $\delta_V<\delta$.

    \textbf{Complexity analysis:}
    The complexity bounds have already been given in the initial description of the protocol, relying on \Cref{lem:interactive-quantum-proofs-vbqc} and results from \cite{haah2016sample, odonnell2016efficient,buadescu2019quantum}.
\end{proof}

We can obtain a similar result for rank-$k$ tomography. The verification protocol has the same structure as the one given above, with two minor differences: On the one hand, $P$ uses $\Theta\left(\frac{kd \log\delta^{-1}}{\varepsilon^2}\right)$ copies to perform rank-$k$ tomography \cite{odonnell2016efficient}, while $V$ makes use of the algorithm at the core of \cite[Corollary 1.5]{buadescu2019quantum} to verify from whether the hypothesis $\hat{\rho}$ sent by $P$ is close, in trace distance, to the unknown state $\rho$. By an analogous reasoning as in the above proof, $V$ can do this using $O(k\log\delta^{-1}/\epsilon^2)$ many queries to her oracle.

\subsubsection{Interactive proofs for agnostic rank-$k$ state tomography}\label{sss:rank_k}

We now consider the problem of agnostic rank-$k$ quantum state tomography with respect to the trace distance, as defined in Problem~\ref{prob:agnostic_rank_k_quantum_state}. To recap, here the problem is as follows: Given an arbitrary quantum state $\rho$, output some rank-$k$ quantum state $\sigma$ that is a close-to-optimal rank-$k$ quantum state approximation to $\rho$, with distance measured by the Schatten $1$-norm. As per the following theorem, there exists an interactive proof protocol which allows a memory constrained-verifier to use even fewer copies than are be required by a coherent multi-copy algorithm. 

\begin{theorem}[Interactive agnostic rank-$k$ state tomography in trace distance (formal version of Theorem~\ref{theorem:agnostic-rank-k-tomography})]
There exists an interactive proof system $(V,P)$ for interactive agnostic rank-$k$ quantum state tomography in trace distance with success probability $\geq 1-\delta$, with quantum communication between an incoherent single-copy $V$ and a coherent multi-copy $P$, such that $V$ uses at most \(\mathcal{O}(k^2\varepsilon^{-4}\log \delta^{-1})\) copies of the unknown state. More precisely, the complexities are as follows:
\begin{itemize}
    \item $V$ uses \(\mathcal{O}(k^2\varepsilon^{-4}\log \delta^{-1})\) copies.
    \item $P$ uses \(\mathcal{O}\left( d^{2}k^2\varepsilon^{-2}\log \delta^{-1}\right)\) copies. 
\end{itemize}
\end{theorem}

To prove the above result we will actually first construct -- in Lemma~\ref{lem:agnostic-rank-k-psd} -- an interactive protocol for the closely related problem of agnostic rank-$k$ PSD tomography, as defined in Problem~\ref{prob:agnostic_rank_k_PSD}. In particular, for this problem the learning algorithm does not have to output a normalized quantum state, but rather a rank-$k$ positive semi-definite (PSD) matrix that is a close-to-optimal rank-$k$ PSD approximation to the unknown state. To this end, we begin by establishing some required notation, as well as the technical Lemmas~\ref{lem:mirsky} and~\ref{lem:technical-agnostic-rank-k}. With this in hand we then provide the IP protocol for agnostic rank-$k$ PSD tomography in Lemma~\ref{lem:agnostic-rank-k-psd}, before finally using this to give a proof of Theorem~\ref{theorem:agnostic-rank-k-tomography}.

For a vector \(\alpha=(\alpha_1,\dots ,\alpha_d)\) and an integer \(1\leq k\leq d\) we use the notation \(\alpha_{1:k}=(\alpha_1,\dots ,\alpha_k, 0\dots ,0)\) and \(\alpha_{(k+1):d}=(0\dots ,0,\alpha_{k+1},\dots ,\alpha_d)\).
For a PSD matrix \(A=U \operatorname{diag}(\alpha)U^{\dagger}\) with eigenvalues \(\alpha_1\geq \alpha_2\geq \dots \geq \alpha_d\), we denote by \(A_{1:k}=U \operatorname{diag}(\alpha_{1:k})U^{\dagger}\) the rank-\(k\) PSD matrix obtained by truncating the eigenvalues of \(A\).
Throughout this section we will without loss of generality always assume eigenvalues to be in non-increasing order. We then have the following technical lemmas:

\begin{lemma}[{\cite[Theorem 5]{mirskySYMMETRICGAUGEFUNCTIONS1960}}]
    \label{lem:mirsky}
    Let \(\norm{\, \cdot \,}\) be any unitarily invariant norm 
    and 
    let \(A, B\in\mathbb{C}^{d\times d}\) be complex matrices with singular values \(\alpha_1 \geq \alpha_2\geq \dots \geq\alpha_d\) and \(\beta_1 \geq \beta_2\geq\dots \geq\beta_d\), respectively.
    Then, 
    \begin{align}
        \norm*{A-B}\geq \norm*{\operatorname{diag}(\alpha-\beta)}\,.
    \end{align}
\end{lemma}

\begin{lemma}[Rank-\(k\) truncation and approximation]
    \label{lem:technical-agnostic-rank-k}
    Let \(p\in [1,\infty]\), \( \varepsilon\in (0,1/\sqrt{2})\), \(1\leq k\leq d\), and let \(\rho=U \operatorname{diag}(\alpha)U^{\dagger} \) be a \(d\)-dimensional quantum state and let \(A\in\mathbb{C}^{d\times d}\) be a positive semidefinite rank-\(k\) matrix with singular values \(a_1 \geq a_2 \geq \dots \geq a_k\) and \(a_i=0\) for \(i>k\).
    Moreover, let \(\sigma=W \operatorname{diag}(\beta)W^{\dagger}\) be a quantum state with \(\norm{\sigma-\rho}_p \leq \varepsilon\).  
    Then, it holds that
    \begin{align}
        \norm{\rho - A}_p^p \geq \sum_{i=k+1}^d \alpha_i^p\,,
        \hspace{3em}\text{and}\hspace{3em}
        \norm{\rho-\sigma_{1:k}}_p \leq 
        \norm*{\operatorname{diag}(\alpha_{(k+1):d})}_p 
        + 2 \varepsilon\,. 
    \end{align}
\end{lemma}

Note: The first inequality in \Cref{lem:technical-agnostic-rank-k} in particular implies that an optimal rank-$k$ PSD approximation to $\rho$ in the trace distance is given by its rank-$k$ truncation $\rho_{1:k}$, and it achieves an approximation of accuracy $\sum_{i=k+1}^d \alpha_i$.
Consequently, an optimal rank-$k$ state approximation to $\rho$ is given by its normalized rank-$k$ truncation $\underline{\rho_{1:k}} = \frac{\rho_{1:k}}{\norm*{\rho_{1:k}}_1}$, and it achieves an approximation of accuracy $2\sum_{i=k+1}^d \alpha_i$ in Schatten $1$-norm.

\begin{proof}[Lemma~\ref{lem:technical-agnostic-rank-k}]
    The first inequality follows from \Cref{lem:mirsky}, the unitary invariance of Schatten \(p\)-norms and 
    \begin{align}
        \norm{\operatorname{diag}(\alpha-a)}_p^p
        =\norm{\alpha-a}_{\ell_p}^p
        =\sum_{i=1}^k \abs{\alpha_i -a_i}^p
        +\sum_{i=k+1}^d \abs{\alpha_i}^p
        \geq \sum_{i=k+1}^d \alpha_i^p\,,
    \end{align}
    where we dropped some non-negative contributions and used the positive semidefiniteness of \(\rho\) in the last step. 

    The second inequality similarly follows from \Cref{lem:mirsky} and 
    \begin{align}
        \norm*{\rho-\sigma_{1:k}}_p 
        &\leq \norm*{\rho-\sigma}_p 
        +\norm*{\sigma-\sigma_{1:k}}_p
        \nonumber\\
        &\leq \varepsilon 
        + \norm*{\operatorname{diag}(\beta_{(k+1):d})}_p  \nonumber \\ 
        &\leq \varepsilon 
        + \underbrace{\norm*{\operatorname{diag}(\alpha_{(k+1):d}-\beta_{(k+1):d})}_p}_{\leq \norm*{\operatorname{diag}(\alpha-\beta)}_p \stackrel{(\star)}{\leq} \norm*{\rho-\sigma}_p} 
        + \norm*{\operatorname{diag}(\alpha_{(k+1):d})}_p \nonumber\\
        &\leq 2\varepsilon + \left(\sum_{i=k+1}^d \alpha_i^p\right)^{\frac{1}{p}}\,, \nonumber
    \end{align}
    where \((\star)\) applies \Cref{lem:mirsky}.
\end{proof}

Given these technical lemmas, we can now construct an interactive proof protocol for agnostic rank-$k$ PSD tomography, which we will use as a crucial building block for construction of an interactive protocol for agnostic rank-$k$ state tomography which proves Theorem~\ref{theorem:agnostic-rank-k-tomography}.

\begin{lemma}[Interactive agnostic rank-\(k\) PSD tomography in trace-distance]\label{lem:agnostic-rank-k-psd}
    There exists an interactive proof system $(V,P)$ for interactive agnostic rank-$k$ PSD tomography in trace distance with success probability $\geq 1-\delta$, with quantum communication between an incoherent single-copy $V$ and a coherent multi-copy $P$, such that $V$ uses at most \(\mathcal{O}(k^2\varepsilon^{-4}\log \delta^{-1})\) copies of the unknown state. More precisely, the complexities are as follows:
    \begin{itemize}
        \item $V$ uses \(\mathcal{O}(k^2\varepsilon^{-4}\log \delta^{-1})\) copies.
        \item $P$ uses \(\mathcal{O}\left( d^{2} k^4\varepsilon^{-2}\log \delta^{-1}\right)\) copies.
    \end{itemize}
\end{lemma}

\begin{proof}
    The verification procedure, described with pseudo-code in \Cref{alg:agnostic-trace-tomography}, hinges upon the following inequality from \cite{odonnell2016efficient}, the derivation of which we repeat in the soundness argument: 
   \begin{equation}\label{eq:inequality_from_OW}
       ||U' \mathrm{diag}(
        \alpha_{1:k}'
        ) U'^{\dagger}- \rho||_1\leq \sqrt{2k}||\mathrm{diag}(\alpha')-U'^{\dagger}\rho U' ||_2+||R||_1
   \end{equation}
   where $R=(U^{\dagger}\rho U)_{d-k}$ is the lower-right $(d-k) \times (d-k)$ submatrix of $U^{\dagger}\rho U$. 
   Thus, for some given rank-k hypothesis $\rho_{1:k}'=U'\mathrm{diag}(\alpha'_{1:k})U'^{\dagger}$, if an estimate for the right hand side of this inequality is upper bounded by an estimate for the optimal loss, $\sum_{i=k+1}^d\alpha_i$, plus some tolerance $\epsilon$, this ensures that $\rho'$ satisfies the agnostic learning criterion 
   \begin{equation}
       ||\rho'-\rho||_1\leq \sum_{i=k+1}^d\alpha_i+\epsilon.
   \end{equation}
   The condition $V$ checks in Step 9 is doing precisely that. Steps 1-8 ensure that $V$ obtains, with high probability, all quantities needed to compute the relevant estimates. In particular, in order to estimate the first term on the RHS of \cref{eq:inequality_from_OW}, one needs estimates of $\tr{\rho^2}$, $ \tr{\mathrm{diag}(
   \alpha')^2}$, as well as $\tr{\mathrm{diag}(
   \alpha')U'^{\dagger}\rho U'}$, which correspond to the quantities $\widehat{\mathrm{pur}}, \mathrm{pur}'$, and $\hat{o}$ respectively. Similarly, $1-\hat{p}$ is an estimate for the second term in the inequality.
   Finally, Step 3 ensures that $V$ has an accurate estimate of the optimal achievable loss. In what follows we will show that $V$ can obtain these quantities by interacting with $P$ in a way that satisfies completeness and soundness.

\begin{algorithm}[htp!]
    \caption{IP for agnostic rank-$k$ PSD tomography in trace distance}\label{alg:agnostic-trace-tomography}
        \begin{algorithmic}[1]
            \Require Verifier copy oracle $\mathsf{O}_V(\rho)$ and prover copy oracle $\mathsf{O}_P(\rho)$, where $\rho$ has eigenvalues $\alpha_1\geq\alpha_2\geq\ldots\geq \alpha_d\geq 0$; rank $1\leq k\leq d$; confidence parameter $\delta\in (0,1)$; accuracy parameter $\varepsilon\in (0,1)$
            \vspace{.2em}
            
            \Ensure subnormalized rank-\(k\) state \(\sigma\) s.t. $\norm{\rho - \sigma}_1\leq \sum_{i>k}\alpha_i +\varepsilon$ or ``abort''.
            \vspace{.25em}
            \hrule
            \vspace{.25em}

            \State $\tilde{\varepsilon}_1\gets \nicefrac{\varepsilon}{10}$, $\tilde{\varepsilon}_2\gets\nicefrac{\varepsilon^2}{96 k}$, $f(\tilde{\varepsilon}_1,\tilde{\varepsilon}_2)\gets \sqrt{6k\tilde{\varepsilon}_2} + 2\tilde{\varepsilon}_1 + \tilde{\varepsilon}_2$, $\tilde{\delta}\gets \nicefrac{\delta}{5}$.
            
            \State \(V\) uses \Cref{lem:interactive-quantum-proofs-vbqc} to interact with the prover to get an estimate \(\widehat{\mathrm{pur}}\) of the purity of \(\rho\) to accuracy \(\tilde{\varepsilon}_1\) with success probability $\geq 1-\tilde{\delta}$ from \(\mathcal{O}(\tilde{\varepsilon}_1^{-2}\log \tilde{\delta}^{-1})\) queries to \(\mathsf{O}_V\). If this subroutine aborts, $V$ aborts overall.
            \vspace{.2em}
            
            \State \(V\) uses \Cref{lem:interactive-quantum-proofs-vbqc} to interact with the prover to  perform the protocol of \cite[Corollary 1.8]{odonnell2016efficient} and thereby get an estimate \(\widehat{\alpha}_{1:k}\) of the \(k\)-the largest eigenvalues of \(\rho\) to accuracy \(\tilde{\varepsilon}_1\) in (truncated) TV distance with success probability $\geq 1-\tilde{\delta}$ from \(\mathcal{O}(k^2\tilde{\varepsilon}_1^{-2}\log \tilde{\delta}^{-1})\) queries to \(\mathsf{O}_V\). If this subroutine aborts, $V$ aborts overall.
            \vspace{.2em}

            \State $V$ asks the prover to perform full state tomography to accuracy $\tilde{\varepsilon}_2$ and with success probability $\geq 1-\tilde{\delta}$ based on Schur sampling (compare \cite{odonnell2016efficient}) from \(\mathcal{O}\left( d^{2}\tilde{\varepsilon}_2^{-2}\log \tilde{\delta}^{-1}\right)\) queries to $\mathsf{O}_P(\rho)$, and to  send back the obtained unitary $U'$ and and spectrum \(\alpha'\) with \(\alpha_1' \geq \alpha_2' \geq \dots\). 
            \vspace{.2em}
            \State If $U'$ is not unitary or if $\alpha'$ is not a valid truncated spectrum of a density matrix in non-increasing order, $V$ aborts the interaction. 
            \vspace{.2em}

            \State \(V\) calculates \(\operatorname{pur}'_k=\operatorname{pur}(\rho') = \sum_{i=1}^d \alpha_i'^2\) exactly. 
            \vspace{.2em}

            \State \(V\) uses \(\mathcal{O}\left(\tilde{\varepsilon}_2^{-2}\log \tilde{\delta}^{-1}\right)\) many queries to \(\mathsf{O}_V(\rho)\) to obtain an estimate \(\widehat{o}\) of \(o=\operatorname{tr}\left[\mathrm{diag}(\alpha')U'\rho U'^{\dagger}\right]\) to accuracy \(\tilde{\varepsilon}_2\) with success probability $\geq 1-\tilde{\delta}$.
            \vspace{.2em}

            \State \(V\) uses further \(\mathcal{O}\left(\tilde{\varepsilon}_2^{-2}\log \tilde{\delta}^{-1}\right)\) many queries to \(\mathsf{O}_V(\rho)\) to get an estimate \(\widehat{p}\) of the expectation value \(p=\operatorname{tr}\left[\Pi U'\rho U'^{\dagger} \right] \) with \(\Pi=\sum_{i\leq k}\proj{i}\) to accuracy \(\tilde{\varepsilon}_2\) with success probability $\geq 1-\tilde{\delta}$.
            \vspace{.2em}
            \If{\(\sqrt{2k\left(\operatorname{pur}'+\widehat{\operatorname{pur}}-2\widehat{o}\right)} + 1 - \widehat{p}\leq 1 - \sum_{i=1}^k \widehat{\alpha}_i + f(\tilde{\varepsilon}_1,\tilde{\varepsilon}_2)\)}
                \vspace{.2em}
                
                \State $V$ outputs $\rho'_{1:k}=U'^{\dagger} \operatorname{diag}(\alpha'_{1:k})U'$.
                \vspace{.2em}

            \Else
                \State $V$ outputs ``abort''
            \EndIf
        \end{algorithmic}
\end{algorithm}
   
\textbf{Completeness:}
    When $V$ interacts with the honest prover $P$, then, by \Cref{lem:interactive-quantum-proofs-vbqc}, Steps 2 and 3 each succeed with probability $\geq 1-\tilde{\delta}$. 
    $V$ succeeds at calculating $\mathrm{pur}'$ with certainty, and Steps 7 and 8 each succeed with probability $1-\tilde{\delta}$.
    Thus, by a union bound, we see that with an overall probability of $\geq 1- 4\tilde{\delta}$, all estimates obtained by $V$ indeed have the desired accuracies. It remains to show that in this case, the check in Step 9 is passed (with high probability), and that $\rho_{1:k}'$ produced in Step 10 is indeed a $\varepsilon$-close to optimal rank-$k$ approximation. To prove the former, note that, if we denote the true spectrum by $\alpha$ ordered as $\alpha_1\geq \alpha_2\geq \ldots$, then
    \begin{align*}
        &\left\lvert \left(\sqrt{2k\left(\operatorname{pur}'+\widehat{\operatorname{pur}}-2\widehat{o}\right)} + 1 - \widehat{p} - \sum_{i=1}^k \widehat{\alpha}_i\right) - \left(\sqrt{2k}\norm{\rho' - \rho}_2 + 1 - p - \sum_{i=1}^k \alpha_i\right)\right\rvert \\
        &\leq  
        \sqrt{2k} \abs*{
            \sqrt{
                \operatorname{pur}'
                +\widehat{\operatorname{pur}}
                -2\widehat{o}
            } 
            - \sqrt{
                \operatorname{pur}'
                +\operatorname{pur}-2o
            }
        }  
        + \abs*{\widehat{p} - p} 
        + \abs*{\sum_{i=1}^k \widehat{\alpha}_i - \sum_{i=1}^k \alpha_i} \\
        &\leq \sqrt{2k}\sqrt{\left\lvert \left(\operatorname{pur}'+\widehat{\operatorname{pur}}-2\widehat{o}\right) - \left(\operatorname{pur}'+\operatorname{pur}-2o \right)\right\rvert} + \tilde{\varepsilon}_1 + \tilde{\varepsilon}_1\\
        &\leq \sqrt{6k\tilde{\varepsilon}_2} +2 \tilde{\varepsilon}_1\, ,
    \end{align*}
    where the first step is by triangle inequality, the second step used the inequality $|\sqrt{x}-\sqrt{y}|^2\leq |x-y|$ as well as that Steps 2 and 3 succeeded, the fourth step used that Steps 7 and 8 succeeded.
    Next, note that for the honest prover $P$ implementing the Schur sampling protocol from \cite{odonnell2016efficient} on $m$ many copies of $\rho$, according to the proof of \cite[Theorem 1.5]{odonnell2016efficient} we have 
    \begin{align*}
        \Ex\left[\sqrt{2k}\norm{\rho' - \rho}_2 + 1 - p\right] - \left(1 - \sum_{i=1}^k \alpha_i\right)
        &\leq 6\sqrt{\frac{kd}{m}} \, ,
    \end{align*}
    where the expectation is over the randomness in the Schur sampling.
    Consequently, the honest prover can ensure that both $\sqrt{2k}\norm{\rho' - \rho}_2 + 1 - p - \left( 1-\sum_{i=1}^k \alpha_i\right) \leq \tilde{\varepsilon}_2$ and, because of \cite[Theorem 1.2]{odonnell2016efficient}, also $\norm{\rho'-\rho}_1\leq\tilde{\varepsilon}_2$ hold simultaneously with success probability $\geq 1-\tilde{\delta}$ by using $\mathcal{O}(d^2\tilde{\varepsilon}_2^{-2}\log\tilde{\delta}^{-1})$ many queries to $\mathsf{O}_P(\rho)$.
    Then, by the above, we conclude that $\sqrt{2k\left(\operatorname{pur}'+\widehat{\operatorname{pur}}-2\widehat{o}\right)} + 1 - \widehat{p} - \left(1 - \sum_{i=1}^k \widehat{\alpha}_i\right)\leq \sqrt{6k\tilde{\varepsilon}_2} + 2\tilde{\varepsilon}_1 + \tilde{\varepsilon}_2 = f(\tilde{\varepsilon}_1,\tilde{\varepsilon}_2)$, and hence the check in Step 9 will be passed.

    To prove that the estimate produced by $V$ is sufficiently close to optimal, note that 
    \begin{align}
        \norm{\rho_{1:k}' - \rho}_1
        &\leq \sum_{i=k+1}^{d} \alpha_i + 2\norm{\rho' - \rho}_1\\
        &\leq  \sum_{i=k+1}^{d} \alpha_i + 2\tilde{\varepsilon}_2\\
        &\leq \sum_{i=k+1}^{d} \alpha_i + \varepsilon\, ,
    \end{align}
    where we used \Cref{lem:mirsky} in the first step; the second step is because we are in the high probability event in which the honest prover $P$ successfully performed state tomography to accuracy $\tilde{\varepsilon}_2$; and the last step is by our choice of $\tilde{\varepsilon}_2$. 
    Union bounding, we see that the check in Step 9 is passed and the estimate produced in Step 10 is sufficiently close to optimal with success probability $\geq 1-5\tilde{\delta}\geq 1-\delta$, by our choice of $\tilde{\delta}$.

    \textbf{Soundness:}
    To prove soundess, consider an arbitrary prover $P'$. 
    By \Cref{lem:interactive-quantum-proofs-vbqc}, in each of Step 2 and 3, the probability that $V$ accepts the interaction and ends up with a not sufficiently accurate estimate is at most $\tilde{\delta}$. Additionally, in each of Steps 7 and 8, the probability that $V$ ends up with a not sufficiently accurate estimate is at most $\tilde{\delta}$.
    So, with overall probability $\geq 1-4\tilde{\delta}\geq 1-\delta$, $V$ succeeds in Steps 2,3,7, and 8. We condition on this high probability event for the remainder of the soundness proof.
    Now, suppose the check in Step 9 is passed. Then, by the computation from the completeness proof, we know that
    \begin{align*}
        \sqrt{2k}\norm{\rho' - \rho}_2 + 1 - p - \left( 1- \sum_{i=1}^k \alpha_i\right)
        \leq f(\tilde{\varepsilon}_1,\tilde{\varepsilon}_2) + \sqrt{6k\tilde{\varepsilon}_2} + 2\tilde{\varepsilon}_1
        = 2\sqrt{6k\tilde{\varepsilon}_2} + 4\tilde{\varepsilon}_1 + \tilde{\varepsilon}_2\, .
    \end{align*}
    We now repeat the reasoning in \cite[Proof of Theorem 1.5]{odonnell2016efficient}.
    That is, we define $R = (\mathds{1}_d-\Pi)U'^\dagger \rho U'(\mathds{1}_d-\Pi)$ and $\Gamma = U'^\dagger \rho U' - R$. Then, by unitary invariance of the Schatten $1$-norm and by triangle inequality, 
    \begin{align*}
        \norm{\rho_{1:k}' - \rho}_1
        \leq \norm{\mathrm{diag}(\alpha_1',\ldots,\alpha_k',0,\ldots,0) - \Gamma}_1 + \norm{R}_1\, .
    \end{align*}
    By construction, both $\mathrm{diag}(\alpha_1',\ldots,\alpha_k',0,\ldots,0)$ and $\Gamma$ have rank at most $k$. Hence, 
    \begin{align*}
        \norm{\mathrm{diag}(\alpha_1',\ldots,\alpha_k',0,\ldots,0) - \Gamma}_1
        &\leq \sqrt{2k}\norm{\mathrm{diag}(\alpha_1',\ldots,\alpha_k',0,\ldots,0) - \Gamma}_2\\
        &\leq \sqrt{2k}\norm{\mathrm{diag}(\alpha') - U'^\dagger \rho U'}_2\\
        &= \sqrt{2k}\norm{\rho' - \rho}_2 \, ,
    \end{align*}
    where the first step is by Cauchy-Schwarz, the second step is because adding a matrix with disjoint sets of non-zero entries can only increase the Schatten $2$-norm, and the third step is by unitary invariance of the Schatten $2$-norm.
    Next, note that 
    \begin{align*}
        \norm{R}_1
        &= \tr{R}\\
        &= \tr{(\mathds{1}_d-\Pi)U'^\dagger \rho U'(\mathds{1}_d-\Pi)}\\
        &= \tr{(\mathds{1}_d-\Pi)U'^\dagger \rho U'}\\
        &= 1 - p\, ,
    \end{align*}
    where the first step holds because $R$ is positive semidefinite as a principal submatrix of a positive semidefinite matrix, the second step is by definition of $R$, the third step uses cyclicity of the trace and $(\mathds{1}_d-\Pi)^2=(\mathds{1}_d-\Pi)$, and the final step uses $\tr{U'^\dagger \rho U'}=\tr{\rho}=1$. Combining the above inequalities, we have that
    \begin{align*}
        \norm{\rho_{1:k}' - \rho}_1
        &\leq \sqrt{2k}\norm{\rho' - \rho}_2 + 1 - p\\
        &\leq \left( 1- \sum_{i=1}^k \alpha_i \right)+ 2\sqrt{6k\tilde{\varepsilon}_2} + 4\tilde{\varepsilon}_1 + \tilde{\varepsilon}_2\\
        &= \sum_{i=k+1}^d \alpha_i + \varepsilon
    \end{align*}
    by our choice of $\tilde{\varepsilon}_1$ and $\tilde{\varepsilon}_2$. This proves soundness.

    \textbf{Complexity analysis:}
    The verifier $V$ queries their own oracle $\mathsf{O}_V(\rho)$ at most \(\mathcal{O}(\tilde{\varepsilon}_1^{-2}\log \tilde{\delta}^{-1})\) times in Steps 2, at most \(\mathcal{O}(k^2\tilde{\varepsilon}_1^{-2}\log \tilde{\delta}^{-1})\) times in Step 3, and at most \(\mathcal{O}(\tilde{\varepsilon}_2^{-2}\log \tilde{\delta}^{-1})\) in Steps 7 and 8. Plugging in our chosen $\tilde{\varepsilon}_1$, $\tilde{\varepsilon}_2$, and $\tilde{\delta}$, we see that the verifier's query complexity is bounded by \(\mathcal{O}(k^2\varepsilon^{-4}\log \delta^{-1})\).

    The honest prover $P$ queries its own oracle $\mathsf{O}_P(\rho)$ at most \(\mathcal{O}\left( \frac{d^{2}\log \tilde{\delta}^{-1}}{\tilde{\varepsilon}_2^2}\right)\) times in Step 4. Again plugging in our choice of $\tilde{\varepsilon}_2$ and $\tilde{\delta}$, this yields a prover query complexity of \(\mathcal{O}\left( \frac{d^{2}k^2\log \delta^{-1}}{\varepsilon^4}\right)\).
\end{proof}

Given Lemma~\ref{lem:agnostic-rank-k-psd} we can finally prove Theorem~\ref{theorem:agnostic-rank-k-tomography}, by constructing an interactive proof protocol for agnostic rank-$k$ quantum state tomography. 

\begin{proof}[Theorem~\ref{theorem:agnostic-rank-k-tomography}]
    We will show that interactive agnostic rank-$k$ state tomography can be achieved by essentially the IP protocol from \Cref{alg:agnostic-trace-tomography}. In particular, it suffices to change Step 10 to output the normalized $\underline{\rho_{1:k}'}$ and to run the protocol for the accuracy parameter $\varepsilon/2$ instead of $\varepsilon$.
    To this end, note that we can follow the same completeness and soundness analysis as for the original protocol, it then only remains to use \Cref{lem:mirsky} to argue that the guarantee $\norm*{\rho_{1:k}' - \rho}_1\leq \sum_{i=k+1}^d \alpha_i + \varepsilon/2$ implies $\norm*{\underline{\rho_{1:k}'} - \rho}_1\leq 2\sum_{i=k+1}^d \alpha_i + \varepsilon$.

    To this end, note that
    \begin{align*}
        \norm*{\frac{\rho'_{1:k}}{\norm*{\rho'_{1:k}}_1} - \rho'_{1:k}}_1
        = \sum_{i=1}^k \widehat{\alpha}_i \left(\frac{1}{\sum_{i=1}^k \widehat{\alpha}_i} - 1\right)
        = 1 - \sum_{i=1}^k \widehat{\alpha}_i
        \leq 1 - \sum_{i=1}^k \alpha_i + \tilde{\varepsilon}_1
        = \sum_{i=k+1}^d \alpha_i + \tilde{\varepsilon}_1\, ,
    \end{align*}
    where the second-to-last inequality holds if Step 3 succeeds.
    Therefore, in the high probability event that Steps 3 succeeds, we get
    \begin{align*}
        \norm*{\frac{\rho'_{1:k}}{\norm*{\rho'_{1:k}}_1} - \rho}_1
        &\leq \norm*{\frac{\rho'_{1:k}}{\norm*{\rho'_{1:k}}_1} - \rho'_{1:k}}_1 + \norm*{\rho'_{1:k} - \rho}_1
        \leq \sum_{i=k+1}^d \alpha_i + \tilde{\varepsilon}_1 + \norm*{\rho'_{1:k} - \rho}_1\, .
    \end{align*}
    From this, it is now immediate that $\norm*{\rho_{1:k}' - \rho}_1\leq \sum_{i=k+1}^d \alpha_i + \varepsilon/2$ indeed implies the desired inequality $\norm*{\underline{\rho_{1:k}'} - \rho}_1\leq 2\sum_{i=k+1}^d \alpha_i + \varepsilon$.
\end{proof}

We conclude this subsection by discussing a variant of the IP protocol given above. Namely, in the above interactive procedure, the prover is asked to perform full state tomography, while the object of interest is actually ``only'' a rank-$k$ approximation to the state. This might seem unnatural as the prover is asked to solve a more demanding task than what seems intuitively necessary. 
We leave open the question of whether this is an artifact of our proof or whether it is a fundamental limitation.
Here, however, we already show that, if we are willing to pay the price of obtaining a weaker learning guarantee (in the sense of giving an interactive $\alpha$-agnostic learning protocol with $\alpha>1$), then it indeed suffices for the honest prover to perform agnostic rank-$k$ tomography. 
More precisely, we argue: There exists an interactive proof for $(2(\sqrt{2k}+1)$-agnostic rank-$k$ state state tomography in which $V$ uses the same number of copies as in \Cref{theorem:agnostic-rank-k-tomography}, but now the honest prover $P$ uses only \(\mathcal{O}\left( d k^5\varepsilon^{-2}\log \delta^{-1}\right)\) copies.
The interaction differs from \cref{alg:agnostic-trace-tomography} only in Step 4, where $V$ asks the prover to perform the protocol behind \cite[Corollary 1.6]{odonnell2016efficient}, and in Step 9. 
As we argued above, the check in Step 9 is based on the inequality
   \begin{equation*}
       ||U' \mathrm{diag}(
        \alpha_{1:k}'
        ) U'^{\dagger}- \rho||_1\leq \sqrt{2k}||\mathrm{diag}(\alpha')-U'^{\dagger}\rho U' ||_2+||R||_1.
   \end{equation*}
from \cite{odonnell2016efficient}.
Its not difficult to see that the following modification also holds:       \begin{align*}
       ||U' \mathrm{diag}(
        \alpha_{1:k}'
        ) U'^{\dagger}- \rho||_1
        &\leq \sqrt{2k}||\mathrm{diag}(\alpha_{1:k}')-U'^{\dagger}\rho U' ||_2+||R||_1 \\ &\leq \sqrt{2k} \left (||\mathrm{diag}(\alpha')-U'^{\dagger}\rho U' ||_2 + ||\mathrm{diag}(\alpha')-\mathrm{diag}(\alpha_{1:k}') ||_2 \right ) +||R||_1\\
        &= \sqrt{2k}\left(\sum_{i=k+1}^d \alpha'_i\right) +  \sqrt{2k} ||\mathrm{diag}(\alpha')-U'^{\dagger}\rho U' ||_2 +||R||_1 \, .
   \end{align*}
Recall from our previous proof, which borrowed from the proof of \cite[Theorem 1.5]{odonnell2016efficient}, that the honest prover with high probability achieves $\sqrt{2k} ||\mathrm{diag}(\alpha')-U'^{\dagger}\rho U' ||_2 +||R||_1\leq 6\sqrt{\frac{kd}{m}}+\sum_{i=k+1}^d \alpha'_i$. 
So, overall we see that the honest prover with high probability achieves
\begin{align}
       ||U' \mathrm{diag}(
        \alpha_{1:k}'
        ) U'^{\dagger}- \rho||_1
        \leq 6\sqrt{\frac{kd}{m}}+(1+\sqrt{2k})\left(\sum_{i=k+1}^d \alpha'_i\right)\, .
   \end{align}
Thus, if we replace the check in Step 9 to instead become a check of the inequality  
\begin{equation}\label{eq:condition_2k_agnostic}    
    \sqrt{2k}||\mathrm{diag}(\alpha_{1:k}')-U'^{\dagger}\rho U' ||_2+||R||_1 \leq (\sqrt{2k}+1)\left(\sum_{i=k+1}^d \alpha'_i \right) + \epsilon.
\end{equation}
via suitable estimates, the honest prover $P$ will pass the check (with high probability), and by the same reasoning that we used to establish soundness previously, here the verifier $V$ can now either reject or ensure the $(2\sqrt{2k}+1)$-agnostic learning guarantee against any prover.
The two main differences: First, $V$ no longer needs $\mathrm{pur'}$ to estimate the Frobenius norm term in \cref{eq:condition_2k_agnostic}; rather, the purity of the truncation, $\tr{\mathrm{diag}(\alpha_{1:k}')^2}$, suffices. And $V$ can compute this quantity when receiving only the largest $k$ eigenvalues in Step 4. 
Note that the honest prover $P$ in Step 4 still obtains a full unitary (not just $k$ eigenvectors) $U'$ sampled from the Keyl distribution and sends it to $V$. Therefore, $V$ can still perform the estimation in Step 8. 

\subsubsection{Interactive proofs for agnostic learning stabilizer states}\label{sss:agnostic-stabilizer}

Finally, we turn to the task of agnostic stabilizer state learning, which is defined in Problem~\ref{prob:gnostic_stabilizer_learning}. To recall the problem, let $\mathcal{S}_P$ denote the set of all pure states and $\mathrm{Stab}_n\subseteq\mathcal{S}_P$ denote the set of all (pure) $n$-qubit stabilizer states. Then, given access to copies of an unknown pure state  $\ket{\psi}\in\mathcal{S}_P$, the task here is to output a stabilizer state $\ket{S}\in\mathrm{Stab}_n$ which is sufficiently close to the optimal stabilizer state approximation, as measured by fidelity~\cite{grewal2024agnostic, chen2024stabilizerbootstrappingrecipeefficient}. More specifically, defining the loss function $\ell:\mathrm{Stab}_n\times \mathcal{S}_P\rightarrow\mathbb{R}_{\geq 0}$ via
\begin{align}
\ell(\ket{S},\ket{\psi}) = 1-F(\ket{S},\ket{\psi})
\end{align}
we wish to output a stabilizer state $\ket{S}$ satisfying
\begin{equation}\label{eq:recap_loss}
\ell(\ket{S},\ket{\psi}) \leq \alpha \min_{\ket{S}\in\mathrm{Stab}_n}\left[\ell(\ket{S},\ket{\psi})\right] + \epsilon.
\end{equation}
Below, we give an interactive proof system for this task, for the case of $\alpha = 8$, which allows the verifier to surpass even the query complexity requirements of a memory-unconstrained prover acting on its own.

As discussed in \Cref{ss:trivial}, for any agnostic learning problem, the key obstacle to ``verification via direct validation of candidate solutions'' is that one requires not only an estimate of the loss of the solution, but also an estimate of the \textit{optimal} loss over the whole model class. For the problem of agnostic stabilizer state learning, the loss $\ell(\ket{S},\ket{\psi})$ of any given hypothesis stabilizer state $\ket{S}$ can be estimated efficiently by simply measuring the projector $\ket{S}\bra{S}$ on the unknown state $\ket{\psi}$. As such, we can focus on the optimal loss, which we denote by $\ell^*(\ket{\psi})$. In this case we have
\begin{align}
\ell^*(\ket{\psi}) 
= \min_{S}\left[\ell(\ket{S},\ket{\psi})\right] 
=\min_S \left[1-F(\ket{S},\ket{\psi})\right]
= 1 - F_{\mathrm{Stab}}(\ket{\psi})\label{eq:opt_via_stab}\, ,
\end{align}
where $F_{\mathrm{Stab}}(\ket{\psi})$ is the \textit{stabilizer fidelity}~\cite{Bravyi2019simulationofquantum}. The idea behind the proof system we present below is that the stabilizer fidelity, and hence the optimal loss, can be bounded from above and below in terms of quantities that the verifier can efficiently estimate (using their own quantum state copies, and the help of the prover). These upper and lower bounds serve as a proxy for the optimal loss (also referred to as a certificate of loss in Ref.\ \cite{goldwasser2021interactive}). Ultimately, we show how they allow the verifier $V$ to check the validity of alleged hypothesis solution sent by the prover, by comparing the loss of the hypothesis to this proxy for the optimal loss.

\begin{theorem}[8-Agnostic stabilizer learning (formal version of Theorem~\ref{theorem:agnostic-stabilizer-state-learning-intro-version})]\label{theorem:agnostic-stabilizer-state-learning-formal-version}
There exists an interactive proof system $(V,P)$ for $8$-agnostic stabilizer state learning, between a single-copy verifier $V$ and a coherent multi-copy prover $P$, who communicate via a quantum communication channel. In particular,  $V$ is computationally efficient and:
\begin{enumerate}
\item $V$ uses $\mathcal{O}\left(\frac{\log(\delta^{-1})}{\epsilon^2}\right)$ oracle queries.
\item The honest prover $P$ uses $\mathcal{O}\left(n\log(\delta^{-1}) + \frac{\log(\delta^{-1})}{\epsilon^2}\right)$ oracle queries.
\item The total communication is efficient.
\end{enumerate}
\end{theorem}

As shown in Table~\ref{table:advantages-via-IPs}, we highlight that even realizable stabilizer state learning is known to require $\Theta(n^2)$ copies with single-copy measurements and $\Theta(n)$ copies with multi-copy measurements \cite{montanaro2017learning, arunachalam2023optimal, arunachalam2026optimalstabilizertestinglearning}. As such, the verifier in \Cref{theorem:agnostic-stabilizer-state-learning-formal-version} solves a more general agnostic learning problem with an $n$-independent (and thus largely improved) quantum sample complexity. 

\begin{algorithm}[htp!]
        \caption{IP for 8-agnostic stabilizer learning}\label{alg:agnostic-stabilizer}
            \begin{algorithmic}[1]
                \Require Verifier copy oracle $\mathsf{O}_V(|\psi\rangle)$; 
                confidence parameter $\delta\in (0,1)$; accuracy parameter $\varepsilon\in (0,1)$; 
                \Ensure $\ket{S}$ s.t. $F(\ket{S},\ket{\psi})\geq F_{\mathrm{stab}}(\ket{\psi}) - \varepsilon$ or ``abort''
                \State Set $\epsilon_1 = \epsilon_2 = \frac{\epsilon}{5}$, and $\epsilon_3 = \frac{3}{20}\epsilon$. 
                \State Set $\delta_1 = \delta_2=\delta_3 = \nicefrac{\delta}{3}$.
                \State $V$ asks $P$ to use the algorithm from Corollary 6.3 of Ref.~\cite{chen2024stabilizerbootstrappingrecipeefficient} to solve the 1-agnostic stabilizer learning problem to precision $\epsilon_1$, with success probability $\geq \delta_1$, with a promise that $F_{\mathrm{stab}}(\ket{\psi}) \geq \nicefrac{7}{8}$. $P$ should send the result $\ket{S}$ to $V$.
                \State $V$ makes $\mathcal{O}(\log(\delta_2^{-1})/\epsilon_2^2)$ queries to its oracle, and measures $\{\ket{S}\bra{S}, I -\ket{S}\bra{S}\} $ on each received state $\ket{\psi}$. From the measurement outcomes $V$ calculates an estimate $\hat{\ell}$ of the loss $\ell(\ket{S},\ket{\psi})$, to additive precision $\epsilon_2$, with probability of success $\geq1-\delta_2$.
                \State $V$ makes another $\mathcal{O}(\log(\delta_3^{-1})/\epsilon_3^2)$ queries to its oracle, and  delegates to the prover, via the generic VQBC protocol discussed previously in \Cref{lem:interactive-quantum-proofs-vbqc}, the estimation of $A_3(\ket{\psi})$ via Bell sampling. With probability $\geq1-\delta_3$, $V$ receives and accepts an estimate $\hat{a}$ for $A_3(\ket{\psi})$ with additive error at most $\epsilon_3$. 
                \State $V$ then calculates $\hat{u} = \min\left(1,\frac{4}{3}\left(1-\hat{a}\right)\right)$, which is an estimate for $\mathrm{UB}$.
                \If{$\hat{\ell}\leq \hat{u} + \nicefrac{3}{5}\varepsilon$}
                    \State $V$ outputs $\ket{S}$
                \Else
                    \State $V$ outputs ``abort".
                \EndIf
            \end{algorithmic}
    \end{algorithm}

\begin{proof} We start by presenting appropriate upper and lower bounds on the optimal loss $\ell^*(\ket{\psi})$, which we will make use of in the interactive proof protocol. To this end, consider the quantity 
\begin{equation}
    A_3(\ket{\psi}) = \frac{1}{2^{n}}\sum_{P\in\{ I,X,Y,Z\}^n}\left(\left|\bra{\psi}P\ket{\psi}\right|^{2}\right)^{3}\, .
\end{equation}
As per Eq. (11) in Ref.~\cite{haug-efficient-quantum}, we note that the stabilizer fidelity can be bounded in terms of this quantity via 
\begin{equation}
A_{3}\left(\ket{\psi}\right)^{1/6}\geq F_{\mathrm{Stab}}\left(\ket{\psi}\right)\geq\frac{4}{3}A_{3}\left(\ket{\psi}\right)-\frac{1}{3}.
\label{eq:upper-lower-bounds-f-stab}
\end{equation}
Additionally, note that, via Bell sampling, $A_3(\ket{\psi})$ can be estimated to additive precision $\epsilon$, with success probability $\geq 1-\delta$, using $\mathcal{O}(\log(\delta^{-1})/\epsilon^2)$ many copies of $\ket{\psi}$ \cite{gross2021schur, grewal-low-stabilizer, haug-efficient-quantum}. Using Eqs.~\eqref{eq:opt_via_stab} and~\eqref{eq:upper-lower-bounds-f-stab} we also see that
\begin{equation}\label{eq:upper_lower_A}
1-A_{3}\left(\ket{\psi}\right)^{1/6}\leq \ell^*(\ket{\psi}) \leq \frac{4}{3}\left(1 - A_3(\ket{\psi}) \right).
\end{equation}
For convenience we will from now on use the notation
\begin{align}
\mathrm{UB}(\ket{\psi})&=\frac{4}{3}\left(1 - A_3(\ket{\psi}) \right),\\
\mathrm{LB}(\ket{\psi})&= 1-A_{3}\left(\ket{\psi}\right)^{1/6},
\end{align}
to denote the upper and lower bounds on the optimal loss $\ell^*(\ket{\psi})$ in terms of $A_3(\ket{\psi})$.

With this in hand, as presented in Algorithm~\ref{alg:agnostic-stabilizer}, the interactive proof protocol proceeds as follows: Set $\epsilon_1 = \epsilon_2 = \frac{\epsilon}{5}$, and $\epsilon_3 = \frac{3}{20}\epsilon$. Set $\delta_1 = \delta_2=\delta_3 = \nicefrac{\delta}{3}$. Then:
\begin{enumerate}
    \item $V$ asks $P$ to use its oracle $\mathsf{O}_P(\ket{\psi})$ for the unknown state $\ket{\psi}$, and the multi-copy algorithm from Corollary 6.3 of Ref.~\cite{chen2024stabilizerbootstrappingrecipeefficient}, to solve the 1-agnostic stabilizer learning problem to precision $\epsilon_1$, with success probability $\geq\delta_1$, and a promise that $F_{\mathrm{Stab}}(|\psi\rangle)\geq \nicefrac{7}{8}$. $P$ should then send a classical description of its hypothesis $\ket{S}$ to $V$.
    \item $V$ makes $\mathcal{O}(\log(\delta_2^{-1})/\epsilon_2^2)$ queries to its oracle, and measures $\{\ket{S}\bra{S}, I -\ket{S}\bra{S}\} $ on each received state $\ket{\psi}$. From the measurement outcomes $V$ calculates an estimate $\hat{\ell}$ of the loss $\ell(\ket{S},\ket{\psi})$, to additive precision $\epsilon_2$, with probability of success $\geq1-\delta_2$.
    \item $V$ makes another $\mathcal{O}(\log(\delta_3^{-1})/\epsilon_3^2)$ queries to its oracle, and  delegates to the prover, via the generic VQBC protocol discussed previously in \Cref{lem:interactive-quantum-proofs-vbqc}, the estimation of $A_3(\ket{\psi})$ via Bell sampling. With probability $\geq1-\delta_3$, $V$ receives and accepts an estimate $\hat{a}$ for $A_3(\ket{\psi})$ with additive error at most $\epsilon_3$. 
    \item $V$ then calculates $\hat{u} = \min\left(1,\frac{4}{3}\left(1-\hat{a}\right)\right)$, which is an estimate for $\mathrm{UB}$. Note that $|\hat{a}-A_3|\leq \epsilon_3$ implies that $|\hat{u} - \mathrm{UB}|\leq \frac{4}{3}\epsilon_3$.
    \item Finally, $V$ checks that the estimated loss satisfies $\hat{\ell}\leq \hat{u} + \frac{3}{5}\epsilon$. If yes, $V$ accepts and outputs $\ket{S}$, otherwise $V$ aborts the interaction.
\end{enumerate}
Before proceeding to the analysis of the protocol, we state some useful facts. Additionally, in what follows we drop the explicit dependence of $\mathrm{UB},\mathrm{LB},A_3$ and $\ell^*$ on $\ket{\psi}$ for notational convenience. With this in mind, we note that $A_3\in[0,1]$, and that 
\begin{equation}
\frac{\mathrm{UB}}{\mathrm{LB}} = \frac{\frac{4}{3}\left(1 - A_3 \right)}{1-A_{3}^{1/6}}
\end{equation}
is a monotonically increasing function on the interval $[0,1]$, which attains its maximum value as ${A_3\to 1}$. The limit can be calculated via l'Hôpital's rule, and we conclude that
\begin{equation}\label{eq:8_factor}
\frac{\mathrm{UB}}{\mathrm{LB}} \leq 8.
\end{equation}
Additionally, it follows from $\ell^*\geq \mathrm{LB}$ that
\begin{equation}\label{eq:inter_fact_2}
\left(1 + \frac{ (\mathrm{UB}-\mathrm{LB})}{\ell^*}\right) \leq \left(1 + \frac{ (\mathrm{UB}-\mathrm{LB})}{\mathrm{LB}}\right) \leq \frac{\mathrm{UB}}{\mathrm{LB}},
\end{equation}
and that
\begin{align}\label{eq:inter_fact}
\mathrm{UB} \leq \ell^* + (\mathrm{UB}-\mathrm{LB}).
\end{align}
Given this, we can analyze the completeness and soundess of the interactive proof protocol.

\textbf{Completeness:} We want to show that if the prover $P$ is honest, then with probability $\geq1-\delta$ the verifier $V$ will accept and output a valid hypothesis in Step 5 of the protocol. To this end, let us first note that the loss function is upper boundded by one -- i.e. $\ell(\ket{S},\ket{\psi}) \in [0,1]$. Additionally, a valid hypothesis is any hypothesis that satisfies Eq.~\eqref{eq:recap_loss}, with $\alpha = 8$. Given this, the only non-trivial case is when $\ell^*(\ket{\psi})\leq \nicefrac{1}{8}$. If $\ell^*(\ket{\psi})>\nicefrac{1}{8}$, then to be valid a hypothesis $\ket{S}$ should satisfy $\ell(\ket{S},\ket{\psi}) \leq 8\times\nicefrac{1}{8} + \epsilon \leq 1 + \epsilon$, which will be satisfied by \textit{any} state $\ket{S}$. As such, we assume that $\ell^*(\ket{\psi})>\nicefrac{1}{8}$, or equivalently that $F_{\mathrm{Stab}}(\ket{\psi}) \geq \nicefrac{7}{8}$. This is what allows us to make this promise to the honest prover in Step 1 of the algorithm.

With this in mind, let us denote with $E_1$ the event where:
\begin{enumerate}
\item The honest prover succeeds in Step 1 -- i.e., $V$ receives some $\ket{S}$ satisfying $\ell(\ket{S},\ket{\psi}) \leq \ell^* + \epsilon_1$.
\item In Step 2, $V$ successfully obtains an estimate $\hat{\ell}$ satisfying $|\hat{\ell} - \ell(\ket{S},\ket{\psi}) |\leq \epsilon_2$.
\item The delegation in Step 3 succeeds and $V$ receives and accepts $\hat{a}$ satisfying $|\hat{a}-A_3|\leq \epsilon_3$ (which then implies $|\hat{u} - UB|\leq \frac{4}{3}\epsilon_3$). 
\end{enumerate}
By a union bound, and the choice of $\delta_1,\delta_2$ and $\delta_3$, we have that $\mathrm{Pr}(E_1) \geq 1-\delta$. Now, lets assume that event $E_1$ has occurred. By assumption we know that honest $P$ succeeded and therefore that $\ket{S}$ satisfies $\ell(\ket{S},\ket{\psi})\leq \ell^* + \epsilon_1 \leq 8\ell^* + \epsilon$. As such, it only remains to show that $V$ will accept. To this end, note that under the assumptions of event $E_1$, and the choices of $\epsilon_1,\epsilon_2$ and $\epsilon_3$,
\begin{align}
\hat{\ell} &\leq \ell(\ket{S},\ket{\psi}) + \epsilon_2 \nonumber\\
&\leq \ell^* + \epsilon_1 + \epsilon_2\nonumber\\
&\leq \mathrm{UB} + \epsilon_1 + \epsilon_2 \nonumber\\
&\leq \hat{u} + \epsilon_1 + \epsilon_2 + \frac{4}{3}\epsilon_3\nonumber\\
&\leq \hat{u} + \frac{3}{5}\epsilon.
\end{align}
Therefore, $V$ accepts and we have proven completeness.

\textbf{Soundness:} We want to show that, for any prover $P$, the probability that $A$ accepts the interaction and outputs an invalid hypothesis is less than $\delta$. To this end, lets denote by $E_2$ the event where
\begin{enumerate}
\item In Step 2, $V$ successfully obtains an estimate $\hat{\ell}$ satisfying $|\hat{\ell} - \ell(\ket{S},\ket{\psi}) |\leq \epsilon_2$.
\item The delegation in Step 3 succeeds and $V$ receives and accepts $\hat{a}$ satisfying $|\hat{a}-A_3|\leq \epsilon_3$ ( which then implies $|\hat{u} - \mathrm{UB}|\leq \frac{4}{3}\epsilon_3$).
\end{enumerate}
By a union bound, and the choices of $\delta_2$ and $\delta_3$, we have $\mathrm{Pr}(E_2) \geq 1-\frac{2}{3}\delta \geq 1-\delta$. Now, assume that event $E_2$ occurred, and that $V$ accepted the interaction (i.e., that event $E_2$ occurred and $\hat{\ell}\leq \hat{u} + \frac{3}{5}\epsilon$). In this case, 
we have that
\begin{align}
\ell(\ket{S},\ket{\psi}) &\leq \hat{\ell} + \epsilon_2 \nonumber\\
&\leq \hat{u} + \frac{3}{5}\epsilon + \epsilon_2 \nonumber\\
&\leq \mathrm{UB} + \frac{4}{3}\epsilon_3 + \frac{3}{5}\epsilon+ \epsilon_2 \nonumber\\
&\leq \ell^* + (\mathrm{UB-\mathrm{LB}}) + \epsilon \hspace{10em}\text{[via Eq.\eqref{eq:inter_fact} and $\epsilon_2,\epsilon_3$]}\nonumber\\
&\leq \ell^*\left(1 + \frac{ (\mathrm{UB}-\mathrm{LB})}{\ell^*}\right) + \epsilon\nonumber\\
&\leq \frac{\mathrm{UB}}{\mathrm{LB}}\ell^* + \epsilon \hspace{13.5em}\text{[via Eq.\eqref{eq:inter_fact_2}]}\nonumber\\
&\leq 8\ell^* + \epsilon, \hspace{14.2em}\text{[via Eq.\eqref{eq:8_factor}]}\nonumber
\end{align}
and therefore $\ket{S}$ is a valid (8-agnostic) solution. In other words, if event $E_2$ occurs, then whenever $V$ accepts, $\ket{S}$ is a valid solution. Therefore, the only way that $V$ could accept and output an invalid hypothesis is if event $E_2$ does not occur, which happens with probability less than $\delta$. As such, we have proved soundness.

\textbf{Complexity analysis:} $V$ uses $\mathcal{O}\left(\frac{\log(\delta^{-1})}{\epsilon^2}\right)$ queries to its oracle (in steps 2 and 3). The honest prover $P$ uses $\mathcal{O}\left(n\log(\delta^{-1}) + \frac{\log(\delta^{-1})}{\epsilon^2}\right)$ queries, when running the algorithm in Corollary 6.3 of Ref.~\cite{chen2024stabilizerbootstrappingrecipeefficient} with $\tau = \nicefrac{7}{8}$. During the proof, $V$ sends $\mathcal{O}\left(\frac{\log(\delta^{-1})}{\epsilon^2}\right)$ $n$-qubit quantum states to $P$ in Step 3. 
The communication is efficient as a consequence of \Cref{lem:interactive-quantum-proofs-vbqc}.
\end{proof}

We note that the above interactive proof system generalizes to agnostic learning settings where the verifier:
\begin{enumerate}
\item Can estimate the loss (necessary for Step 2).
\item Can estimate a quantity which can be used to upper bound and lower bound the optimal loss (necessary for Step 4 and soundness analysis).
\end{enumerate}
In particular, the value of $\alpha$ which can be achieved in the soundness analysis depends on the tightness of the upper and lower bounds on the optimal loss. Specifically, the above interactive protocol will yield an interactive proof system for $\alpha$-agnostic learning for any $\alpha$ such that $\mathrm{UB}/\mathrm{LB}\leq \alpha$. 

\section*{Acknowledgments}

The authors thank Mina Doosti, Vedran Dunjko, Alexandru Gheorghiu, Tom Gur, Elham Kashefi, Ninad Rajgopal, Chirag Wadhwa, and especially Alex Grilo for insightful discussions.

The Berlin authors gratefully acknowledge support from the BMBF (FermiQP, MuniQCAtoms, DAQC), the Munich Quantum Valley (K-8), the Quantum Flagship (PasQuans2, Millenion), the QuantERA (HQCC), the DFG (CRC 183), the Einstein Research Unit, Berlin Quantum, and the ERC (DebuQC). MCC was partially supported by a DAAD PRIME fellowship.

\section*{Author contribution statement}

All authors contributed equally to this work across the phases of idea generation, proof development, and paper writing. No AI tools were
used in creating this work.

\newpage

\appendix

\section{Considered problems}\label{app:problems}

We provide here a list of precise definitions for testing and learning problems which are considered explicitly in this work. Throughout this section we use the notation $\mathcal{S}$ and $\mathcal{S}_p$ for the set of quantum states and pure quantum states respectively. The dimension of the states will either be specified in the text, or via a subscript. We use the notation $\mathcal{D}$ for the set of discrete distributions over bit strings. The length of the bit strings will either be specified in the text, or with a subscript.

\subsection{Many-vs-one distinguishing problems}

\begin{problem}[Uniformity testing]\label{prob:uniformity} Let $U\in\mathcal{D}$ denote the uniform distribution over length-$n$ bit strings. Let $B_{\epsilon}(U)\subset \mathcal{D}$ be the $\epsilon$-ball centered at $U$ with respect to the total variation distance. Uniformity testing is the many-vs-one distinguishing problem $(U,\mathcal{D}\setminus B_{\epsilon}(U),\mathcal{D})$. In words, uniformity testing is the problem which asks one to decide whether an unknown instance is the uniform distribution, or at least $\epsilon$-far from the uniform distribution
\end{problem}

\begin{problem}[Parity distribution testing]\label{prob:parity_testing}
Let $\mathcal{D}_n$ be the set of probability distributions over length-$n$ bit strings, and $U_n\in\mathcal{D}_n$ the uniform distribution. For any $s\in\{0,1\}^n$ define the parity distribution $D_s\in\mathcal{D}_{n+1}$ as the distribution which is sampled from by the following process: Draw $x\sim U_n$, then output $x \mathbin\Vert \mathrm{par}_s(x)\in\{0,1\}^{n+1}$ where $\mathrm{par}_s(x) = s\cdot x\mod 2$ is the $s$-parity of $x$.  Finally, define the set of parity distributions $\mathcal{D}_P:=\{D_s\,|\,s\in\{0,1\}^n\}$. With this in hand, parity distribution testing is the many-vs-one decision problem $(U_{n+1},\mathcal{D}_P,\mathcal{D}_{n+1})$. In words, parity distribution testing is the problem which asks one to decide whether an unknown distribution is the uniform disrtibution or a parity distribution.
\end{problem}

\begin{problem}[Purity testing]\label{prob:purity_testing} Purity testing is the the many-vs-one distinguishing task $(\mathds{1}/d,\mathcal{S}_p,\mathcal{S})$. In words, purity testing is the problem which asks one to decide whether an unknown state is the maximally mixed state or a pure state.  
\end{problem}

\begin{problem}[Unitarity testing]\label{prob:unitarity_testing}
    Unitarity testing is the the many-vs-one distinguishing task $(\Delta_1,\mathsf{U},\mathsf{CPTP})$, where $\Delta_1$ denotes the maximally depolarizing channel, $\mathsf{U}$ denotes the set of unitary channels, and $\mathsf{CPTP}$ denotes the set of all channels, all for some fixed dimension $d$. In words, unitarity testing is the problem which asks one to decide whether an unknown channel is the maximally depolarizing channel or a unitary channel.
\end{problem}

\begin{problem}[Quantum state certification]\label{prob:quantum-state-certification} For any state $\rho\in\mathcal{S}$ let $B_\epsilon(\rho)\subset \mathcal{S}$ be the epsilon ball centered at $\rho$, with respect to the trace distance. For any state $\rho\in\mathcal{S}$ the associated quantum state certification problem 
is the many-vs-one distinguishing problem $(\rho,\mathcal{S}\setminus B_{\epsilon}(\rho),\mathcal{S})$. In words, quantum state certification for $\rho$ is the problem which asks one to decide whether an unknown state is $\rho$, or at least $\epsilon$-far from $\rho$ with respect to the trace distance.
\end{problem}

\begin{problem}[Stabilizer testing]\label{prob:stabilizer-testing} Let $\mathrm{Stab}_n\subset \mathcal{S}_P$ be the set of $n$-qubit stabilizer states. Stabilizer testing is the many-vs-one distinguishing problem $(\mathds{1}_2^{\otimes n}/2^n,\mathrm{Stab}_n,\mathcal{S})$. In words, stabilizer testing is the problem which asks one to decide whether an unknown $n$-qubit state is the maximally mixed state, or a stabilizer state.
\end{problem}

\begin{problem}[Clifford testing]\label{prob:clifford_testing}
    Clifford testing is the the many-vs-one distinguishing task $(\Delta_1,\mathsf{Cl},\mathsf{CPTP})$, where $\Delta_1$ denotes the maximally depolarizing channel, $\mathsf{Cl}$ denotes the set of Clifford unitary channels, and $\mathsf{CPTP}$ denotes the set of all channels, all for some fixed dimension $d=2^n$. In words, unitarity testing is the problem which asks one to decide whether an unknown channel is the maximally depolarizing channel or a Clifford unitary channel.
\end{problem}

\begin{problem}[Pauli spike detection]\label{prob:Pauli-spike-detection} For any $\varepsilon\in[0,1)$ define the set of ``$\varepsilon$-Pauli-spiked'' $n$-qubit states $\mathcal{S}_{\mathrm{PS}}$ via
\begin{equation}
\mathcal{S}_{\mathrm{PS}} = \left\{\frac{\mathds{1}_2^{\otimes n} + 3\varepsilon P}{2^n}\,|\, P\in\{\mathds{1}_2, X,Y,Z\}^{\otimes n}\setminus \{\mathds{1}_2^{\otimes n}\}\right\}.
\end{equation}
Pauli spike detection is then the many-vs-one distinguishing problem $\{\nicefrac{\mathds{1}_2^{\otimes n}}{2^n},\mathcal{S}_{\mathrm{PS}},\mathcal{S}\}$. In words, Pauli spike detection is the problem which asks one to decide whether an unknown state is the maximally mixed state, or a Pauli-spiked state. We note that Choi state and channel version of Pauli spike detection have been studied in Refs.~\cite{caro2023learning, chen2023efficient, chen2024tight}.
\end{problem}

\subsection{Learning problems}

\begin{problem}[Agnostic parity learning]\label{prob:agnostic_parity_learning} Let $\mathcal{D}_P\subset\mathcal{D}_{n+1}$ be the set of parity distributions defined in Problem~\ref{prob:parity_testing}. Distributional agnostic parity learning is the agnostic learning problem $\{\mathcal{D}_{n+1},\mathcal{D}_P,\ell,\epsilon,\alpha\}$, where the loss function $\ell$ is given by
\begin{equation}\label{eq:misclassification_prob}
    \ell (D_s,D) =\underset{(x,y)\sim D}{\mathrm{Pr}}[\mathrm{par}_s(x)\neq y] = \underset{(x,y)\sim D}{\mathbb{E}}[\mathds{1}_{\mathrm{par}_s(x)\neq y}],
\end{equation}
which is the misclassification probability when using $\mathrm{par}_s$ to predict the last bit of a sample from $D$, given the first $n$ bits. In words, agnostic parity learning is the problem which asks one, when given access to an unknown distribution, to output a parity distribution which is close enough to the best possible parity distribution, with respect to the misclassification probability.
\end{problem}

\begin{problem}[Agnostic quantum state tomography]\label{prob:agnostic_qst} Let $\mathcal{M}\subseteq{S}$ be a subset of quantum states. Agnostic quantum state tomography of $\mathcal{M}$ is the agnostic learning problem $\{\mathcal{S},\mathcal{M},\ell,\epsilon,\alpha\}$, where $\ell$ is typically the trace distance or the fidelity (strictly, one minus the fidelity), but can be any other metric for comparing quantum states. We also define \textit{pure-state} agnostic quantum state tomography as the agnostic learning problem $\{\mathcal{S}_P,\mathcal{M},\ell,\epsilon,\alpha\}$, where $\mathcal{S}_P$ is any subset of pure states. In words, (pure-state) agnostic quantum state tomography is the problem which asks one, when given access to an unknown (pure) quantum state, to output a state in $\mathcal{M}$ which is close enough to the best possible state in $\mathcal{M}$, with respect to $\ell$. 
\end{problem}

\begin{problem}[Quantum state tomography]\label{prob:qst} Quantum state tomography is the special case of agnostic quantum state tomography in which $\mathcal{M}=\mathcal{S}$ is the set of all quantum states. As such, quantum state tomography can also be written as the realizable learning problem $\{\mathcal{S},\ell,\epsilon\}$, where again $\ell$ is most-often either trace distance or fidelity, but can also be any other metric on quantum states. In words, quantum state tomography is the problem which asks one, when given access to an unknown quantum state, to output a description of a quantum state which is $\epsilon$-close to the unknown state with respect to $\ell$.
\end{problem}

\begin{problem}[Agnostic rank-$k$ quantum state tomography]\label{prob:agnostic_rank_k_quantum_state}
Let $\mathcal{S}_{\mathrm{rank}(k)}\subset\mathcal{S}$ denote the set of rank-$k$ quantum states. We then define agnostic rank-$k$ quantum state tomography as the agnostic learning problem $(\mathcal{S},\mathcal{S}_{\mathrm{rank}(k)},\ell,\epsilon,\alpha)$. When $\ell$ is the distance induced by the Schatten $p$-norm, then we refer to the problem as agnostic rank-$k$ quantum state tomography with respect to the Schatten $p$-norm. In words, agnostic rank-$k$ quantum state tomography is the problem which asks one, when given access to an unknown quantum state, to output a description of a rank-$k$ quantum state which is close enough to the best possible rank-$k$ approximation of the unknown quantum state.
\end{problem}

\begin{problem}[Agnostic rank-$k$ PSD tomography]\label{prob:agnostic_rank_k_PSD}
Let $\mathrm{PSD}_{\mathrm{rank}(k)}\supset \mathcal{S}_{\mathrm{rank}(k)}$ denote the set of all {$d$-dimensional} rank-$k$ positive semi-definite operators. We then define agnostic rank-$k$ PSD tomography as the agnostic learning problem $(\mathcal{S},\mathrm{PSD}_{\mathrm{rank}(k)},\ell,\epsilon,\alpha)$. Again, when $\ell$ is the distance induced by the Schatten $p$-norm, then we refer to the problem as agnostic rank-$k$ PSD tomography with respect to the Schatten $p$-norm. 
\end{problem}

Note that agnostic rank-$k$ PSD tomography is extremely similar to agnostic rank-$k$ \textit{quantum state} tomography, however, in agnostic rank-$k$ PSD tomography, any rank-$k$ PSD matrix is allowed as a hypothesis, even if it is not a valid rank-$k$ quantum state (which is required of hypotheses in agnostic rank-$k$ quantum state tomography).
In that sense, agnostic rank-$k$ PSD tomography can be viewed as an instance of improper learning. 

\begin{problem}[Agnostic stabilizer state learning]\label{prob:gnostic_stabilizer_learning} Agnostic stabilizer state learning is the pure-state agnostic quantum state tomography problem $\{\mathcal{S}_P,\mathrm{Stab}_n,1-F,\epsilon,\alpha\}$. In words, agnostic stabilizer state learning is the problem which asks one, when given access to an unknown pure state, to output a stabilizer state which is close to the best possible stabilizer state approximation to the unknown state, measured via the fidelity.
\end{problem}

\begin{problem}[Pauli spike search]\label{problem:pauli_spike_search} Pauli spike search is a learning version of Pauli spike detection, defined in Problem~\ref{prob:Pauli-spike-detection}, in which instead of merely requiring a quantum algorithm to distinguish between the maximally mixed state and an $\varepsilon$-Pauli-spiked state, any algorithm which claims a spike should also correctly identify the corresponding Pauli operator \cite{chen2023efficient, chen2024tight}. 
\end{problem}

\begin{problem}[Pauli shadow tomography]\label{problem:pauli_shadow_tomography}
Let $\mathcal{P} =\{\mathds{1}_2, X,Y,Z\}^{\otimes n}\}$ be the set of $n$-qubit Pauli operators. Pauli shadow tomography~\cite{huang2021information, king2024triply, chen2024optimalpauli} is the problem in which, when given oracle access to an unknown quantum state $\rho$, a learning algorithm should, with probability $\geq 1-\delta$, output a classically efficient function $f:\mathcal{P}\rightarrow\mathbb{R}$ satisfying $|f(P) - \tr{P\rho}|\leq \varepsilon$, for all $P\in\mathcal{P}$. In words, Pauli shadow tomography is the problem which asks one, when given access to an unknown quantum state, to produce simultaneously $\varepsilon$-accurate estimates for the expectation values $\tr{P\rho}$, for all $P\in\mathcal{P}$. We note that Pauli shadow tomography is a special case of the more general shadow tomography problem~\cite{aaronson2019shadow, badescu2023improved, abbas2023quantum}, and that  one can also define a special case of Pauli shadow tomography which relaxes the requirement of predicting \textit{all} Pauli expectation values, to only predicting some subset~\cite{king2024triply, chen2024optimalpauli}.
\end{problem}

\section{More motivating examples}\label{app:motivating_examples}

Here, we supplement the motivating examples given in \Cref{ss:motivating_examples} with a variety of additional examples, all of which illustrate resource dependent query complexity separations, and are therefore interesting settings for this work. We note that this is by no means an exhaustive list.

\begin{example}[Parity distribution testing -- Problem~\ref{prob:parity_testing}]\label{example:parity-testing-informal}
For the randomized version of parity distribution testing (c.f. Definition~\ref{def:solve_random_problem}), in which  $\mu$ is the uniform measure over $\mathcal{X}_R$, one recovers the problem of distinguishing between a uniformly random parity function and the fully uniform distribution. For this problem, it is known that \textit{sample access} is a relevant resource. More specifically, given a sample oracle $O(n)$ queries are sufficient. However, \(2^{\Omega(n)}\) queries are necessary to solve that task from statistical queries \cite{kearns1998efficient}. As such, one can ask whether there exists an interactive proof protocol which allows a verifier with statistical query access, through interaction with an unconstrained prover with sample access, to outperform the lower bound imposed by statistical queries.  
\end{example}

\begin{example}[Quantum state certification -- Problem~\ref{prob:quantum-state-certification}]\label{example:Quantum_state_certification} For this problem, Ref~\cite{buadescu2019quantum} showed that $\Theta\left(\nicefrac{d}{\varepsilon^2}\right)$ copies of the unknown state are necessary and sufficient for quantum state certification if one has enough quantum memory to make coherent multi-copy measurements. However, this complexity cannot be achieved with only incoherent single-copy measurements~\cite{chen2022tight}. As such, here one can ask whether a resource-constrained verifier, with only enough quantum memory to make incoherent single-copy measurements, can gain any advantage through interaction with a prover who has the ability to make coherent multi-copy measurements.
\end{example}

\begin{example}[Stabilizer testing -- Problem~\ref{prob:stabilizer-testing}]\label{example:stabilizer-testing}
This task can be solved from \(\mathcal{O}(n^2)\) single-copies of the unknown state \cite{aaronson2008identifying} and \(\mathcal{O}(n)\) many copies with coherent multi-copy measurements \cite{montanaro2017learning}. As such, \textit{quantum memory} is a relevant resource for stabilizer testing. Additionally, for inverse polynomially accurate 
quantum statistical queries (QSQs), this task requires \(\Omega(2^n)\) 
QSQs to the unknown state \(\rho\) \cite{arunachalam2023role,nietner2023unifying}  and \(\Omega(2^n)\) QSQs to the two-copy state \(\rho\otimes \rho\) \cite{nietner2023unifying}. As such, \textit{quantum copies}, as opposed to QSQs, are also a relevant resource for this problem. 
\end{example}

\begin{example}[Pauli spike detection -- Problem~\ref{prob:Pauli-spike-detection}]\label{example:pauli-spike-detection-informal}
For the randomized version of Pauli spike detection (c.f. Definition~\ref{def:solve_random_problem})
with $\mu$ the uniform measure over $\mathcal{X}_R$, as well as the Choi state and channel versions of this problem~\cite{caro2023learning, chen2023efficient, chen2024tight}, it is know that \textit{quantum memory} is a resource~\cite{chen2022exponential, chen2023efficient, chen2024tight}.
More specifically, for many variations of the task, coherent processing with a quantum memory is known to provide an exponential advantage over incoherent quantum processing, either in terms of the number of copies/queries required or in terms of the number of measurements to be performed \cite{chen2022exponential, chen2023efficient, caro2023learning, chen2024optimalpauli}.
\end{example}

\begin{example}[Agnostic stabilizer learning -- Problem~\ref{prob:gnostic_stabilizer_learning}]
For $\varepsilon$ small enough, this task is at least as challenging as realizable stabilizer learning, where the unknown state $\rho$ is promised to itself be a pure stabilizer state. Thus, the lower bounds of $\Omega(n^2)$ for single-copy measurements and $\Omega(n)$ for multi-copy measurements~\cite{arunachalam2023optimal, arunachalam2026optimalstabilizertestinglearning} carry over to agnostic stabilizer state learning, and \textit{quantum memory} is again a relevant resource. 
\end{example}

\begin{example}[Realizable discrete distribution learning]\label{example:distribution_learning} 
Let $\mathcal{M}\subset \mathcal{D}$ be a subset of distributions (with some common structure) over length-$n$ bitstrings, and $(\mathcal{M},d_{\mathrm{TV}},\epsilon)$ the associated \textit{realizable} learning problem with respect to the total variation distance. Two typical oracle models for realizable distribution learning are the sample model and the statistical query model. However, we note that in both cases there do not exist generic efficient (with respect to $n$) algorithms for evaluating $d_{\mathrm{tv}}(h,c)$ -- and subsequently deciding valid solutions in the sense of Definition~\ref{def:decide-valid} -- given a description of $h$ and access to $\textsf{O}(c)$~\cite{canonne2020survey}. We therefore see that realizable distribution learning is a non-trivial setting for interactive proofs, contrary to realizable learning of Boolean functions. For concrete examples of sets $\mathcal{M}$ for which there is a query complexity separation between samples and statistical queries for realizable learning, one can consider either the parity distributions defined in Problem~\ref{prob:parity_testing}, or the output distributions of linear depth ``Clifford + one T'' quantum circuits~\cite{hinsche2023one}.
\end{example}  

\section{VBQC-based interactive proofs}\label{app:VBQC}
In this appendix, we provide a proof of Observation~\ref{lem:interactive-quantum-proofs-vbqc}, which we restate below for convenience. 

\labelRestate{lem:interactive-quantum-proofs-vbqc}
\obsgenericVBQC*

\begin{proof} To start, we note that without loss of generality we can write the (coherent multi-copy) algorithm $\mathcal{A}$, for solving the desired problem with success probability $\geq 1-\nicefrac{\delta}{2}$, in the following form: 

\begin{enumerate}
        \item Query \(m\) copies of the unknown state \(\rho\). 
        \item Prepare the state \(\psi=\rho^{\otimes m}\otimes 0^{\otimes l}\) where \(0=\ketbra{0}{0}\) denotes the all-zero state and \(l\) is some number of auxilliary qubits.
        \item Execute some unitary circuit \(\psi \mapsto \psi'=\mathcal{U}(\psi)=U \psi U^{\dagger}\). 
        \item Measure a specific subset of  $\psi'$'s qubits called output qubits in the computational basis to obtain a bit string \(b\) that represents the output of \(\mathcal{A}\).
\end{enumerate}

In particular, we note that any potential adaptivity or classical post-processing of measurement outcomes has been incorporated into the unitary $U$, possibly via the introduction of ancilla qubits. 
Then, the strategy for the interactive protocol is simple: $V$ will run the algorithm as written above, but use the universal VBQC protocol from Ref.~\cite{fitzsimonsUnconditionallyVerifiableBlind2017} to verifiably and blindly delegate the execution of the unitary circuit in Step 3, and the measurement of the output state in Step 4, to the coherent multi-copy prover. 
At first glance, this seems to contradict the quantum memory restriction of the prover.
However, 
the protocol from Ref.~\cite{fitzsimonsUnconditionallyVerifiableBlind2017} only requires the verifier to send the suitably masked qubits of $\psi$ to $P$ in a stream \textit{one at a time}. 
Despite $V$'s limited memory, it \textit{can} do this by successively querying the oracle, masking the qubits appropriately, sending them to $P$ to clear memory, and then re-querying the oracle. 
The interactive protocol works as follows:

\begin{enumerate}
\item Set security parameter $d\gets \frac{5}{2}\log_{\nicefrac{5}{6}}\left(\nicefrac{\delta}{2}\right)$.
\item $V$ uses Protocol 8 from Ref.~\cite{fitzsimonsUnconditionallyVerifiableBlind2017} to delegate Step 3 (the execution of $\mathcal{U}$ on $\psi$) to the prover, using security parameter $d$ when Protocol 8 calls Protocol 7. During execution of the protocol $V$ receives a bit string $b_\mathrm{trap}$ containing measurement outcomes from trap qubits. At the end of the protocol, $P$ sends some state $\psi_\mathrm{out}\in \mathcal{H}_{\mathrm{comp}}\otimes \mathcal{H}_{\mathrm{trap}}\otimes \mathcal{H}_{\mathrm{junk}}$ to $V$, one qubit at a time. 
Here, \(\mathcal{H}_{\mathrm{comp}}\) is the Hilbert space of \(\mathcal{A}\)'s output qubits, \(\mathcal{H}_{\mathrm{trap}}\) is the Hilbert space of remaining trap qubits, and \(\mathcal{H}_{\mathrm{junk}}\) is the Hilbert space of the remaining qubits.
\item $V$ should now decide whether to accept or reject:
\begin{enumerate}
\item If all bits of $b_\mathrm{trap}$ are correct (as per Protocol 8 in Ref.~\cite{fitzsimonsUnconditionallyVerifiableBlind2017}), then $V$ 
streams the incoming qubits of \(\psi_{\mathrm{out}}\) received from \(P\). 
The verifier discards  any qubit from \(\mathcal{H}_{\mathrm{junk}}\).
Qubits from \(\mathcal{H}_{\mathrm{comp}}\) are measured and \(V\) stores those outputs in a bit-string \(b_{\mathrm{out}}\).
Similarly, qubits from \(\mathcal{H}_{\mathrm{trap}}\) are measured and, conditioned on the output of the trap-measurement the interaction either continues or is rejected and \(V\) returns ``abort''. 
$V$ then outputs the bit string $b_\mathrm{out}$.
\item Else $V$ rejects the interaction and outputs ``abort".
\end{enumerate}
\end{enumerate}

We can now prove completeness and soundness.

\textbf{Completeness:} 
As before, we first stress that the delegation in Step 3 above is possible for $V$, as Protocol 8 from Ref.~\cite{fitzsimonsUnconditionallyVerifiableBlind2017} only requires $V$ to send qubits of $\psi$ to $P$ one at a time (after single-qubit masking). It then follows from Theorem 10 in Ref.~\cite{fitzsimonsUnconditionallyVerifiableBlind2017} that if the prover is honest, $V$ will accept the interaction, and receive in the computation register $\psi'$, with probability 1. From the assumption on $\mathcal{A}$, we therefore have that $b_\mathrm{out}$ is a valid solution with probability $\geq 1-\nicefrac{\delta}{2}$. 

\textbf{Soundness:} 
Define the ideal VBQC output state after interacting with a dishonest prover as
\begin{equation*}
\psi_{\mathrm{ideal}}(p) = p\psi'\otimes \proj{\mathrm{acc}} + (1-p)\rho\otimes \proj{\mathrm{rej}}
\end{equation*}
for some $p \in [0,1]$, and with $\rho$ any state on \(\mathcal{H}_{\mathrm{comp}}\otimes \mathcal{H}_{\mathrm{junk}}\). The states $\ket{\mathrm{acc}}$ and $\ket{\mathrm{rej}}$  in \(\mathcal{H}_{\mathrm{trap}}\) with \(\braket{\mathrm{acc}}{\mathrm{rej}}=0\) denote the flag-states that lead $V$ to accept or reject the interaction after measuring the respective labels. 

Recall, that algorithm $\mathcal{A}$ fails with probability $\leq\nicefrac{\delta}{2}$ when measuring $\psi'$ in the computational basis. 
Thus, the probability of \(V\) accepting the interaction, yet outputting an invalid $b_\mathrm{out}$ 
when measuring $\psi_{\mathrm{ideal}}(p)$ in the computational basis is $\leq \nicefrac{\delta}{2}$. 
From Theorem 12 in Ref.~\cite{fitzsimonsUnconditionallyVerifiableBlind2017}, together with the choice of security parameter $d$, we know that 
\begin{equation*}
    \frac{1}{2}\norm{\psi_\mathrm{out} - \psi_{\mathrm{ideal}}(p)}_1
    \leq \left(\frac{5}{6}\right)^{\left\lceil \frac{2d}{5}\right\rceil}
    \leq \frac{\delta}{2}\, ,
\end{equation*}
for some $p$.
Therefore, the probability of \(V\) with access to \(\psi_{\mathrm{out}}\) for accepting the interaction and outputting an invalid string \(b_\mathrm{out}\) is at most \(\frac{\delta}{2}\) larger than the same probability with \(\psi_{\mathrm{out}}\) replaced by \(\psi_{\mathrm{ideal}}(p)\).
Hence, the probability of \(V\) accepting the interaction and outputting an invalid string \(b_\mathrm{out}\) is at most \(\delta\), which completes the soundness analysis.

\textbf{Communication complexity:} The claimed communication complexity follows from properties of Protocol 8 from Ref.~\cite{fitzsimonsUnconditionallyVerifiableBlind2017}. We note that the protocol could be modified in such a way that $P$ only sends classical bits to $V$ (i.e. $P$ also performs the desired measurements) without affecting the analysis.
\end{proof}

\section{An interactive proof for memory-constrained uniformity testing}\label{app:uniformity_ip}

We show here how one can use known streaming interactive proofs~\cite{cormode2011verifyingcomputationsstreaminginteractive} to construct an interactive proof protocol for uniformity testing, which allows a memory-constrained verifier to surpass query complexity lower bounds resulting from memory limitations. To be more specific, recall from Problem~\ref{prob:unitarity_testing} that uniformity testing is the problem which asks one to decide whether an unknown distribution is the uniform distribution or $\varepsilon$-far from the uniform distribution with respect to trace distance~\cite{canonne2022topics}. As discussed in Example~\ref{example:uniformity_testing}, it is known that for distributions with support size $k$, when there are no memory constraints, $\Theta\left(\nicefrac{\sqrt{k}}{\varepsilon^2}\right)$ queries are necessary and sufficient~\cite{paninski2008coincidence,canonne2022topics}. On the other hand, when only $m$ bits of memory are available, then $\Omega\left(\nicefrac{k}{m\varepsilon^2}\right)$ queries are necessary~\cite{diakonikolas2019communication,canonne2023simpler,berg2022memory}. As such, its natural to wonder whether there exists an interactive proof system which allows a memory-constrained verifier to gain a query complexity advantage via delegation to a memory unconstrained prover. The following theorem answers this question affirmatively.

\begin{theorem} For any $\epsilon\geq \nicefrac{12}{k^{1/4}}$, there exists an interactive proof system $(V,P)$ for uniformity testing, with success probability $\nicefrac{2}{3}$, between a verifier with $\mathcal{O}(\mathrm{polylog}(k))$ classical bits of memory, and a memory unconstrained prover. More specifically, for this protocol: 
\begin{enumerate}
\item The verifier has query complexity $\mathcal{O}(\nicefrac{\sqrt{k}}{\epsilon^2})$.
\item The prover and the verifier exchange $\mathcal{O}\left(\sqrt{k}\,\mathrm{polylog}(k)\right)$ bits during the protocol.
\end{enumerate}
\end{theorem}

Note that for a memory of size $m=\mathrm{polylog}(k)$, a memory-constrained verifier in isolation requires $\Omega\left(\nicefrac{k}{m\varepsilon^2}\right) = \Omega\left(\nicefrac{k}{\mathrm{polylog}(k) \varepsilon^2}\right)$ samples. This is (in the regime of large $k$) strictly larger than the upper bound of $\mathcal{O}(\nicefrac{\sqrt{k}}{\epsilon^2})$ achieved by the verifier in the above interactive proof for uniformity testing, and thus demonstrates that the interaction with the prover allows the verifier to overcome the lower bound arising from their memory constraint.

\begin{proof} 
    The high level idea of the interactive proof protocol is for the verifier to run the ``unique elements uniformity tester" (Algorithm 2 in Ref.~\cite{canonne2022topics}), by using the streaming interactive proof protocol for unique elements from Theorem 6 in Ref.~\cite{cormode2011verifyingcomputationsstreaminginteractive} to delegate the computation of the number of unique elements in the stream of samples, without having to store more than one sample at a time. 
    To make this more precise, let us define the support of the unknown distribution as $[k] = \{1,\ldots,k\}$. For any set of $n$ samples $\{x_1,\ldots,x_n\}\in[k]^n$ we define the associated frequency vector $\mathbf{a} = (a_1,\ldots,a_k)$, where for any $j\in [k]$ one defines $a_j = |\{i\,|\, x_i = j\}|$ as the number of times that $j\in[k]$ appears in the set of samples. Next, we define the function $h:\mathbb{N}_0\rightarrow\mathbb{N}_0$ via $h(1) = 1$ and $h(j) = 0$ for $j\neq 1$. This allows us to define the number of unique elements in the set of samples as $Z = \sum_{i\in [k]}h(a_i)$. With this in hand, the interactive proof protocol is as follows:
    
    \begin{enumerate}
    \item Get $\varepsilon \geq \nicefrac{12}{k^{1/4}}$.
    \item Set $n\gets \ceil{140\nicefrac{\sqrt{k}}{\epsilon^2}}$. 
    \item Set $\tau \gets \left(1-\nicefrac{1}{k}\right)^{n-1} - \nicefrac{n\epsilon^2}{8k}$. 
    \item For $j\in\{1,\ldots,n\}$: 
    \begin{enumerate}
    \item $V$ queries the oracle and receives a sample $x_j\in[k]$.
    \item $V$ sends $x_j$ to $P$, and follows the protocol for calculating ``frequency based functions", from Section 6.2 in Ref.~\cite{cormode2011verifyingcomputationsstreaminginteractive}, adapted to calculate the number of unique elements in the sample stream, and with soundness parameter $1/6$.
    \end{enumerate}
    \item $V$ receives $\tilde{Z}$ from $P$ -- the claimed number of unique elements. Still following the protocol from Section 6.2 in Ref.~\cite{cormode2011verifyingcomputationsstreaminginteractive}, $V$ decides whether to accept or reject the interaction with $P$.
    \begin{enumerate}
    \item If $V$ rejects the interaction with $P$, then the protocol is aborted.
    \item Else, if $V$ accepts the protocol then:
    \begin{enumerate}
    \item If $\tilde{Z}\leq \tau$ output ``reject'' (i.e. not uniform).
    \item Else, output ``accept" (i.e. uniform).
    \end{enumerate}
    \end{enumerate}
    \end{enumerate}
    
    \textbf{Completeness:} We note, from Theorem 6 in Ref.~\cite{cormode2011verifyingcomputationsstreaminginteractive}, together with the insistence on perfect completeness for valid protocols in Definition 1 and Definition 2 of Ref.~\cite{cormode2011verifyingcomputationsstreaminginteractive}, that for the honest prover $P$ the verifier $V$ will always accept the interaction, and it will always be the case that $\tilde{Z} = Z$. In this case, we note that the above interactive protocol implements the ``unique elements uniformity tester" (Algorithm 2 in Ref.~\cite{canonne2020survey}). From Theorem 2.4 and its proof in Ref.~\cite{canonne2020survey}, one can also verify that, with $n$ set as above, the algorithm will output a correct solution with probability greater than $\nicefrac{5}{6}$, provided $\varepsilon \geq \nicefrac{12}{k^{1/4}}$.
    
    \textbf{Soundness:} Let us denote by $E_1$ the event that $V$ accepts the interaction with $P$ and $\tilde{Z} = Z$, but $V$ still outputs an invalid solution. As $\tilde{Z} = Z$, verifier $V$ correctly implements the unique elements uniformity tester, and therefore outputs an invalid solution with probability less than $\nicefrac{1}{6}$ (via Theorem 2.4 in Ref.~\cite{canonne2022topics}). As such, the probability of event $E_1$ is less than $\nicefrac{1}{6}$. 
    
    Now, let us denote by $E_2$ the event that $P$ is dishonest and outputs some $\tilde{Z}$ which causes $V$ to output an invalid solution, but $V$ nevertheless accepts the interaction with $P$. By the soundness guarantee of the protocol for frequency based functions from Ref.~\cite{cormode2011verifyingcomputationsstreaminginteractive} (Section 6.2), the probability of event $E_2$ is less than $\nicefrac{1}{6}$. 
    The probability that $V$ accepts the interaction and output an incorrect solution is at most the probability of $E_1 \cup E_2$. By a union bound, this probability is less than $\nicefrac{1}{3}$.
    
    \textbf{Complexity analysis:} From the description of the protocol above, it is clear that $V$ has query complexity $\mathcal{O}\left(\nicefrac{\sqrt{k}}{\varepsilon^2}\right)$. From a space perspective, $V$ requires the space to store $\epsilon$, $n$ and $\tau$, keep track of the round counter $j$, receive single samples from the distribution, and execute the delegation protocol for frequency based functions from Ref.~\cite{cormode2011verifyingcomputationsstreaminginteractive}. For the relevant parameter regimes, this space cost is dominated by the requirements of the protocol from Ref.~\cite{cormode2011verifyingcomputationsstreaminginteractive}. More specifically, as per Theorem~6 (and Corollary 2) in Ref.~\cite{cormode2011verifyingcomputationsstreaminginteractive}, we have that $V$ uses $\mathcal{O}(\log(k))$ \textit{words} of space to execute this protocol, where as defined in Ref.~\cite{cormode2011verifyingcomputationsstreaminginteractive}, a \textit{word} is some sequence of $\mathrm{polylog}(k)$ bits. As such, $V$ requires $\mathcal{O}\left(\mathrm{polylog}(k)\right)$ \textit{bits} of space. Finally, the communication complexity of the protocol is determined solely by the communication required to execute the protocol from Ref.~\cite{cormode2011verifyingcomputationsstreaminginteractive}. Again from Theorem 6 in Ref.~\cite{cormode2011verifyingcomputationsstreaminginteractive}, we have that  $\mathcal{O}(\sqrt{k}\log(k))$ words, or $\mathcal{O}\left(\sqrt{k}\,\mathrm{polylog}(k)\right)$ bits,  are exchanged between $P$ and $V$ during execution of the protocol.
\end{proof}

{
\setlength{\emergencystretch}{4em}
\printbibliography
\addcontentsline{toc}{section}
{\protect\textbf{References}}
}








\end{document}